\documentclass[12pt]{article}
\usepackage{amsfonts,amsthm}
\usepackage{lscape}
\usepackage{amsmath,amssymb,mathrsfs}
\usepackage{enumitem}
\usepackage[longnamesfirst]{natbib}
\usepackage{color}
\usepackage{caption}
\usepackage{booktabs}
\usepackage{xr}
\usepackage{setspace}
\usepackage[top=1in,bottom=1in,left=1in,right=1in]{geometry}
\usepackage{appendix}
\newtheorem{theorem}{Theorem}

\newtheorem{corollary}{Corollary}

\newtheorem{definition}{Definition}

\newtheorem{lemma}{Lemma}

\newtheorem{assumption}{Assumption}

\newtheorem{proposition}{Proposition}
\newtheorem{remark}{Remark}

\DeclareMathOperator{\plim}{plim}

\newcommand{\fancyt}{\mathscr{T}_n}
\newcommand{\fancyv}{\mathscr{V}}
\numberwithin{equation}{section}
\numberwithin{theorem}{section}
\numberwithin{lemma}{section}
\numberwithin{proposition}{section}
\numberwithin{corollary}{section}
\begin{document}

	\title{Consistent specification testing under spatial dependence\thanks{We thank the editor, co-editor and three referees for insightful comments that improved the paper. We are grateful to Swati Chandna, Miguel Delgado, Emmanuel Guerre, Fernando L\'opez Hernand\'ez, Hon Ho Kwok, Arthur Lewbel, Daisuke Murakami, Ryo Okui and Amol Sasane  for helpful comments, and audiences at YEAP 2018 (Shanghai University of Finance and Economics), NYU Shanghai, Carlos III Madrid, SEW 2018 (Dijon), Aarhus University, SEA 2018 (Vienna), EcoSta 2018 (Hong Kong), Hong Kong University, AFES 2018 (Cotonou), ESEM 2018 (Cologne), CFE 2018 (Pisa), University of York, Penn State, Michigan State, University of Michigan, Texas A\&M, 1st Southampton Workshop on Econometrics and Statistics and MEG 2019 (Columbus). We also thank Xifeng Wen from the Experiment and Data Center of Antai College of
			Economics and Management (SJTU) for expert computing assistance.}}
	
	\date{\today }
	\author{Abhimanyu Gupta\thanks{Department of Economics, University of Essex, Wivenhoe Park, Colchester CO4 3SQ, UK. E-mail: a.gupta@essex.ac.uk.} \thanks{Research supported by ESRC grant ES/R006032/1.} \and  Xi Qu\thanks{Antai  College  of  Economics  and  Management,  Shanghai  Jiao Tong  University, Shanghai, China, 200052. E-mail: xiqu@sjtu.edu.cn.} \thanks{Research supported by the National Natural Science Foundation of China, Project Nos. 72222007, 71973097 and 72031006.}}

	\maketitle
	
	\begin{abstract}
		We propose a series-based nonparametric specification test for a regression function when data are spatially dependent, the `space' being of a general economic or social nature. Dependence can be parametric, parametric
		with increasing dimension, semiparametric or any combination thereof, thus covering a vast variety of settings. These include spatial error models of varying types and levels of complexity.  Under a new smooth spatial dependence condition, our test statistic is asymptotically standard normal. To prove the latter property, we establish a central limit theorem for quadratic forms in linear processes in an increasing dimension setting. Finite sample performance is investigated in a simulation study, with a bootstrap method also justified and illustrated. Empirical examples illustrate the test with real-world data.
	\end{abstract}
	
	

	%
	\clearpage
	\allowdisplaybreaks
	\section{Introduction}
	Models for spatial dependence have recently become the subject of vigorous research. This burgeoning interest has roots in the  needs of practitioners who frequently have access to data sets featuring inter-connected cross-sectional units. Motivated by these practical concerns, we propose a specification test for a regression function in a general setup that covers a vast variety of commonly employed spatial dependence models and permits the complexity of dependence to increase with sample size. Our test is consistent, in the sense that a parametric specification is tested with asymptotically unit power against a nonparametric alternative. The `spatial' models
	that we study are not restricted in any way to be geographic in nature, indeed `space' can be a very general economic or social space. Our empirical examples feature conflict alliances and technology externalities as examples of `spatial dependence', for instance.
	
	Specification testing is an important problem, and this is reflected in a huge literature studying consistent tests. Much of this is based on independent, and often also identically distributed, data. However data frequently exhibit dependence and consequently a branch of the literature has also examined specification tests under time series dependence.
	Our interest centers on dependence across a `space', which differs quite fundamentally from dependence in a time series context. Time series are naturally ordered and locations of the observations can be observed, or at least the process generating these locations may be modelled. It can be imagined that concepts from time series dependence be extended to settings where the data are observed on a geographic space and dependence can be treated as a decreasing function of distance between observations. Indeed much work has been done to extend notions of time series dependence in this type of setting, see e.g. \cite{Jenish2009,Jenish2012}.
	
	However, in a huge variety of economics and social science applications agents influence each other in ways that do not conform to such a setting. For example, farmers affect the demand of farmers in the same village but not in different villages, as in \cite{case1991spatial}.  Likewise, price competition among firms exhibits spatial features (\cite{pinkse2002}), input-output relations lead to complementarities between sectors (\cite{Conley2003}), co-author connections form among scientists (\cite{Oettl2012}, \cite{Mohnen2017}),  R\&D spillovers occur through technology and product market spaces (\cite{Bloom2013}), networks form due to allegiances in conflicts (\cite{konig2017networks}) and overlapping bank portfolios lead to correlated lending decisions (\cite{Gupta2018a}). Such examples cannot be studied by simply extending results developed for time series and illustrate the growing need for suitable methods.

	A very popular  model for general spatial dependence is the spatial autoregressive (SAR) class, due to \cite{cliff1973spatial}. The key feature of SAR models, and various generalizations such as SARMA (SAR moving average) and matrix exponential spatial specifications (MESS, due to \cite{LeSage2007}), is the presence of one or more spatial weight matrices
	whose elements characterize the  links between agents. As noted above, these links may form for a variety of reasons, so the `spatial' terminology represents a very general notion of space, such as social or economic space.  Key papers on the estimation of SAR models and their variants include \cite{kelejian1998generalized} and \cite{lee2004asymptotic}, but research on various aspects of these is active and ongoing, see e.g.  \cite{Robinson2015, Hillier2018, Hillier2018a, Kuersteiner2020, Han2021, Hahn2021}.

	Unlike work focusing on independent or time series data, a general drawback of spatially oriented research has been the lack of general unified theory. Typically, individual papers have studied specific special cases of various spatial specifications. A strand of the literature has introduced the notion of a cross-sectional linear-process to help address this problem, and we follow this approach. This representation can accommodate SAR models in the error term (so called spatial error models (SEM)) as a special case, as well as variants like SARMA and MESS, whence its generality is apparent. The linear-process structure shares some similarities with that familiar from the time series literature (see e.g. \cite{hannan1970multiple}). Indeed, time series versions may be regarded as very special cases but, as stressed before, the features of spatial dependence must be taken into account in the general formulation. Such a representation was introduced by \cite{Robinson2011} and further examined in other situations by \cite{Robinson2012b} (partially linear regression), \cite{Delgado2015} (non-nested correlation testing), \cite{Lee2016} (series estimation of nonparametric regression)  and \cite{Hidalgo2017} (cross-sectionally dependent panels).
	
	
	In this paper, we propose a test statistic similar to that of \cite{Hong1995},  based on estimating the nonparametric specification via series approximations. Assuming an independent and identically distributed sample, their statistic is based on the
	sample covariance between the residual from the parametric model and the
	discrepancy between the parametric and nonparametric fitted values. Allowing additionally for spatial dependence through the form of a linear process as discussed above, our statistic is shown to be asymptotically standard normal, consistent and possessing nontrivial power against local alternatives of a certain type. To prove asymptotic normality, we present a new central limit theorem (CLT) for quadratic forms in linear processes in an increasing dimension setting that may be of independent interest. A CLT for quadratic forms under time series dependence in the context of series estimation can be found in \cite{Gao2000}, and our result can be viewed as complementary to this. The setting of \cite{Su2017} is a very special case of our framework. There has been recent interest in specification testing for spatial models, see for example \cite{Sun2020} for a kernel-based model specification test and \cite{Lee2020} for a consistent omnibus test. We contribute to this literature by studying a linear process based increasing parameter dimension framework.
	
	Our linear process framework permits spatial dependence to be parametric, parametric
	with increasing dimension, semiparametric or any combination thereof, thus covering a vast variety of settings. A class of models of great empirical interest are `higher-order' SAR models in the outcome variables, but with spatial dependence structure also in the errors. We initially present  the familiar nonparametric regression to clarify the exposition, and then cover this class as the main model of interest. Our theory covers as special cases SAR, SMA, SARMA, MESS models for the error term. These specifications may be of any fixed spatial order, but our theory also covers the case where they are of increasing order.
	
	Thus we permit a more complex model of spatial dependence as more data become available, which encourages a more flexible approach to modelling such dependence as stressed by \cite{Gupta2013,Gupta2018} in a higher-order SAR context,  \cite{huber1973robust}, \cite{portnoy1984asymptotic,portnoy1985asymptotic} and \cite{Anatolyev2012} in a regression context and \cite{Koenker1999} for the generalized method of moments setting, amongst others. This literature focuses on a sequence of true models, rather than a sequence of models approximating an infinite true model. Our paper also takes the same approach. On the other hand, in the spatial setting, \cite{Gupta2018d} considers increasing lag models as approximations to an infinite lag model with lattice data and also suggests criteria for choice of lag length.
	
	Our framework is also extended to the situation where spatial dependence occurs through nonparametric functions of raw distances (these may be exogenous economic or social distances, say), as in \cite{pinkse2002}. This allows for greater flexibility in modelling spatial weights as the practitioner only has to choose an exogenous economic distance measure and allow the data to determine the functional form. It also adds a degree of robustness to the theory by avoiding potential parametric misspecification. The case of geographical  data is also covered, for example the important classes of Mat\'ern and Wendland (see e.g. \cite{Gneiting2002}) covariance functions. Finally, we introduce a new notion of smooth spatial dependence that provides more primitive, and checkable, conditions for certain properties than extant ones in the literature.
	
	To illustrate the performance of the test in finite samples, we present Monte Carlo simulations that exhibit satisfactory small sample properties. The test is demonstrated in three empirical examples, including two based on recently published work on social networks: \cite{Bloom2013} (R\&D spillovers in innovation), \cite{konig2017networks} (conflict alliances during the Congolese civil war). Another example studies cross-country spillovers in economic growth. Our test may or may not reject the null hypothesis of a linear regression in these examples, illustrating its ability to distinguish well between the null and alternative models.
	
	The next section introduces our basic setup using a nonparametric regression with no SAR structure in responses. We treat this abstraction as a base case, and Section \ref{sec:test} discusses  estimation and defines the test statistic, while Section \ref{sec:asymptotics} introduces assumptions and the key asymptotic results of the paper. Section \ref{sec:SAR_ext} examines the most commonly employed higher-order SAR models, while Section \ref{sec:nonpar_ext} deals with nonparametric spatial error structures. Nonparametric specification tests are often criticized for poor finite sample performance when using the asymptotic critical values. In Section \ref{sec:bootstrap} we present a bootstrap version of our testing procedure. Sections \ref{sec:mc} and \ref{sec:apps} contain a study of finite sample performance and the empirical examples respectively, while Section \ref{sec:conc} concludes. Proofs are contained in appendices, including a supplementary online appendix which also contains additional simulation results.
	
	For the convenience of the reader, we collect some frequently used notation here. First, we introduce three notational conventions for any parameter $\nu$ for the rest of the paper: $\nu\in\mathbb{R}^{d_\nu}$, $\nu_0$ denotes the true value of $\nu$ and for any scalar, vector or matrix valued function $f(\nu)$, we denote $f\equiv f(\nu_0)$.  Let $\overline{\varphi}(\cdot)$ (respectively $\underline{\varphi}(\cdot)$) denote
	the largest (respectively smallest) eigenvalue of a generic square
	nonnegative definite matrix argument. For a generic matrix $A$, denote  $\left\Vert A\right\Vert=\left[\overline{\varphi}(A'A)\right]^{1/2}$, i.e. the spectral norm of $A$ which reduces to the Euclidean norm if $A$ is a vector. $\left\Vert A\right\Vert_R$ denotes the maximum absolute row sum norm of a generic matrix $A$ while $\left \Vert A \right \Vert _{F}=\left[tr(AA')\right]^{1/2}$, the Frobenius norm. Throughout the paper $|\cdot|$ is absolute value when applied to a scalar and determinant when applied to a matrix. Denote by $c$ ($C$) generic positive constants, independent of any quantities that tend to infinity, and arbitrarily small (big).
	
	\section{Setup}\label{sec:setup}
	To illustrate our approach, we first consider the nonparametric regression
	\begin{equation}
		y_{i}=\theta _{0}\left( x_{i}\right) +u_{i},i=1,\ldots ,n,  \label{model}
	\end{equation}%
	where $\theta _{0}(\cdot )$ is an unknown function and $x_{i}$ is a vector of strictly exogenous explanatory variables with support $\mathcal{X}\subset \mathbb{R}^k$.  Spatial dependence is explicitly modeled via the error term $u_{i}$, which we assume is generated by:
	\begin{equation}
		u_{i}=\sum_{s=1}^{\infty }b_{is}\varepsilon _{s},  \label{linpro_gamma}
	\end{equation}
	where $\varepsilon _{s}$ are independent random variables, with zero mean and identical variance $\sigma_0^2$. Further conditions on the $\varepsilon_s$ will be assumed later. The linear process coefficients $b_{is}$ can depend on $n$, as may the covariates $x_i$. This is generally the case with spatial models and implies that asymptotic theory ought to be developed for triangular arrays. There are a number of reasons to permit dependence on sample size. The $b_{is}$ can depend on spatial weight matrices, which are usually normalized for both stability and identification purposes.
	
	Such normalizations, e.g. row-standardization or division by spectral norm, may be $n$-dependent. Furthermore, $x_i$ often includes underlying covariates of `neighbors' defined by spatial weight matrices. For instance, for some $n\times 1$ covariate vector $z$ and exogenous spatial weight matrix $W\equiv W_n$, a component of $x_i$ can be $e_i'Wz$, where $e_i$ has unity in the $i$-th position and zeros elsewhere, which depends on $n$. Thus, subsequently, any spatial weight matrices will also be allowed to depend on $n$. Finally, treating triangular arrays permits re-labelling of quantities that is often required when dealing with spatial data, due to the lack of natural ordering, see e.g. \cite{Robinson2011}. We suppress explicit reference to this $n$-dependence of various quantities for brevity, although mention will be made of this at times to remind the reader of this feature.
	
	Now, assume the
	existence of a $d_\gamma\times 1$ vector $\gamma _{0}$ such that $%
	b_{is}=b_{is}(\gamma _{0})$, possibly with $d_\gamma\rightarrow \infty $ as $n\rightarrow
	\infty $, for all $i=1,\ldots ,n$ and $s\geq 1$. Let $u$ be the $n\times 1$
	vector with typical element $u_{i}$, $\varepsilon $ be the infinite
	dimensional vector with typical element $\varepsilon _{s},$ and $B$ be an
	infinite dimensional matrix \citep{Cooke1950} with typical element $b_{is}.$ In matrix form,
	\begin{equation}\label{matrix_form}
		u=B\varepsilon \text{ and }\mathcal{E}\left( uu^{\prime }\right) =\sigma_0
		^{2}BB^{\prime }=\sigma_0^{2}\Sigma \equiv\sigma_0^{2}\Sigma \left( \gamma
		_{0}\right) .
	\end{equation}
	We assume that $\gamma _{0}\in \Gamma $, where $\Gamma $ is a compact subset
	of $\mathbb{R}^{d_\gamma}$.
	With $d_\gamma$ diverging, ensuring $\Gamma $ has bounded volume requires some
	care, see \cite{Gupta2018}.
	For a known function $f(\cdot)$, our aim is to test
	\begin{equation}\label{null}
		H_{0}:P[\theta _{0}\left( x_{i}\right)=f(x_{i},\alpha _{0}) ]=1,\text{ for some }
		\alpha _{0}\in \mathcal{A}\subset\mathbb{R}^{d_\alpha},
	\end{equation}
	against the global alternative $H_{1}:P\left[\theta _{0}\left( x_{i}\right)\neq f(x_{i},\alpha )  \right]>0,\text{ for
		all }\alpha \in \mathcal{A}$.
	
	We now nest commonly used models for spatial dependence in (\ref{matrix_form}). Introduce a set of $n\times n$ spatial weight (equivalently network adjacency) matrices $W_j$, $j=1,\ldots,m_1+m_2$. Each $W_j$ can be thought of as representing dependence through a particular space. Now, consider models of the form $\Sigma(\gamma)=A^{-1}(\gamma)A'^{-1}(\gamma)$. For example, with $\xi$ denoting a vector of iid disturbances with variance $\sigma_0^2$, the model with SARMA$(m_1,m_2)$ errors is $
	u=\sum_{j=1}^{m_1}\gamma_jW_ju+\sum_{j=m_1+1}^{m_1+m_2}\gamma_jW_j\xi+\xi$, with $A(\gamma)=\left(I_n+\sum_{j=m_1+1}^{m_1+m_2}\gamma_jW_j\right)^{-1}\left(I_n-\sum_{j=1}^{m_1}\gamma_jW_j\right)$, assuming conditions that guarantee the existence of the inverse. Such conditions can be found in the literature, see e.g. \cite{Lee2010} and \cite{Gupta2018}.
	The SEM model is obtained by setting $m_2=0$ while the model with SMA errors has $m_1=0$. The model with MESS$(m)$ errors (\cite{LeSage2007}, \cite{Debarsy2015}) is $
	u=\exp\left(\sum_{j=1}^{m} \gamma_jW_j\right)\xi,A(\gamma)=\exp\left(-\sum_{j=1}^{m} \gamma_jW_j\right).
	$
	
	In some cases the space under consideration is geographic i.e. the data may be observed at irregular points in Euclidean space. Making the identification $u_i\equiv U\left(t_i\right)$, $t_i\in\mathbb{R}^d$ for some $d>1$, and assuming covariance stationarity, $U(t)$ is said to follow an isotropic model if, for some function $\delta$ on $\mathbb{R}$, the covariance at lag $s$ is $r(s)=\mathcal{E}\left[U(t)U(t+s)\right]=\delta(\Vert s\Vert)$. An important class of parametric isotropic models is that of \cite{Matern1986}, which can be parameterized in several ways, see e.g. \cite{Stein1999}. Denoting by $\Gamma_f$ the Gamma function and by $\mathcal{K}_{\gamma_1}$ the modified Bessel function of the second kind (\cite{Gradshteyn1994}), take
	$
	\delta(\left\Vert s\right\Vert,\gamma)=\left(2^{\gamma_1-1}\Gamma_f(\gamma_1)\right)^{-1}\left(\gamma_2^{-1}\sqrt{2\gamma_1}\left\Vert s\right\Vert\right)^{\gamma_1}\mathcal{K}_{\gamma_1}\left(\gamma_2^{-1}\sqrt{2\gamma_1}\left\Vert s\right\Vert\right),
	$
	with $\gamma_1,\gamma_2>0$ and $d_\gamma=2$. With $d_\gamma=3$, another model takes $\delta(\left\Vert s\right\Vert,\gamma)=\gamma_1\exp\left(-\left\Vert s/\gamma_2\right\Vert^{\gamma_3}\right)$, see e.g. \cite{DeOliveira1997}, \cite{Stein1999}. \cite{Fuentes2007} considers this model with $\gamma_3=1$, as well as a specific parameterization of the Mat\`{e}rn covariance function.
	\section{Test statistic}\label{sec:test}
	
	We estimate $\theta _{0}(\cdot )$ via a series approximation. Certain technical conditions are needed to allow for $\mathcal{X}$ to have unbounded support. To this end, for a function $g(x)$ on $\mathcal{X}$, define a weighted sup-norm (see e.g. \cite{Chen2005}, \cite{Chen2007}, \cite{Lee2016}) by $\left\Vert g\right\Vert_{w}=\sup_{x\in\mathcal{X}}\left\vert g(x)\right\vert\left(1+\left\Vert x\right\Vert^2\right)^{-w/2}, \text{ for some } w>0$. Assume that
	there exists a sequence of functions $\psi _{i}:=\psi \left( x_{i}\right) :%
	\mathbb{R}^{k}\mapsto \mathbb{R}^{p}$, where $p\rightarrow \infty $ as $%
	n\rightarrow \infty $, and a $p\times 1$ vector of coefficients $\beta_0 $
	such that
	\begin{equation}
		\theta _{0}\left( x_{i}\right) =\psi _{i}^{\prime }\beta_0 +e\left(
		x_{i}\right),  \label{theta_series_appr}
	\end{equation}%
	where $e(\cdot)$ satisfies:
	\setcounter{assumption}{0}
	\renewcommand\theassumption{R.\arabic{assumption}}
	\begin{assumption}\label{ass:approx_error}
		There exists a constant $\mu>0$ such that $\left \Vert e \right \Vert_{w_x} =O\left( p^{-\mu}\right),$ as $p\rightarrow\infty$, where $w_x\geq 0$ is the largest value such that $\sup_{i=1,\ldots,n} \mathcal{E}\left\Vert x_i\right\Vert^{w_x}<\infty$, for all $n$.
	\end{assumption}
	By Lemma 1 in Appendix B of \cite{Lee2016}, this assumption implies that
	\begin{equation}
		\sup_{i=1,\ldots,n}\mathcal{E}\left ( e^2\left( x_i\right) \right ) =O\left( p^{-2\mu
		}\right).  \label{appr_error}
	\end{equation}
	Due to the large number of assumptions in the paper, sometimes with changes reflecting only the various setups we consider, we prefix assumptions with R in this section and the next, to signify `regression'. In Section \ref{sec:SAR_ext} the prefix is SAR, for `spatial autoregression', while in Section \ref{sec:nonpar_ext} we use NPN, for `nonparametric network'.
	
	Let $%
	y=(y_{1},\ldots,y_{n})^{\prime },{\theta _{0}}=(\theta _{0}\left(
	x_{1}\right) ,\ldots,\theta _{0}\left( x_{n}\right) )^{\prime }, \Psi =(\psi _{1},\ldots,\psi _{n})^{\prime }$. We will estimate $\gamma_0$ using a quasi maximum likelihood estimator (QMLE) based on a Gaussian likelihood, although Gaussianity is nowhere assumed. For
	any admissible values $\beta $, $\sigma ^{2}$ and $\gamma $, the (multiplied
	by $2/n$) negative quasi log likelihood function based on using the
	approximation (\ref{theta_series_appr}) is
	\begin{equation}
		{L}(\beta,\sigma^2 ,\gamma )=\ln \left( 2\pi \sigma ^{2}\right) +\frac{1}{n}%
		\ln \left \vert \Sigma \left( \gamma \right) \right \vert +\frac{1}{n\sigma
			^{2}}(y-\Psi \beta )^{\prime }\Sigma \left( \gamma \right) ^{-1}(y-\Psi
		\beta ),  \label{likelihood}
	\end{equation}%
	which is minimised with respect to $\beta $ and $\sigma ^{2}$ by
	\begin{eqnarray}
		\bar{\beta}\left( \gamma \right) &=&\left( \Psi ^{\prime }\Sigma \left(
		\gamma \right) ^{-1}\Psi \right) ^{-1}\Psi ^{\prime }\Sigma \left( \gamma
		\right) ^{-1}y,  \label{betapmle} \\
		\bar{\sigma}^{2}\left( \gamma \right) &=&{n^{-1}}y^{\prime }E(\gamma
		)^{\prime }M(\gamma )E(\gamma )y,  \label{sigmapmleconc}
	\end{eqnarray}%
	where $M(\gamma )=I_{n}-E(\gamma )\Psi \left( \Psi ^{\prime }\Sigma (\gamma
	)^{-1}\Psi \right) ^{-1}\Psi ^{\prime }E(\gamma )^{\prime }$ and $E(\gamma )$
	is the $n\times n$ symmetric matrix such that $E(\gamma )E(\gamma )^{\prime }=\Sigma
	(\gamma )^{-1}$. The use of the approximate likelihood relies on the negligibility of $e(\cdot)$, which in turn permits the replacement of $\theta_0(\cdot)$ by $\psi'\beta_0$ with asymptotically negligible cost. Thus the concentrated likelihood function is
	\begin{equation}
		\mathcal{L}(\gamma )=\ln (2\pi )+\ln \bar{\sigma}^{2}(\gamma )+\frac{1}{n}%
		\ln \left \vert \Sigma \left( \gamma \right) \right \vert.
		\label{conc_likelihood}
	\end{equation}%
	We define the QMLE of $\gamma _{0}$ as $\widehat{\gamma}=\text{arg min}_{\gamma
		\in \Gamma }\mathcal{L}(\gamma )$ and the QMLEs of $\beta_0 $ and $\sigma_0 ^{2}$
	as $\widehat{\beta}=\bar{\beta}\left( \widehat{\gamma}\right) $ and $\widehat{\sigma}%
	^{2}=\bar{\sigma}^{2}\left( \widehat{\gamma}\right) $. At a given $x_1,\ldots,x_n$, the
	series estimate of $\theta _{0} $ is defined as
	\begin{equation}
		\widehat{\theta} =\left(\hat{\theta}(x_1),\ldots,\hat{\theta}(x_n)\right)'=\left(\psi (x_1)^{\prime }\widehat{\beta},\ldots,\psi (x_n)^{\prime }\widehat{\beta}\right)'.
		\label{theta_estimate}
	\end{equation}
	Let $\widehat{\alpha }_{n}\equiv \widehat{\alpha }$ denote an estimator consistent for $\alpha _{0}$
	under $H_{0}$, for example the (nonlinear) least squares estimator. Note that $\widehat\alpha$ is consistent only under $H_0$, so we introduce a general probability limit of $\widehat{\alpha } $, as in \cite{Hong1995}.
	\begin{assumption}\label{ass:alpha_order}
		There exists a deterministic sequence $\alpha_n^*\equiv\alpha^*$ such that $\widehat\alpha-\alpha^*=O_p\left(1/\sqrt{n}\right)$.
	\end{assumption}
	\noindent Examples of estimators that satisfy this assumption include
	(nonlinear) least squares, generalized method of moments estimators or adaptive efficient
	weighted least squares \citep{Stinchcombe1998}.
	
	Following \cite{Hong1995}, define the
	regression error $u_{i}\equiv y_{i}-f(x_{i},\alpha ^{\ast })$ and the
	specification error $v_{i}\equiv \theta _{0}(x_{i})-f(x_{i},\alpha ^{\ast })$. Our test statistic is based on a scaled and centered version of
	$\widehat{m}_{n}=\widehat{\sigma }^{-2}\widehat{{v}}^{\prime }\Sigma
	\left( \widehat{\gamma }\right) ^{-1}\widehat{ {u}}/n=\widehat{\sigma }%
	^{-2}\left(\widehat{ {\theta }}- {f}\left(x,\widehat{\alpha }\right)\right)^{\prime }\Sigma \left( \widehat{\gamma }\right) ^{-1}\left(y- {f}\left(x,%
	\widehat{\alpha }\right)\right)/n$, where $f(x,\alpha)=\left(f\left(x_1,\alpha\right),\ldots,f\left(x_n,\alpha\right)\right)'$. Precisely, it is defined as
	\begin{equation}\label{statistic}
		\mathscr{T}_n=\frac{n\widehat{m}_{n}-p}{\sqrt{2p}}.
	\end{equation}
	The motivation for such a centering and scaling stems from the fact that, for fixed $p$, $n\widehat{m}_{n}$ has an asymptotic $\chi^2_p$ distribution. Such a distribution has mean $p$ and variance $2p$, and it is a well-known fact that $\left(\chi^2_p-p\right)/{\sqrt{2p}}\overset{d}\longrightarrow N(0,1),\text{ as }p\rightarrow \infty$. This motivates our use of (\ref{statistic}) and explains why we aspire to establish a standard normal distribution under the null hypothesis. Intuitively, the test statistic is based on the sample covariance between the residual from the parametric model and the discrepancy between the parametric and nonparametric fitted values, as in \cite{Hong1995}.
	
	\cite{Hong1995} also note that, due to the nonparametric nature of the problem, such a statistic vanishes faster than the parametric ($n^{\frac{1}{2}}$) rate, thus a $n^{\frac{1}{2}}$-normalization leads to degeneracy of the test. A proper normalization as in (\ref{statistic}) will yield a non-degenerate limiting distribution. As \cite{Hong1995} noted, our test is one-sided. This is because asymptotically negative
	values of our test statistic can occur only under the null, while under the alternative it tends to a positive, increasing number. Thus, we reject the
	null if our test statistic is on the right tail.

	\section{Asymptotic theory}\label{sec:asymptotics}
	
	\subsection{Consistency of $\widehat{\protect \gamma }$}
	We first provide conditions under which our estimator $\widehat\gamma$ of $\gamma_0$ is consistent. Such a property is necessary for the results that follow. The following assumption is a rather standard type of asymptotic boundedness and full-rank condition on $\Sigma(\gamma)$.
	\begin{assumption}\label{ass:Sigma_spec_norm}
		\[\varlimsup_{n\rightarrow\infty}\sup_{\gamma\in\Gamma}\bar\varphi\left(\Sigma(\gamma)\right)<\infty \text { and } \varliminf_{n\rightarrow\infty}\inf_{\gamma\in\Gamma}\underline\varphi\left(\Sigma(\gamma)\right)>0.\]
	\end{assumption}
	\begin{assumption}
		\label{ass:errors_epsilon_3} The $u_i, i=1,\ldots,n,$ satisfy the representation (\ref{linpro_gamma}). The $\varepsilon _{s}$, $s\geq 1$, have zero mean, finite third and
		fourth moments $\mu _{3}$ and $\mu _{4}$ respectively and, denoting by $\sigma _{ij}(\gamma )$ the $(i,j)$-th element of $\Sigma
		(\gamma )$ and defining
		$
		b_{is}^{\ast }={b_{is}}/{\sigma _{ii}^{\frac{1}{2}}},\;i=1,\ldots ,n,\;n\geq
		1,s\geq 1,
		$
		we have
		\begin{equation}
			\underset{n\rightarrow \infty }{\overline{\lim }}\sup_{i=1,\ldots
				,n}\sum_{s=1}^{\infty }\left \vert b_{is}^{\ast }\right \vert +\sup_{s\geq 1}%
			\underset{n\rightarrow \infty }{\overline{\lim }}\sum_{i=1}^{n}\left \vert
			b_{is}^{\ast }\right \vert <\infty .  \label{cstar}
		\end{equation}
	\end{assumption}
	\noindent  By Assumption \ref{ass:Sigma_spec_norm}, $\sigma_{ii}$ is bounded and bounded away from zero, so the normalization of the $b_{is}$ in Assumption \ref{ass:errors_epsilon_3} is well defined. The summability conditions in (\ref{cstar}) are typical conditions on linear process coefficients that are needed to control dependence; for instance in the case of stationary time series $b^*_{is}=b^*_{i-s}$. The infinite linear process assumed in (\ref{linpro_gamma}) is further discussed by \cite{Robinson2011}, who introduced it, and also by \cite{Delgado2015}. These assumptions imply an increasing-domain asymptotic setup and preclude infill asymptotics.
	
	Because we often need to consider the difference between values of the matrix-valued function $\Sigma(\cdot)$ at distinct points, it is useful to introduce an appropriate concept of `smoothness'. This concept has been employed before in economics, see e.g. \cite{Chen2007}, and is defined below.
	\begin{definition} Let $\left(X,\left\Vert\cdot\right\Vert_X\right)$ and $\left(Y,\left\Vert\cdot\right\Vert_Y\right)$ be Banach spaces, $\mathscr{L}(X,Y)$ be the Banach space of linear continuous maps from $X$ to $Y$ with norm $\left\Vert T\right\Vert_{\mathscr{L}(X,Y)}=\sup_{\left\Vert x\right\Vert_X\leq 1}\left\Vert T(x)\right\Vert_Y$ and $U$ be an open subset of $X$. A map $F:U\rightarrow Y$ is said to be Fr\'echet-differentiable at $u\in U$ if there exists $L\in\mathscr{L}(X,Y)$ such that
		\begin{equation}\label{frechet_def}
			\lim_{\left\Vert h\right\Vert_X\rightarrow 0}\frac{F(u+h)-F(u)-L(h)}{\left\Vert h\right\Vert_X}=0.
		\end{equation}
		$L$ is called the Fr\'echet-derivative of $F$ at $u$. The map $F$ is said to be Fr\'echet-differentiable on $U$ if it is Fr\'echet-differentiable for all $u\in U$.
	\end{definition}
	The above definition extends the notion of a derivative that is familiar from real analysis to the functional spaces and allows us to check high-level assumptions that past literature has imposed. To the best of our knowledge, this is the first use of such a concept in the literature on spatial/network models.
	Denote by $\mathcal{M}^{n\times n}$ the set of real, symmetric and positive semi-definite $n\times n$ matrices.  Let $\Gamma^o$ be an open subset of $\Gamma$ and consider the Banach spaces $\left(\Gamma,\left\Vert\cdot\right\Vert_g\right)$ and $\left(\mathcal{M}^{n\times n},\left\Vert\cdot\right\Vert\right)$, where $\left\Vert \cdot\right\Vert_g$ is a generic $\ell_p$ norm, $p\geq 1$.  The following assumption ensures that $\Sigma(\cdot)$ is a `smooth' function, in the sense of Fr\'echet-smoothness.
	\begin{assumption}\label{ass:Sigma_frech_der}
		The map $\Sigma:\Gamma^o\rightarrow \mathcal{M}^{n\times n}$ is Fr\'echet-differentiable on $\Gamma^o$ with Fr\'echet-derivative denoted $D\Sigma\in\mathscr{L}\left(\Gamma^o,\mathcal{M}^{n\times n}\right)$. Furthermore, the map $D\Sigma$ satisfies
		\begin{equation}\label{Sigma_fre_der_bdd}
			\sup_{\gamma\in\Gamma^o}\left\Vert D\Sigma(\gamma)\right\Vert_{\mathscr{L}\left(\Gamma^o,\mathcal{M}^{n\times n}\right)}\leq C.
		\end{equation}
	\end{assumption}
	Assumption \ref{ass:Sigma_frech_der} is a functional smoothness condition on spatial dependence. It has the advantage of being checkable for a variety of commonly employed models. For example, a first-order SEM has $\Sigma(\gamma)=A^{-1}(\gamma)A'^{-1}(\gamma)$ with $A=I_n-\gamma W$. Corollary \ref{cor:sem_SAR_frech_der} in the supplementary appendix shows $\left(D\Sigma(\gamma)\right)\left(\gamma^\dag\right)=\gamma^\dag A^{-1}(\gamma)\left(G'(\gamma)+G(\gamma)\right)A'^{-1}(\gamma)$,  at a given point $\gamma\in\Gamma^o$, where $G(\gamma)=WA^{-1}(\gamma)$. Then, taking \begin{equation}\label{SEM_frech}
		\left\Vert W\right\Vert+\sup_{\gamma\in\Gamma}\left\Vert A^{-1}(\gamma)\right\Vert<C
	\end{equation}
	yields Assumption \ref{ass:Sigma_frech_der}. Condition (\ref{SEM_frech}) limits the extent of spatial dependence and is very standard in the spatial literature; see e.g. \cite{lee2004asymptotic} and numerous subsequent papers employing similar conditions.
	
	Fr\'echet derivatives for higher-order SAR, SMA, SARMA and MESS error structures are computed in supplementary appendix \ref{app:lemmas}, in Lemmas \ref{lemma:sem_SARMA_frech_der}-\ref{lemma:sem_MESS_frech_der} and Corollaries \ref{cor:sem_SAR_frech_der}-\ref{cor:sem_SMA_frech_der}. Strictly speaking, Gateaux differentiability might suffice for the type of results that we target. We opt for Fr\'echet differentiability because this derivative map is linear and continuous or, equivalently, a bounded linear operator, a property that makes Assumption \ref{ass:Sigma_frech_der} more reasonable.
	
	The following proposition is very useful in `linearizing' perturbations in the $\Sigma(\cdot)$.
	
	\begin{proposition}\label{prop:Sigma_diff_bound}
		If Assumption \ref{ass:Sigma_frech_der} holds, then for any $\gamma_1,\gamma_2\in\Gamma^o$,
		\begin{equation}\label{Sigma_diff}
			\left\Vert\Sigma\left(\gamma_1\right)-\Sigma\left(\gamma_2\right)\right\Vert\leq C \left\Vert\gamma_1-\gamma_2\right\Vert.
		\end{equation}
	\end{proposition}
	To illustrate how the concept of Fr\'echet-differentiability allows us to check high-level assumptions extant in the literature, a consequence of Proposition \ref{prop:Sigma_diff_bound} is the following corollary, a version of which  appears as an assumption in \cite{Delgado2015}.
	\begin{corollary}
		\label{cor:Sigma_equic} For any $\gamma ^{* }\in \Gamma^o $ and any $\eta
		>0 $,
		\begin{equation}
			\underset{n\rightarrow \infty }{\overline{\lim }}\sup_{\gamma \in \left \{
				\gamma :\left \Vert \gamma -\gamma ^{*}\right \Vert <\eta \right
				\} \cap \Gamma^o }\left \Vert \Sigma (\gamma )-\Sigma \left( \gamma ^{*
			}\right) \right \Vert <C\eta .  \label{Sigma_equic_cond}
		\end{equation}
	\end{corollary}
	
	We now introduce  regularity conditions needed to establish the consistency of $\hat\gamma$. Define \[
	\sigma ^{2}\left(
	\gamma \right) =n^{-1}\sigma ^{2}tr\left( \Sigma (\gamma )^{-1}\Sigma
	\right) =n^{-1}\sigma ^{2}\left \Vert E(\gamma )E^{-1}\right \Vert _{F}^{2},
	\]
	which is nonnegative by definition and bounded by Assumption \ref{ass:Sigma_spec_norm}, red with the matrix $E(\gamma)$ defined after (\ref{sigmapmleconc}).
	
	\begin{assumption}
		\label{ass:sigma_unif}  $c\leq \sigma ^{2}\left( \gamma \right)\leq C$ for all $\gamma \in
		\Gamma$.
	\end{assumption}
	
	\begin{assumption}
		\label{ass:gamma_ident} $\gamma _{0}\in \Gamma $ and, for any $\eta >0$,
		\begin{equation}
			\varliminf_{n\rightarrow \infty }\inf_{\gamma \in \overline{\mathcal{N}}%
				^{\gamma }(\eta )}\frac{n^{-1}tr\left( \Sigma (\gamma )^{-1}\Sigma \right) }{%
				\left \vert \Sigma (\gamma )^{-1}\Sigma \right \vert ^{1/n}}>1,
			\label{gamma_ident}
		\end{equation}%
		where $\overline{\mathcal{N}}^{\gamma }(\eta )=\Gamma \setminus \mathcal{N}%
		^{\gamma }(\eta )$ and $\mathcal{N}^{\gamma }(\eta )=\left \{ \gamma
		:\left
		\Vert \gamma -\gamma _{0}\right \Vert <\eta \right \} \cap \Gamma $.
	\end{assumption}

	\begin{assumption}
		\label{ass:Psi_GLS_type} $\left \{ \underline{\varphi }\left( n^{-1}\Psi ^{\prime}\Psi \right) \right \} ^{-1}+\overline{\varphi }\left( n^{-1}\Psi ^{\prime}\Psi \right)=O_{p}(1)$%
		.
	\end{assumption}
	\noindent Assumption \ref{ass:sigma_unif} is a boundedness condition originally considered in \cite{Gupta2018}, while Assumptions \ref{ass:gamma_ident} and \ref{ass:Psi_GLS_type} are identification conditions. Indeed, Assumption \ref{ass:gamma_ident} requires that $\Sigma(\gamma)$ be identifiable in a small neighborhood around $\gamma_0$. This is apparent on noticing that the ratio in (\ref{gamma_ident}) is at least one by the inequality between arithmetic and geometric means, and equals one when $\Sigma(\gamma)=\Sigma$. Similar assumptions arise frequently in related literature, see e.g.  \cite{lee2004asymptotic}, \cite{Delgado2015}. Assumption \ref{ass:Psi_GLS_type} is a typical asymptotic boundedness and non-multicollinearity condition, see e.g. \cite{Newey1997} and much other literature on series estimation. Primitive conditions for this assumption to hold require the convergence (in matrix norm) of $n^{-1}\Psi'\Psi$ to its expectation, and this entails restrictions on the extent of spatial dependence in the $x_i$. A reference is \cite{Lee2016}, wherein consider Assumption A.4 and the proof of Theorem 1. By Assumption \ref{ass:Sigma_spec_norm}, \ref{ass:Psi_GLS_type} implies $\sup_{\gamma\in\Gamma}\left\{\underline{\varphi}\left(n^{-1}\Psi'\Sigma(\gamma)^{-1}\Psi\right)\right\}^{-1}=O_p(1)$.
	\begin{theorem}\label{thm:consistency}
		\sloppy Under either $H_0$ or $H_1$, Assumptions \ref{ass:approx_error}-\ref{ass:Psi_GLS_type} and $p^{-1}+\left(d_\gamma+p\right)/n\rightarrow 0$ as $n\rightarrow\infty$, $\left \Vert \left(\widehat{\gamma},\hat{\sigma}^2\right)-\left(\gamma_0,\sigma_0^2\right) \right \Vert \overset{p}{
			\longrightarrow }0.$
	\end{theorem}

	\subsection{Asymptotic properties of the test statistic}
	Write $\Sigma_j(\gamma)=\partial\Sigma(\gamma)/\partial\gamma_j$, $j=1,\ldots,d_\gamma$,  the matrix differentiated element-wise. While Assumption \ref{ass:Sigma_frech_der} guarantees that these partial derivatives exist, the next assumption imposes a uniform bound on their spectral norms.
	\begin{assumption}\label{ass:Sigma_der_spec}
		$\varlimsup_{n\rightarrow\infty}\sup_{j=1,\ldots,d_\gamma}\left\Vert \Sigma_j(\gamma)\right\Vert<C$.
	\end{assumption}
	We will later consider the sequence of local alternatives %
	\begin{equation}\label{local_alternatives}
		H_{\ell n}\equiv H_{\ell}:f(x_{i},\alpha _{n}^{\ast })=\theta
		_{0}(x_{i})+(p^{1/4}/n^{1/2})h(x_{i}),a.s.,
	\end{equation}%
	where $h$ is square integrable on the support $\mathcal{X}$ of the $x_i$. Under the null $H_0$, we have $h(x_i)=0$, a.s..
	\begin{assumption}\label{ass:alpha_der}
		\sloppy For each $n\in\mathbb{N}$ and $i=1,\ldots,n$, the function $f:\mathcal{X}\times \mathcal{A}\rightarrow\mathbb{R}$ is such that $f\left(x_i,\alpha\right)$ is measurable for each $\alpha\in \mathcal{A}$, $f\left(x_i,\cdot\right)$ is a.s. continuous on $\mathcal{A}$, with $\sup_{\alpha\in \mathcal{A}}f^2\left(x_i,\alpha\right)\leq D_n\left(x_i\right)$, where $\sup_{n\in\mathbb{N}}D_n\left(x_i\right)$ is integrable and $\sup_{\alpha\in \mathcal{A}}\left\Vert\partial f\left(x_i,\alpha\right)/\partial\alpha\right\Vert^2\leq D_n\left(x_i\right)$, $\sup_{\alpha\in \mathcal{A}}\left\Vert\partial^2 f\left(x_i,\alpha\right)/\partial\alpha\partial\alpha'\right\Vert\leq D_n\left(x_i\right)$, all holding a.s..
	\end{assumption}
	Define the infinite-dimensional matrix $\fancyv=B^{\prime }\Sigma ^{-1}\Psi \left(
	\Psi ^{\prime }\Sigma ^{-1}\Psi \right) ^{-1}\Psi ^{\prime }\Sigma ^{-1}B$,
	which is symmetric, idempotent and has rank $p$. We now show that our test statistic is approximated by a quadratic form in $\varepsilon$, weighted by $\fancyv$.

	\begin{theorem}\label{thm:stat_appr}
		Under Assumptions \ref{ass:approx_error}-\ref{ass:alpha_der}, $p^{-1}+p\left(p+d_\gamma^2\right)/n+\sqrt{n}/p^{\mu+1/4}\rightarrow 0$, as $n\rightarrow\infty$, and $H_0$, $
		\fancyt-{\left(\sigma_0^{-2}\varepsilon'\fancyv\varepsilon-p\right)}/{\sqrt{2p}}=o_p(1).
		$
	\end{theorem}
	\begin{assumption}
		\label{ass:rsums_no_coll}
		$\underset{n\rightarrow \infty }{\overline{\lim }}\left \Vert \Sigma
		^{-1}\right \Vert _{R}<\infty.$
	\end{assumption}
	Because $\left \Vert \Sigma
	^{-1}\right \Vert\leq \left \Vert \Sigma
	^{-1}\right \Vert _{R}$, this restriction on spatial dependence is somewhat stronger than a restriction on spectral norm but is typically imposed for central limit theorems in this type of setting, cf. \cite{lee2004asymptotic}, \cite{Delgado2015}, \cite{Gupta2018}. The next assumption  is needed in our proofs to check a Lyapunov condition. A typical approach would be assume moments of order $4+\epsilon$, for some $\epsilon>0$. Due to the linear process structure under consideration, taking $\epsilon=4$ makes the proof tractable, see for example \cite{Delgado2015}.
	\begin{assumption}\label{ass:err_8th}
		The $\varepsilon_s$, $s\geq 1$, have  finite eighth moment.
	\end{assumption}
	The next assumption is strong if the basis functions $\psi_{ij}(\cdot)$ are polynomials, requiring all moments to exist in that case.
	\begin{assumption}\label{ass:psi_mom}
		$\mathcal{E}\left\vert\psi_{ij}\left(x\right)\right\vert<C$, $i=1,\ldots,n$ and $j=1,\ldots,p$.
	\end{assumption}
	The next theorem establishes the asymptotic normality of the approximating quadratic form introduced above.
	
	\begin{theorem}\label{thm:appr_clt}
		Under Assumptions \ref{ass:Sigma_spec_norm}, \ref{ass:errors_epsilon_3}, \ref{ass:Psi_GLS_type}, \ref{ass:rsums_no_coll}-\ref{ass:psi_mom}  and $p^{-1}+p^3/n\rightarrow 0$, as $n\rightarrow\infty$, $
		{\left(\sigma _{0}^{-2}\varepsilon ^{\prime }\fancyv\varepsilon -p\right)}/{\sqrt{2p}}\overset{d}{\longrightarrow}N(0,1).
		$
	\end{theorem}
	This is a new type of CLT, integrating both a linear process framework as well as an increasing dimension element. A linear-quadratic form in iid disturbances is treated by \cite{Kelejian2001}, while a quadratic form in a linear process framework is treated by \cite{Delgado2015}. However both results are established in a parametric framework, entailing no increasing dimension aspect of the type we face with $p\rightarrow\infty$.
	
	Next, we summarize the properties of our test statistic in a theorem that records its asymptotic normality under the null, consistency and ability to detect local alternatives at $p^{1/4}/n^{1/2}$ rate. This rate has been found also by \cite{DeJong1994} and \cite{Gupta2018c}. Introduce the quantity $\varkappa=\left({\sqrt{2}\sigma_0^2}\right)^{-1}\plim_{n\rightarrow\infty} {n^{-1}h'\Sigma^{-1}h}$, where $h=\left(h\left(x_1\right),\ldots,h\left(x_n\right)\right)'$ and $h\left(x_i\right)$ is from (\ref{local_alternatives}).
	\begin{theorem}\label{thm:stat_properties}
		Under the conditions of Theorems \ref{thm:stat_appr} and \ref{thm:appr_clt},
		(1) $\fancyt\overset{d}{\rightarrow} N(0,1)$ under $H_0$, (2)
		$\fancyt$ is a consistent test statistic,
		(3) $\fancyt\overset{d}{\rightarrow}N \left(\varkappa,1\right)$ under local alternatives $H_{\ell}$.
	\end{theorem}
	\section{Models with SAR structure in responses}\label{sec:SAR_ext}
	We now introduce the SAR model
	\begin{equation}
		y_{i}=\sum_{j=1}^{d_{\lambda}}\lambda_{0j} w_{i,j}'y+\theta _{0}\left( x_{i}\right) +u_{i},i=1,\ldots ,n,  \label{sar_model}
	\end{equation}%
	where $W_j$, $j=1,\ldots,d_{\lambda}$, are known spatial weight matrices with $i$-th rows denoted $w_{i,j}'$, as discussed earlier, and $\lambda_{0j}$ are unknown parameters measuring the strength of spatial dependence. We take $d_\lambda$ to be fixed for convenience of exposition. The error structure remains the same as in (\ref{linpro_gamma}). Here spatial dependence arises not only in errors but also responses. For example, this corresponds to a situation where agents in a network influence each other both in their observed and unobserved actions. Note that the error term $u_i$ can be generated by the same $W_j$, or different ones.
	
	While the model in (\ref{sar_model}) is new in the literature, some related ones are discussed here. Models such as (\ref{sar_model}) but without dependence in the error structure are considered by \cite{Su2010} and \cite{Gupta2013,Gupta2018}, but the former consider only $d_{\lambda}=1$ and the latter only parametric $\theta_0(\cdot)$. Linear $\theta_0(\cdot)$ and $d_{\lambda}>1$ are permitted by \cite{Lee2010}, but the dependence structure in errors differs from what we allow in (\ref{sar_model}). Using the same setup as \cite{Su2010} and independent disturbances, a specification test for the linearity of $\theta_0(\cdot)$ is proposed by \cite{Su2017}. In comparison, our model is much more general and our test can handle more general parametric null hypotheses. We thank a referee for pointing out that (\ref{sar_model}) is a particular case of \cite{Sun2016} when $u_i$ are iid and of \cite{Malikov2017} when $d_\lambda=1$.
	
	Denoting $S(\lambda)=I_n-\sum_{j=1}^{d_{\lambda}} \lambda_j W_j$, the quasi likelihood function based on Gaussianity and conditional on covariates is
	\begin{eqnarray}
		L(\beta,\sigma^2,\phi)=\log{(2\pi\sigma^2)}-\frac{2}{n}\log {\left\vert{S\left(\lambda\right)}\right\vert} +\frac{1}{n}\log {\left\vert{\Sigma\left(\gamma\right)}\right\vert}\nonumber\\+ \frac{1}{\sigma^{2}{n}}\left( S\left(\lambda\right)y-\Psi\beta\right)'\Sigma(\gamma)^{-1}\left( S\left(\lambda\right)y-\Psi\beta\right),\label{likelihood_SAR}
	\end{eqnarray}
	at any admissible point $\left(\beta',\phi',\sigma^2\right)'$ with $\phi=\left(\lambda',\gamma'\right)'$, for nonsingular $S(\lambda)$ and $\Sigma(\gamma)$.
	For given $\phi=\left(\lambda',\gamma'\right)'$, (\ref{likelihood_SAR}) is minimised with respect to $\beta$ and $\sigma^2$ by
	\begin{eqnarray}
		\bar{\beta}\left(\phi\right) & = & \left(\Psi'\Sigma(\gamma)^{-1}\Psi\right)^{-1}\Psi'\Sigma(\gamma)^{-1}S\left(\lambda\right)y \label{betapmle_SAR},\\
		\bar{\sigma}^{2}\left(\phi\right) & = & {n^{-1}}y'S'\left(\lambda\right) E(\gamma)'M(\gamma) E(\gamma) S\left(\lambda\right)y \label{sigmapmleconc_SAR}.
	\end{eqnarray}
	The QMLE of $\phi_0$ is $\widehat\phi=\operatorname*{arg\,min}_{\phi\in\Phi}\mathcal{L}\left(\phi\right)$,
	where
	\begin{equation}\label{conclik}
		\mathcal{L}\left(\phi\right)=\log\bar{\sigma}^{2}\left(\phi\right)+{n^{-1}}\log\left\vert S'^{-1}\left(\lambda\right)\Sigma(\gamma) S^{-1}\left(\lambda\right)\right\vert,
	\end{equation}
	and $\Phi=\Lambda\times\Gamma$ is taken to be a compact subset of $\mathbb{R}^{d_{\lambda}+d\gamma}$.  The QMLEs of $\beta_{0}$ and $\sigma_0^2$ are defined as $\bar{\beta}\left(\widehat{\phi}\right)\equiv\widehat{\beta}$ and  $\bar{\sigma}^{2}\left(\widehat{\phi}\right)\equiv\widehat{\sigma}^2$ respectively. The following assumption controls spatial dependence and is discussed below equation (\ref{SEM_frech}).
	\setcounter{assumption}{0}
	\renewcommand\theassumption{SAR.\arabic{assumption}}
	\begin{assumption}\label{ass:spec_norm_G}
		$\max_{j=1,\ldots,d_{\lambda}}\left\Vert W_j\right\Vert+\left\Vert S^{-1}\right\Vert<C$.
	\end{assumption}
	Writing $T(\lambda)=S(\lambda)S^{-1}$ and $\phi=\left(\lambda',\gamma'\right)'$, define the quantity
	\[\sigma ^{2}\left( \phi \right) =n^{-1}\sigma_0^2tr\left(T'(\lambda)\Sigma(\gamma)^{-1}T(\lambda)\Sigma\right)=n^{-1}\sigma_0^2\left\Vert E(\gamma)T(\lambda)E^{-1}\right\Vert_F^2,
	\] which is nonnegative  by definition and bounded by Assumptions \ref{ass:Sigma_spec_norm} and \ref{ass:spec_norm_G}.
	The assumptions below directly extend Assumptions \ref{ass:sigma_unif} and \ref{ass:gamma_ident} to the present setup.
	\begin{assumption}\label{ass:sigma_unif_SAR}
		$c\leq\sigma ^{2}\left( \phi \right)\leq C$, for all $\phi\in\Phi$.
	\end{assumption}
	\begin{assumption}\label{ass:tau_lambda_ident}
		$\phi_0\in\Phi$ and, for any $\eta>0$,
		\begin{equation}\label{tau_lambda_ident}
			\varliminf_{n\rightarrow\infty}\inf_{\phi\in\overline{\mathcal{N}}^{\phi}(\eta)}\frac{n^{-1}tr\left(T'(\lambda)\Sigma(\gamma)^{-1}T(\lambda)\Sigma\right)}{\left\vert T'(\lambda)\Sigma(\gamma)^{-1}T(\lambda)\Sigma\right\vert ^{1/n}}>1,
		\end{equation}
		where $\overline{\mathcal{N}}^{\phi}(\eta)=\Phi\setminus \mathcal{N}^{\phi}(\eta)$ and $\mathcal{N}^{\phi}(\eta)=\left\{\phi:\left\Vert\phi-\phi_0\right\Vert<\eta\right\}\cap\Phi$.
	\end{assumption}
	We now introduce an identification condition that is required in the setup of this section.
	\begin{assumption}\label{ass:reg_ident_SAR} $\beta_0\neq 0$ and for any $\eta>0$,
		\begin{equation}\label{reg_ident_SAR}
			{P\left(\varliminf_{n\rightarrow\infty}\inf_{\left(\lambda',\gamma'\right)'\in\Lambda\times \overline{\mathcal{N}}^{\gamma}(\eta)}n^{-1}\beta_0'\Psi^{\prime }T^{\prime}(\lambda ) E(\gamma)'M\left( \gamma \right) E(\gamma) T(\lambda )\Psi\beta_0/\left\Vert \beta_0\right\Vert^2>0\right)=1.}
		\end{equation}
	\end{assumption}
	Upon performing minimization with respect to $\beta$, the event inside the probability in (\ref{reg_ident_SAR}) is equivalent to the event
	\[{\varliminf_{n\rightarrow\infty}\min_{\beta\in\mathbb{R}^p}\inf_{\left(\lambda',\gamma'\right)'\in\Lambda\times\overline{\mathcal{N}}^{\gamma}(\eta)}n^{-1}\left( \Psi\beta-T(\lambda)\Psi\beta_0\right)'\Sigma(\gamma)^{-1}\left( \Psi\beta-T(\lambda)\Psi\beta_0\right)/\left\Vert \beta_0\right\Vert^2>0,}
	\]
	which is analogous to the identification condition for the nonlinear
	regression model with a
	parametric linear factor in \cite{robinson1972non}, weighted by the inverse of the error covariance matrix. This reduces the condition to a scalar form of a rank condition, making the identifying nature of the assumption transparent. A similar identifying assumption is used by \cite{Gupta2018}.
	\begin{theorem}\label{thm:cons_SAR}
		Under either $H_0$ or $H_1$,  Assumptions \ref{ass:approx_error}-\ref{ass:Sigma_frech_der}, \ref{ass:Psi_GLS_type}, \ref{ass:spec_norm_G}-\ref{ass:reg_ident_SAR} and \[p^{-1}+\left(d_\gamma+p\right)/n\rightarrow 0, \text{ as }n\rightarrow\infty,\] $\left\Vert\left(\widehat\phi,\widehat{\sigma}^2\right)-\left(\phi_0,{\sigma_0}^2\right)\right\Vert\overset{p}\longrightarrow 0$ as $n\rightarrow\infty$.
	\end{theorem}
	The test statistic $\fancyt$ can be constructed as before but with the null residuals redefined to incorporate the spatially lagged terms, i.e. $\hat u=S(\hat\lambda)y-f(x,\hat\alpha)$. Then we have the following theorem.
	\begin{theorem}\label{thm:stat_appr_SAR}
		Under Assumptions \ref{ass:approx_error}-\ref{ass:Sigma_frech_der}, \ref{ass:Psi_GLS_type}-\ref{ass:alpha_der}, \ref{ass:spec_norm_G}-\ref{ass:reg_ident_SAR}, \[p^{-1}+p\left(p+d_\gamma^2\right)/n+\sqrt{n}/p^{\mu+1/4}+d^2_\gamma/p\rightarrow 0, \text{ as }n\rightarrow\infty,\]
		and $H_0$, $
		\fancyt-{\left(\sigma_0^{-2}\varepsilon'\fancyv\varepsilon-p\right)}/{\sqrt{2p}}=o_p(1).
		$
	\end{theorem}
	\begin{theorem}\label{thm:stat_properties_SAR}
		Under the conditions of Theorems \ref{thm:appr_clt}, \ref{thm:cons_SAR} and \ref{thm:stat_appr_SAR},
		(1) $\fancyt\overset{d}{\rightarrow} N(0,1)$ under $H_0$, (2)
		$\fancyt$ is a consistent test statistic,
		(3) $\fancyt\overset{d}{\rightarrow}N \left(\varkappa,1\right)$ under local alternatives $H_{\ell}$.
	\end{theorem}
	\section{Nonparametric spatial weights}\label{sec:nonpar_ext}
	In this section we are motivated by settings where spatial dependence occurs through nonparametric functions of raw distances (this may be geographic, social, economic, or any other type of distance), as is the case in \cite{pinkse2002}, for example. In their kind of setup, $d_{ij}$ is a raw distance between units $i$ and $j$  and the corresponding element of the spatial weight matrix is given by $w_{ij}=\zeta_0\left(d_{ij}\right)$, where $\zeta_0(\cdot)$ is an unknown nonparametric function. \cite{pinkse2002} use such a setup in a SAR model like (\ref{sar_model}), but with a linear regression function. In contrast, in keeping with the focus of this paper we instead model dependence in the errors in this manner.  Our formulation is rather general, covering, for example, a specification like $w_{ij}=f\left(\gamma_0,\zeta_0\left(d_{ij}\right)\right)$, with $f(\cdot)$ a \emph{known} function, $\gamma_0$ an \emph{unknown} parameter of possibly increasing dimension, and $\zeta_0(\cdot)$ an \emph{unknown} nonparametric function. For the sake of simplicity, we do not permit the $x_i$ in this section to be generated by such nonparametric weight matrices although they can be generated from other, known weight matrices.
	
	Let $\Xi$ be a compact space of functions, on which we will specify more conditions later. For notational simplicity we abstract away from the SAR dependence in the responses. Thus we consider (\ref{model}), but with
	\begin{equation}
		u_{i}=\sum_{s=1}^{\infty }b_{is}\left(\gamma_0,\zeta_0\left({z_i}\right)\right)\varepsilon _{s},  \label{linpro_np}
	\end{equation}
	where $\zeta_0(\cdot)=\left(\zeta_{01}(\cdot),\ldots,\zeta_{0d_\zeta}(\cdot)\right)'$ is a fixed-dimensional vector of real-valued nonparametric functions with $\zeta_{0\ell}\in\Xi$ for each $\ell=1,\ldots,d_\zeta$, and ${z}_i$ a fixed-dimensional vector of data, independent of the $\varepsilon_s$, $s\geq 1$, with support $\mathcal{Z}$. One can also take $z_i$ to be a fixed distance measure. We base our estimation on approximating each $\zeta_{0\ell}({z_i})$, $\ell=1,\ldots,d_\zeta$, with the series representation $\delta_{0\ell}'\varphi_\ell({z_i})$, where $\varphi_\ell\left({z_i}\right)\equiv\varphi_{\ell}$ is an $r_\ell\times 1$ ($r_\ell\rightarrow\infty$ as $n\rightarrow\infty$) vector of basis functions with typical function $\varphi_{\ell k}$, $k=1,\ldots,r_\ell$. The set of linear combinations $\delta_\ell'\varphi_\ell({z_i})$ forms the sequence of sieve spaces $\Phi_{r_\ell}\subset \Xi$ as $r_\ell\rightarrow \infty$, for any $\ell=1,\ldots,d_\zeta$, and
	\begin{equation}\label{zeta_series_appr}
		\zeta_{0\ell}\left({z}\right)=\delta_{0\ell}'\varphi_{\ell}+\nu_{\ell},
	\end{equation}
	with the following restriction on the function space $\Xi$:
	\setcounter{assumption}{0}
	\renewcommand\theassumption{NPN.\arabic{assumption}}
	\begin{assumption}\label{ass:approx_error_order_np}
		For some scalars $\kappa_\ell>0$, $\left \Vert \nu_\ell \right \Vert_{w_z} =O\left(r_\ell^{-\kappa_\ell}\right),$ as $r_\ell\rightarrow\infty$, $\ell=1,\ldots,d_\zeta$, where $w_z\geq 0$ is the largest value such that $\sup_{z\in\mathcal{Z}} \mathcal{E}\left\Vert z\right\Vert^{w_z}<\infty$
	\end{assumption}
	Just as Assumption \ref{ass:approx_error} implied (\ref{appr_error}), by Lemma 1 of \cite{Lee2016}, we obtain
	\begin{equation}\label{approx_error_zeta}
		\sup_{z\in\mathcal{Z}}\mathcal{E}\left( \nu^2_{\ell}\right)=O\left(r_\ell^{-2\kappa_\ell}\right), \ell=1,\ldots,d_\zeta.
	\end{equation}
	Thus we now have an infinite-dimensional nuisance parameter $\zeta_0(\cdot)$ and increasing-dimensional nuisance parameter $\gamma$. Writing $\sum_{\ell=1}^{d_\zeta}r_\ell=r$ and $\tau=(\gamma',\delta'_1,\ldots,\delta'_{d_\zeta})'$, which has increasing dimension $d_\tau=d_\gamma+r $, define
	$
	\varsigma(r)=\sup_{z\in\mathcal{Z}; \ell=1,\ldots,d_\zeta}\left\Vert\varphi_{\ell}\right\Vert.$ Write $\Sigma(\tau)$ for the covariance matrix of the $n\times 1$ vector of $u_i$ in (\ref{linpro_np}), with $\delta_{\ell}'\varphi_\ell$ replacing each admissible function $\zeta_{\ell}(\cdot)$. This is analogous to the definition of $\Sigma(\gamma)$ in earlier sections, and indeed after conditioning on $z$ it can be treated in a similar way because $d_\gamma\rightarrow\infty$ was already permitted. For example, suppose that $u=(I_n-W)^{-1}\varepsilon$, where $\left\Vert W\right\Vert<1$ and the elements satisfy $w_{ij}=\zeta_0\left(d_{ij}\right)$, $i,j=1,\ldots,n$, for some fixed distances $d_{ij}$ and unknown function $\zeta_0(\cdot)$, see e.g. \cite{Pinkse1999}. Approximating $\zeta_0(z)= \tau_0'\varphi(z)+\nu$, for some $r\times 1$ basis function vector $\varphi(z)$ and approximation error $\nu$, we define $W(\tau)$ as the $n\times n$ matrix with elements $w_{ij}(\tau)=\tau_0'\varphi\left(d_{ij}\right)$, and set $\Sigma(\tau)=\text{var}\left((I_n-W(\tau))^{-1}\varepsilon\right)=\sigma_0^2(I_n-W(\tau))^{-1}(I_n-W'(\tau))^{-1}$.
	
	For
	any admissible values $\beta $, $\sigma ^{2}$ and $\tau $, the redefined (multiplied
	by $2/n$) negative quasi log likelihood function based on using the
	approximations (\ref{theta_series_appr}) and (\ref{zeta_series_appr}) is
	\begin{equation}
		{L}(\beta,\sigma^2 ,\tau )=\ln \left( 2\pi \sigma ^{2}\right) +\frac{1}{n}%
		\ln \left \vert \Sigma \left( \tau \right) \right \vert +\frac{1}{n\sigma
			^{2}}(y-\Psi \beta )^{\prime }\Sigma \left( \tau \right) ^{-1}(y-\Psi
		\beta ),  \label{likelihood_zeta}
	\end{equation}%
	which is minimised with respect to $\beta $ and $\sigma ^{2}$ by
	\begin{eqnarray}
		\bar{\beta}\left( \tau \right) &=&\left( \Psi ^{\prime }\Sigma \left(
		\tau \right) ^{-1}\Psi \right) ^{-1}\Psi ^{\prime }\Sigma \left( \tau
		\right) ^{-1}y,  \label{betapmle_zeta} \\
		\bar{\sigma}^{2}\left( \tau \right) &=&{n^{-1}}y^{\prime }E(\tau
		)^{\prime }M(\tau )E(\tau )y,  \label{sigmapmleconc_zeta}
	\end{eqnarray}%
	where $M(\tau )=I_{n}-E(\tau )\Psi \left( \Psi ^{\prime }\Sigma (\tau
	)^{-1}\Psi \right) ^{-1}\Psi ^{\prime }E(\tau )^{\prime }$ and $E(\tau )$
	is the $n\times n$ symmetric matrix such that $E(\tau )E(\tau )^{\prime }=\Sigma
	(\tau )^{-1}$. Thus the concentrated likelihood function is
	\begin{equation}
		\mathcal{L}(\tau )=\ln (2\pi )+\ln \bar{\sigma}^{2}(\tau )+\frac{1}{n}%
		\ln \left \vert \Sigma \left( \tau \right) \right \vert.
		\label{conc_likelihood_zeta}
	\end{equation}%
	Again, for compact $\Gamma$ and sieve coefficient space $\Delta$, the QMLE of $\tau _{0}$ is $\widehat{\tau}=\text{arg min}_{\tau
		\in \Gamma\times\Delta }\mathcal{L}(\tau )$ and the QMLEs of $\beta $ and $\sigma ^{2}$
	are $\widehat{\beta}=\bar{\beta}\left( \widehat{\tau}\right) $ and $\widehat{\sigma}%
	^{2}=\bar{\sigma}^{2}\left( \widehat{\tau}\right) $. The
	series estimate of $\theta _{0} $ is defined as in (\ref{theta_estimate}). Define also the product Banach space $\mathcal{T}=\Gamma\times \Xi^{d_\zeta}$ with norm $\left\Vert \left(\gamma',\zeta'\right)'\right\Vert_{\mathcal{T}_w}=\left\Vert \gamma\right\Vert+\sum_{\ell=1}^{d_\zeta}\left\Vert \zeta_\ell\right\Vert_{w}$, and consider the map $\Sigma:\mathcal{T}^{o}\rightarrow \mathcal{M}^{n\times n}$, where $\mathcal{T}^{o}$ is an open subset of $\mathcal{T}$.
	\begin{assumption}\label{ass:Sigma_frech_der_np}
		The map $\Sigma:\mathcal{T}^o\rightarrow \mathcal{M}^{n\times n}$ is Fr\'echet-differentiable on $\mathcal{T}^o$ with Fr\'echet-derivative denoted $D\Sigma\in\mathscr{L}\left(\mathcal{T}^o,\mathcal{M}^{n\times n}\right)$. Furthermore, conditional on ${z}$, the map $D\Sigma$ satisfies
		\begin{equation}\label{Sigma_semi_fre_der_bdd}
			\sup_{t\in\mathcal{T}^o}\left\Vert D\Sigma(t)\right\Vert_{\mathscr{L}\left(\mathcal{T}^o,\mathcal{M}^{n\times n}\right)}\leq C,
		\end{equation}
		on its domain $\mathcal{T}^o$.
	\end{assumption}
	\noindent This assumption can be checked in a similar way to how we checked Assumption \ref{ass:Sigma_frech_der}, where a diverging dimension for the argument was already permitted.
	\begin{proposition}\label{prop:Sigma_diff_bound_np}
		If Assumption \ref{ass:Sigma_frech_der_np} holds, then for any $t_1,t_2\in\mathcal{T}^o$, conditional on $z$,
		\begin{equation}\label{Sigma_diff_np}
			\left\Vert\Sigma\left(t_1\right)-\Sigma\left(t_2\right)\right\Vert\leq C \varsigma(r)\left\Vert t_1-t_2\right\Vert.
		\end{equation}
	\end{proposition}
	\begin{corollary}
		\label{cor:Sigma_equic_np} For any $t ^{* }\in \mathcal{T}^o $ and any $\eta
		>0 $, conditional on $z$,
		\begin{equation}
			\underset{n\rightarrow \infty }{\overline{\lim }}\sup_{t \in \left \{
				t :\left \Vert t -t ^{*}\right \Vert <\eta \right
				\} \cap \mathcal{T}^o }\left \Vert \Sigma (t )-\Sigma \left( t ^{*
			}\right) \right \Vert <C\varsigma(r)\eta .  \label{Sigma_equic_cond_np}
		\end{equation}
	\end{corollary}

	\begin{assumption}
		\label{ass:sigma_unif_np}  $c\leq \sigma ^{2}\left( \tau \right)\leq C$ for $\tau \in
		\Gamma\times \Delta$, conditional on $z$.
	\end{assumption}
	Denote $\Sigma\left(\tau_0\right)=\Sigma_0$. Note that this is not the true covariance matrix, which is $\Sigma\equiv \Sigma\left(\gamma_0,\zeta_0\right)$.
	\begin{assumption}
		\label{ass:gamma_ident_np} $\tau _{0}\in \Gamma\times\Delta $ and, for any $\eta >0$, conditional on $z$,
		\begin{equation}
			\varliminf_{n\rightarrow \infty }\inf_{\tau \in \overline{\mathcal{N}}%
				^{\tau }(\eta )}\frac{n^{-1}tr\left( \Sigma (\tau )^{-1}\Sigma_0 \right) }{%
				\left \vert \Sigma (\tau )^{-1}\Sigma_0 \right \vert ^{1/n}}>1,
			\label{gamma_ident_np}
		\end{equation}%
		where $\overline{\mathcal{N}}^{\tau }(\eta )=(\Gamma\times\Delta) \setminus \mathcal{N}%
		^{\tau }(\eta )$ and $\mathcal{N}^{\tau }(\eta )=\left \{ \tau
		:\left
		\Vert \tau -\tau _{0}\right \Vert <\eta \right \} \cap (\Gamma\times \Delta) $.
	\end{assumption}

	\begin{remark}
		Expressing the identification condition in Assumption \ref{ass:gamma_ident_np} in terms of $\tau$ implies that identification is guaranteed via the sieve spaces $\Phi_{r_\ell}$, $\ell=1,\ldots,d_\zeta$. This approach is common in the sieve estimation literature, see e.g.  \cite{Chen2007}, p. 5589, Condition 3.1.
	\end{remark}
	\begin{theorem}\label{thm:consistency_np}
		Under either $H_0$ or $H_1$,  Assumptions \ref{ass:approx_error}-\ref{ass:errors_epsilon_3} (with \ref{ass:Sigma_spec_norm} and \ref{ass:errors_epsilon_3} holding for $t\in\mathcal{T}$ rather than $\gamma\in\Gamma$), \ref{ass:Psi_GLS_type}, \ref{ass:approx_error_order_np}-\ref{ass:gamma_ident_np}  and $p^{-1}+\left(\min_{\ell=1,\ldots,d_\zeta}r_\ell\right)^{-1}+\left(d_\gamma+p+\max_{\ell=1,\ldots,d_\zeta}r_\ell\right)/n\rightarrow 0$ as $n\rightarrow\infty$, $\left \Vert \left(\widehat{\tau},\hat\sigma^2\right)-\left(\tau _{0},\sigma^2_0\right)\right \Vert \overset{p}{
			\longrightarrow }0.
		$
	\end{theorem}
	
	\begin{theorem}\label{thm:stat_appr_np}
		Under the conditions of Theorems \ref{thm:stat_appr} and \ref{thm:consistency_np}, but with $\tau$ and $\mathcal{T}$ replacing $\gamma$ and $\Gamma$ in assumptions prefixed with R and $p\rightarrow\infty$, \[\left(\min_{\ell=1,\ldots,d_\zeta}r_\ell\right)^{-1}+
		\frac{p^2}{n}+\frac{\sqrt{n}}{p^{\mu+1/4}} +p^{1/2}\varsigma (r) \left(\frac{ d_{\gamma }+\displaystyle\max_{\ell=1,\ldots,d_\zeta}r_\ell}{\sqrt{n}}+\sqrt{\sum_{\ell=1}^{d_\zeta}r_\ell^{-2\kappa_\ell}}\right )\rightarrow 0,
		\] as $n\rightarrow\infty$, and $H_0$, $
		\fancyt-\left({\sigma_0^{-2}\varepsilon'\fancyv\varepsilon-p}\right)/{\sqrt{2p}}=o_p(1).
		$
	\end{theorem}
	\begin{theorem}\label{thm:stat_properties_np}
		Let the conditions of Theorems \ref{thm:appr_clt} and \ref{thm:stat_appr_np} hold, but with $\tau$ and $\mathcal{T}$ replacing $\gamma$ and $\Gamma$ in assumptions prefixed with R. Then
		(1) $\fancyt\overset{d}{\rightarrow} N(0,1)$ under $H_0$, (2)
		$\fancyt$ is a consistent test statistic,
		(3) $\fancyt\overset{d}{\rightarrow}N \left(\varkappa,1\right)$ under local alternatives $H_{\ell}$.
	\end{theorem}
	
	\section{Fixed-regressor residual-based bootstrap test}\label{sec:bootstrap}
	The performance of nonparametric tests based on asymptotic distributions often leaves something to be desired in finite samples. An alternative approach is to use the
	bootstrap approximation. In this section, we propose a bootstrap version of our test, focusing on the setting of Section \ref{sec:SAR_ext}. In our simulations and empirical studies, we consider test
	statistics based on both $\widehat{m}_{n}=\widehat{\sigma }^{-2}\widehat{v}^{\prime }\Sigma \left( \widehat{\gamma }\right) ^{-1}\widehat{u}/n$ and $\widetilde{m}_{n}=\widehat{\sigma }^{-2}(\widehat{u}^{\prime }\Sigma \left(
	\widehat{\gamma }\right) ^{-1}\widehat{u}-\widehat{\eta }^{\prime }\Sigma
	\left( \widehat{\gamma }\right) ^{-1}\widehat{\eta })/n$, where $\widehat{\eta
	}=S(\hat\lambda)y-\widehat{\theta }$, i.e., the residual from nonparametric estimation, $\hat u=S(\hat\lambda)y-f(x,\hat \alpha)$, and $\hat v=\hat\theta-f(x,\hat \alpha)$.
	Analogous to the definition of $\mathscr{T}_n$, define the statistic $
	\mathscr{T}_{n}^{a}={\left(n\widetilde{m}_{n}-p\right)}/{\sqrt{2p}}.$
	In the case of no spatial autoregressive term, and under the power series, $\mathscr{T}_{n}^{a}$ and $\mathscr{T}_{n}$ are numerically identical, as was
	observed by \cite{Hong1995}. However, in the SARSE setting a
	difference arises due to the spatial structure in the response $y$. We show
	that $\mathscr{T}_{n}^{a}-\mathscr{T}_{n}=o_{p}(1)$ under the null or local alternatives in Theorem \ref{thm:tilde_hat_equiv} in the online supplementary appendix.
	
	The bootstrap versions of the
	test statistics $\fancyt$ and $\fancyt^{a}$ are
	\begin{eqnarray*}
		\fancyt^* &=&\frac{n\widehat{m}_{n}^{\ast }-p}{\sqrt{2p}}=\frac{\widehat{%
				\sigma }^{\ast -2}\widehat{{v}}^{\ast \prime }\Sigma \left( \widehat{%
				\gamma }^{\ast }\right) ^{-1}\widehat{{u}}^{\ast }-p}{\sqrt{2p}} \\
		\fancyt^{a\ast } &=&\frac{n\widetilde{m}_{n}^{\ast }-p}{\sqrt{2p}}=\frac{%
			\widehat{\sigma }^{\ast -2}(\widehat{{u}}^{\ast \prime }\Sigma \left(
			\widehat{\gamma }^{\ast }\right) ^{-1}\widehat{{u}}^{\ast }-\widehat{%
				\eta }^{\ast \prime }\Sigma \left( \widehat{\gamma }^{\ast }\right) ^{-1}%
			\widehat{\eta }^{\ast })-p}{\sqrt{2p}},
	\end{eqnarray*}%
	respectively, where $\widehat{{u}}^{\ast }$ is the bootstrap residual vector under the
	null, $\widehat{\eta }^*$ is the bootstrap residual vector under the alternative, $%
	\widehat{{v}}^{\ast }=\widehat{\theta }^{\ast }(x)-f(x,\widehat{%
		\alpha }^{\ast })$, and $\left(\widehat{%
		\gamma }^{\ast },\lambda^*,\widehat{\sigma }^{\ast 2},\widehat{\theta }^{\ast },\widehat{\alpha }^{\ast }\right)$ is the estimator using the
	bootstrap sample. We elaborate on the bootstrap statistics using the SARARMA($m_{1}$,$%
	m_{2},m_{3}$) model as an example:
	\begin{equation*}
		y=\sum_{k=1}^{m_{1}}\lambda _{k}W_{1k}y+\theta (x)+u\text{, }%
		u=\sum_{l=1}^{m_{2}}\gamma _{2l}W_{2l}u+\sum_{l=1}^{m_{3}}\gamma
		_{3l}W_{3l}\xi +\xi .
	\end{equation*}%
	Following \cite{Jin2015}, we first deduct the empirical mean of the
	residual vector from
	\begin{equation*}
		\widehat{{\xi }}=\left( \sum_{l=1}^{m_{3}}\widehat{\gamma }%
		_{3l}W_{3l}+I_{n}\right) ^{-1}\left( I_{n}-\sum_{l=1}^{m_{2}}\widehat{\gamma
		}_{2l}W_{2l}\right) \left( y-\sum_{k=1}^{m_{1}}\widehat{\lambda }_{k}W_{1k}y-%
		\widehat{{\theta }}_{n}\right)
	\end{equation*}%
	to obtain $\widetilde{{\xi }}=(I_{n}-\frac{1}{n}l_{n}l_{n}^{\prime })%
	\widehat{{\xi }}$. Next, we sample randomly with replacement $n$
	times from elements of $\widetilde{{\xi }}$ to obtain a vector of $%
	\mathbf{\xi }^{\ast }.$ After this, we generate the bootstrap sample $%
	y^{\ast }$ by treating $\widehat{{f}}=f(x,\widehat{\alpha })$, $\hat\lambda$ and
	$\widehat{\gamma }$ as the true parameter:%
	\begin{equation*}
		y^{\ast }=\left( I_{n}-\sum_{k=1}^{m_{1}}\widehat{\lambda }_{k}W_{1k}\right)
		^{-1}\left( \widehat{{f}}+\left( I_{n}-\sum_{l=1}^{m_{2}}\widehat{%
			\gamma }_{2l}W_{2l}\right) ^{-1}\left( \sum_{l=1}^{m_{3}}\widehat{\gamma }%
		_{3l}W_{3l}+I_{n}\right) {\xi }^{\ast }\right) .
	\end{equation*}%
	We estimate the model based on the bootstrap sample $y^{\ast }$ using QMLE
	to obtain the estimator $\widehat{{\theta }}^{\ast }=\psi
	^{\prime }\widehat{\beta }^{\ast },$ $\widehat{\lambda }^{\ast }$, and $
	\widehat{\gamma }^{\ast }$ under the alternative hypothesis and $\widehat{
		\alpha }^{\ast }$ under the null hypothesis of $\theta (x)=f(x,\alpha
	_{0}).$ Then, $\widehat{{\eta}}^{\ast }=y^{\ast }-\sum_{k=1}^{m_{1}}
	\widehat{\lambda }_{k}^{\ast }W_{1k}y^{\ast }-\widehat{{\theta }}^{\ast }$, $\widehat{{u }}^{\ast }=y^{\ast }-\sum_{k=1}^{m_{1}}%
	\widehat{\lambda }_{k}^{\ast }W_{1k}y^{\ast }-f(x,\widehat{\alpha }^{\ast }).$
	
	This procedure is repeated $B$ times to obtain the sequence $\left \{
	\mathscr{T}_{nj}^{\ast }\right \} _{j=1}^{B}$. We reject the null when $p^{\ast
	}=B^{-1}\sum_{j=1}^{B}\mathbf{1(}\fancyt<\mathscr{T}_{nj}^{\ast })$ is smaller than the
	given level of significance. An identical procedure holds for the test based on $\fancyt^{a\ast }.$ The asymptotic validity of the bootstrap method can be shown as in Theorem 4 of \cite{Su2017} and Lemma 2 in \cite{Jin2015}, and detailed analysis can be found in the supplementary appendix, see proof of Theorem TS.1.
	
	\section{Finite sample performance}\label{sec:mc}

	\subsection{Parametric error spatial structure}

	\sloppy Taking $n=60,100,200$, we choose two specifications to generate $y$ from the
	SARARMA($m_{1}$,$m_{2},m_{3}$) models:
	\begin{eqnarray*}
		\text{SARARMA(}0\text{,1,0): } &&y=\theta (x)+u,\text{ }u=\gamma
		_{2}W_{2}u+\xi \\
		\text{SARARMA(}1\text{,0,1): } &&y=\lambda _{1}W_{1}y+\theta (x)+u,\text{ }%
		u=\gamma _{3}W_{3}\xi +\xi ,
	\end{eqnarray*}%
	where $\xi$ is $N(0,I_n)$. The DGP of $\theta (x)$ is
	\begin{equation*}
		\theta (x_{i})=x_{i}^{\prime }\alpha +cp^{1/4}n^{-1/2}  \sin
		(x_{i}^{\prime }\alpha ),
	\end{equation*}%
	where $x_{i}^{\prime }\alpha =1+x_{1i}+x_{2i}$, with $%
	x_{1i}=(z_{i}+z_{1i})/2 $, $x_{2i}=(z_{i}+z_{2i})/2$. We choose two
	settings: compactly supported regressors where $z_{i},z_{1i}$ and $z_{2i}$ are i.i.d., $U[0,2\pi ]$ and unboundedly supported
	regressors where $z_{i}, z_{1i}$ and $z_{2i}$ are i.i.d. $N(0,1).$ We report the compact support setting in the main text, while the results for unbounded support are reported in the online supplement.
	
	\sloppy We use three series bases for our experiments: power
	(polynomial) series of third and fourth order ($p=10,p=15$),
	trigonometric series $trig_1=\left(1, \sin\left(x_1\right), \sin\left(x_1/2\right), \sin\left(x_2\right), \sin\left(x_2/2\right),\cos\left(x_1\right),\cos\left(x_1/2\right),\cos\left(x_2\right),\cos\left(x_2/2\right)\right)'$ and $trig_2=\left(trig_1',\sin\left(x_1^2\right),\cos\left(x_1^2\right),\sin\left(x_2^2\right),\cos\left(x_2^2\right)\right)'$, and the B-spline bases of fourth and seventh order
	($p=9,p=14$),  We also set $\gamma _{2}=0.3$, $\lambda _{1}=0.3$
	and $\gamma _{3}=0.4$; the value $c=0,3,6$ indicates the null hypothesis and
	the local alternatives. The spatial weight matrices are generated using
	LeSage's code make\_neighborsw from http://www.spatial-econometrics.com/,
	where the row-normalized sparse matrices are generated by choosing a specific
	number of the closest locations from randomly generated coordinates and we
	set the number of neighbors to be $n/20$. We employ 100 bootstrap
	replications in each of 500 Monte Carlo replications except for the
	SARARMA(1,0,1) design with $n=200$, where we set 50 bootstrap replications
	in view of the computation time. We report the rejection frequencies
	of tests based on bootstrap critical values in the main text, while tests
	based on asymptotic critical values are reported in the online supplement.
	
	Tables \ref{table:newsims1}-\ref{table:newsims4} report the empirical rejection frequencies using the
	bootstrap test statistics $\fancyt^\ast$ (Tables \ref{table:newsims1}, \ref{table:newsims3}) and $\fancyt^{a \ast}$ (Tables \ref{table:newsims2}, \ref{table:newsims4}), when nominal levels are given by 1\%, 5\% and
	10\%. To see how the choice of $p$ and the basis functions affect
	small sample outcomes, we report two sets of results for each basis function family:
	the first row for each value of $c$ is from the smaller $p$ ($p=9$ or $10$), while the second row is from
	the larger $p$ ($p=14$ or $15$). We summarize some
	important findings. First, we see that for most DGPs, our bootstrap test is
	closer to the nominal level than the asymptotic test (reported in
	the online supplement) although the sizes of both types of tests improve
	generally as the sample size increases. Second, both bootstrap and
	asymptotic tests are powerful in detecting any deviations from linearity in
	the local alternatives. The patterns are similar across all cases:
	the bootstrap generally affords better size control, albeit not always.
	
	All three types of bases give qualitatively similar results, but we
	note that $\fancyt^\ast=\fancyt^{ \ast a }$ when using polynomial series under
	the SARARMA(0,1,0) model, as observed in \cite{Hong1995}. When using
	trigonometric and B-spline series, tests based on these two
	statistics give slightly different rejection rates. However, under the
	SARARMA(1,0,1) model, all series give quantitatively different results, as
	illustrated in Tables \ref{table:newsims3} and \ref{table:newsims4}. When using B-spline bases, $p=14$ does not perform well compared to $p=9$. In the
	other cases, both choices of $p$ work well.

	\subsection{Nonparametric error spatial structure}
	
	Now we examine finite sample performance in the setting of Section \ref{sec:nonpar_ext}. The
	three DGPs of $\theta (x)$\ are the same as the parametric setting but we
	generate the $n\times n$ matrix $W^*$ as $w^*_{ij}=\Phi (-d_{ij}) I(c_{ij}<0.05)$ if $%
	i\neq j$, and $w^*_{ii}=0$, where $\Phi (\cdot )$ is the standard
	normal cdf, $d_{ij}\sim$iid $U[-3,3]$, and $c_{ij}\sim$iid $U[0,1]$. From this construction, we ensure
	that $W^*$ is sparse with no more than $5\%$ elements being nonzero. Then, $y$
	is generated from $
	y=\theta (x)+u,\text{ }u=Wu+\xi ,
	$
	where $\xi\sim N(0,I_n)$ and $W=W^*/{1.2\overline{\varphi }\left(W^*\right)}$, ensuring the
	existence of $(I-W)^{-1}$. In estimation,
	we know the distance $d_{ij}$ and the indicator $I(c_{ij}<0.05)$, but we do not
	know the functional form of $w_{ij}$, so we approximate elements
	in $W$ by $
	\widehat{w}_{ij}=\sum_{l=0}^{r}a_{l}d_{ij}^{l}I(c_{ij}<0.05)\text{ if }%
	i\neq j\text{; }\widehat{w}_{ii}=0.$
	
	Table \ref{table:sims5} reports the rejection rates using 500 Monte Carlo simulation at the
	5\% asymptotic level 1.645 using polynomial bases with $r=2,3,4, 5$ and $p=10,15, 20$. We take $n=150, 300, 500, 600, 700$, larger sample sizes than earlier because two nonparametric functions must be estimated in this spatial setting. The two largest bandwidths ($r=5, p=20$) are only employed for the largest sample size $n=700$. We observe a clear pattern of rejection rates approaching the theoretical level as sample size increases. Power improves as $c$ increases for all designs and is non-trivial in all cases even for $c=3$. Sizes are acceptable for $n=500$, particularly when $p=15$. Size performance improves further as $n=600$, indicating asymptotic stability. Note that with two diverging bandwidths ($p$ and $r$), we expect sizes to improve in a diagonal pattern going from top left corner to bottom right corner in Table \ref{table:sims5}. This is indeed the case. For $n=700$, we observe that the pairs $(r,p)=(5,15),(5,20)$ deliver acceptable sizes.

	\section{Empirical applications}\label{sec:apps}
	
	{ In this section, we illustrate the specification test presented in
		previous sections using several empirical examples. }
	\subsection{\protect  Conflict alliances}
	
	This example is based on a study of how a network of military alliances and
	enmities affects the intensity of a conflict, conducted by \cite{konig2017networks}.
	They stress that understanding the role of informal networks of military alliances and
	enmities is important not only for predicting outcomes, but also for
	designing and implementing policies to contain or put an end to violence.
	\cite{konig2017networks} obtain a closed-form characterization of the Nash
	equilibrium and perform an empirical analysis using data on the Second Congo
	War, a conflict that involves many groups in a complex network of informal
	alliances and rivalries.
	
	To study the fighting effort of each group the authors use a panel data model with individual fixed
	effects, where key regressors include total fighting effort of allies and enemies. They further correct the
	potential spatial correlation in the error term by using a spatial heteroskedasticity and autocorrelation robust standard error. We use their data and the main structure of the
	specification and build a cross-sectional SAR(2) model
	with two weight matrices, $W^A$ ($W^A_{ij}=1$ if group $i$
	and $j$ are allies, and $W^A_{ij}=0$ otherwise) and $W^E$ ($W^E_{ij}=1$ if
	group $i$ and $j$ are enemies, and $W^E_{ij}=0$ otherwise):
	\begin{equation*}
		y=\lambda _{1}W^A y+\lambda _{2}W^E y+\mathbf{1}%
		_{n}\beta _{0}+X\beta +u,
	\end{equation*}%
	where $y$ is a vector of fighting efforts of each group and $X$ includes the current rainfall, rainfall from the last year, and
	their squares.\footnote{We follow the analysis in the original paper and do not row normalize. This is because the economic content of the weight matrices is defined by total fights of allies or enemies.} To consider the spatial correlation in the error term, we
	consider both the Error SARMA(1,0) and Error SARMA(0,1) structures. For these, we employ a spatial weight matrix $W^d$,
	based on the inverse distance between group locations and set to be 0 after 150 km, following \cite{konig2017networks}. The idea is that geographical spatial correlation dies out as groups become further apart. We also report results using a nonparametric estimator of the spatial weights, as described in Section \ref{sec:nonpar_ext} and studied in simulations in Section \ref{sec:mc}. For the nonparametric estimator we take $r=2$.
	
	In the original dataset, there are 80 groups, but groups 62 and 63 have the
	same variables and the same locations, so we drop one group and end up with
	a sample of 79 groups. We use data from 1998 as an example and further use
	the pooled data from all years as a robustness check. $H_{0}$ stands for restricted model
	where the linear functional form of the regression is imposed, while
	$H_{1}$ stands for the unrestricted
	model where we use basis functions comprising of power series with $p=10$. In
	all our specifications, the test statistics are negative, so we cannot
	reject the null hypothesis that the model is correctly specified. As Table \ref{table:conflict} indicates, this failure to reject the null persists when we use pooled data from 13 years, yielding 1027 observations. Thus we conclude that a linear
	specification is not inappropriate for this setting. One possible reason is that the original regression, though linear, has already included the squared terms of the rainfall as regressors. This finding is robust to using the bootstrap tests of Section \ref{sec:bootstrap}, which generally yield smaller p-values but unchanged conclusions.
	
	\subsection{\protect  Innovation spillovers}
	
	This example is based on the study of the impact of R\&D on growth from \cite{Bloom2013}.
	They develop a general framework incorporating two types of spillovers: a
	positive effect from technology (knowledge) spillovers and a negative
	`business stealing' effect from product market rivals. They implement this
	model using panel data on U.S. firms.
	
	We consider the Productivity Equation in \cite{Bloom2013}:
	\begin{equation}\label{prod_eqn}
		\ln y=\varphi _{1}\ln (R\&D)+\varphi _{2}\ln (Sptec)+\varphi _{3}\ln
		(Spsic)+\varphi _{4}X+error,
	\end{equation}%
	where $y$ is a vector of sales, $R\&D$ is a vector of R\&D stocks, and regressors in $X$ include the
	log of capital ($Capital$), log of labor ($Labor$), $R\&D$, a dummy for missing values in $R\&D$, a price
	index, and two spillover terms constructed as the log of $W_{SIC} R\&D$ ($Spsic$)
	and the log of $W_{TEC} R\&D$ ($Sptec$), where $W_{SIC}$ measures the product
	market proximity and $W_{TEC}$ measures the technological proximity. Specifically, they define
	\[W_{SIC,ij}={S_{i}S_{j}^{\prime }}/{(S_{i}S_{i}^{\prime
		})^{1/2}(S_{j}S_{j}^{\prime })^{1/2}},W_{TEC,ij}={
		T_{i}T_{j}^{\prime }}/{(T_{i}T_{i}^{\prime })^{1/2}(T_{j}T_{j}^{\prime
		})^{1/2}},\]
	where  $S_{i}=(S_{i1},S_{i2},\ldots,S_{i597})'$, with $S_{ik}$
	being the share of patents of firm $i$ in the four digit industry $k$ and
	$T_{i}=(T_{i1},T_{i2},\ldots,T_{i426})'$, with $T_{i\tau }$
	being the share of patents of firm $i$ in technology class $\tau $.  Focusing on a cross-sectional analysis, we use observations from
	the year 2000 and obtain a sample size of
	577. Both weight matrices are row normalized.

	The column FE of Table \ref{table:R&D1} is from Table 5 of \cite{Bloom2013} based on
	their panel fixed effects estimation and we use it as a baseline for comparison. This table reports results for SARARMA(0,1,0) models using $W_{SIC}$ and $W_{TEC}$ separately. We use
	both $W_{SIC}$ and $W_{TEC}$ simultaneously in SARARMA(0,2,0), SARARMA(0,2,0), and Error MESS(2)
	models, reported in Table \ref{table:R&D2}. In all of these specifications, the test statistics are larger than
	1.645, so we reject the null hypothesis of the linear specification. This rejection also persists with the bootstrap tests, albeit the p-values go up compared to the asymptotic ones. However, we can say even more as our estimation also sheds light on spatial effects in the disturbances in (\ref{prod_eqn}). As before $H_{0}$ imposes linear functional form of the regressors, while
	$H_{1}$ uses the nonparametric series estimate employing power series with $p=10$. Regardless of the specification of the regression function, the disturbances suggest a strong spatial effect as the coefficients on $W_{TEC}$ and $W_{SIC}$ are large in magnitude.

	\subsection{\protect  Economic growth}
	
	{ The final example is based on the study of economic growth rate in \cite{Ertur2007}. Knowledge accumulated in one area might depend on knowledge
		accumulated in other areas, especially in its neighborhoods, implying the
		possible existence of spatial spillover effects. These questions are of interest to both economists as well as regional scientists. For example, \cite{Autant-Bernard2011} examine spatial spillovers associated with research
		expenditures for French
		regions, while \cite{Ho2013} examine the
		international spillover of economic growth through bilateral trade amongst OECD countries, Cuaresma and
		Feldkircher (2013) study spatially correlated growth spillovers in the
		income convergence process of Europe, and \cite{Evans2014}
		study the spatial dynamics of growth and convergence in Korean regional
		incomes. }
	
	{ In this section, we want to test the linear SAR
		model specification in \cite{Ertur2007}. Their dataset covers a sample
		of 91 countries over the period 1960-1995, originally from \cite{Heston2002}, obtained from the Penn World Tables (PWT version 6.1).
		The variables in use include per worker income in 1960 ($y60$) and 1995 ($%
		y95 $), average rate of growth between 1960 and 1995 $(gy)$, average
		investment rate of this period ($s$), and average rate of growth of
		working-age population ($n_{p}$).
		
		{ \cite{Ertur2007} consider
			the model%
			\begin{equation}\label{Ertur_model}
				y=\lambda Wy+X\beta
				+WX\theta +\varepsilon,
			\end{equation}%
			where the dependent variable is log real income per worker $\ln (y95)$,
			elements of the explanatory variable $X=(x_1',x_2')$ include log investment rate $%
			\ln (s)=x_1$ and log physical capital effective rate of depreciation\textbf{\ }$%
			\ln (n_{p}+0.05)=x_2$, with corresponding subscripted coefficients $\beta_1,\beta_2,\theta_1,\theta_2$. A restricted regression based on the joint constraints
			$\beta _{1}=-\beta _{2}$\ and $\theta _{1}=-\theta _{2}$ (these constraints are implied by economic theory) is also considered
			in \cite{Ertur2007}. The model (\ref{Ertur_model}) has regressors $(X,WX)$ and iid errors, so the test derived in Section \ref{sec:SAR_ext}
			can be directly applied here. Denoting by $d_{ij}$ the great-circle distance between the capital cities of countries $i$ and $j$, one construction of $W$ takes $w_{ij}=d_{ij}^{-2}$ while the other takes $w_{ij}=e^{-2d_{ij}}$, following \cite{Ertur2007}.
			
			Table \ref{table:growth} presents the estimation and testing results based on using linear and quadratic power series basis functions with $p=10$ and a sample size of $n=91$. We impose additive structure in our estimation to at least alleviate the curse of dimensionality, always a concern in nonparametric estimation. We also use only linear and quadratic basis functions to reduce the number of terms for series estimation. 	
			
			We cannot reject linearity of the regression function for the unrestricted model. On the other hand, linearity is rejected for the restricted model, which is the preferred specification of \cite{Ertur2007}, with $w_{ij}=e^{-2d_{ij}}$. Thus, not only can we conclude that the specification of the model is under suspicion we can also infer this is due to constraints from economic theory. The findings are supported by the bootstrap tests of Section \ref{sec:bootstrap}.
			
			\section{Conclusion}\label{sec:conc}
			
			This paper justifies a specification test for the regression function in a model where data are spatially dependent. The test is based on a nonparametric series approximation and is consistent. The paper also allows for some robustness in error spatial dependence by permitting this to be a nonparametric function of an underlying economic distance. On the other hand, our Section \ref{sec:SAR_ext} imposes correct specification of the spatial weight matrices $W_j$ in the SAR context, while \cite{Sun2020} allows these to be nonparametric functions as well. Thus our work acts as a complement to existing results in the literature and future work might combine both aspects.
			
			\bigskip \pagebreak
			
			\small{
				\begin{table}
					\centering{
						\begin{tabular}{cccccccccccc}
							\hline\hline
							$\fancyt^{\ast }$ & \multicolumn{11}{c}{SARARMA(0,1,0)} \\
							\hline
							&  & \textbf{PS} &  &  &  & \textbf{Trig} &  &  &  & \textbf{B-s} &
							\\
							\hline
							{\small $n=60$} & {\small 0.01} & {\small 0.05} & {\small 0.10} &  & {\small %
								0.01} & {\small 0.05} & {\small 0.10} &  & {\small 0.01} & {\small 0.05} &
							{\small 0.1} \\
							\hline
							${\small c=0}$ & ${\small 0.008}$ & ${\small 0.032}$ & ${\small 0.084}$ &  &
							${\small 0.004}$ & ${\small 0.04}$ & ${\small 0.092}$ &  & ${\small 0.006}$
							& ${\small 0.048}$ & ${\small 0.104}$ \\
							& ${\small 0.004}$ & ${\small 0.038}$ & ${\small 0.096}$ &  & ${\small 0.004}
							$ & ${\small 0.038}$ & ${\small 0.094}$ &  & ${\small 0.006}$ & ${\small %
								0.034}$ & ${\small 0.098}$ \\
							${\small c=3}$ & ${\small 0.036}$ & ${\small 0.154}$ & ${\small 0.296}$ &  &
							${\small 0.092}$ & ${\small 0.276}$ & ${\small 0.39}$ &  & ${\small 0.098}$
							& ${\small 0.292}$ & ${\small 0.470}$ \\
							& ${\small 0.154}$ & ${\small 0.414}$ & ${\small 0.62}$ &  & ${\small 0.056}$
							& ${\small 0.22}$ & ${\small 0.374}$ &  & ${\small 0.036}$ & ${\small 0.150}$
							& ${\small 0.292}$ \\
							${\small c=6}$ & ${\small 0.22}$ & ${\small 0.544}$ & ${\small 0.748}$ &  & $%
							{\small 0.454}$ & ${\small 0.794}$ & ${\small 0.908}$ &  & ${\small 0.432}$
							& ${\small 0.814}$ & ${\small 0.938}$ \\
							& ${\small 0.844}$ & ${\small 0.992}$ & ${\small 1}$ &  & ${\small 0.314}$ &
							${\small 0.714}$ & ${\small 0.872}$ &  & ${\small 0.174}$ & ${\small 0.542}$
							& ${\small 0.732}$ \\
							\hline
							${\small n=100}$ &  &  &  &  &  &  &  &  &  &  &  \\
							\hline
							${\small c=0}$ & ${\small 0.006}$ & ${\small 0.044}$ & ${\small 0.098}$ &  &
							${\small 0.002}$ & ${\small 0.04}$ & ${\small 0.09}$ &  & ${\small 0.008}$ &
							${\small 0.038}$ & ${\small 0.110}$ \\
							& ${\small 0.012}$ & ${\small 0.046}$ & ${\small 0.096}$ &  & ${\small 0.006}
							$ & ${\small 0.036}$ & ${\small 0.102}$ &  & ${\small 0.01}$ & ${\small 0.056%
							}$ & ${\small 0.108}$ \\
							${\small c=3}$ & ${\small 0.294}$ & ${\small 0.578}$ & ${\small 0.72}$ &  & $%
							{\small 0.214}$ & ${\small 0.508}$ & ${\small 0.626}$ &  & ${\small 0.272}$
							& ${\small 0.572}$ & ${\small 0.712}$ \\
							& ${\small 0.37}$ & ${\small 0.662}$ & ${\small 0.824}$ &  & ${\small 0.194}$
							& ${\small 0.45}$ & ${\small 0.632}$ &  & ${\small 0.188}$ & ${\small 0.46}$
							& ${\small 0.63}$ \\
							${\small c=6}$ & ${\small 0.95}$ & ${\small 0.99}$ & ${\small 0.996}$ &  & $%
							{\small 0.902}$ & ${\small 0.99}$ & ${\small 0.998}$ &  & ${\small 0.922}$ &
							${\small 0.994}$ & ${\small 1}$ \\
							& ${\small 0.992}$ & ${\small 0.998}$ & ${\small 1}$ &  & ${\small 0.856}$ &
							${\small 0.988}$ & ${\small 1}$ &  & ${\small 0.852}$ & ${\small 0.98}$ & $%
							{\small 0.998}$ \\
							\hline
							$\small n=200$ &  &  &  &  &  &  &  &  &  &  &  \\
							\hline
							${\small c=0}$ & ${\small 0.006}$ & ${\small 0.038}$ & ${\small 0.104}$ &  &
							${\small 0.008}$ & ${\small 0.042}$ & ${\small 0.112}$ &  & ${\small 0.024}$
							& ${\small 0.074}$ & ${\small 0.132}$ \\
							& ${\small 0.006}$ & ${\small 0.048}$ & ${\small 0.088}$ &  & ${\small 0.016}
							$ & ${\small 0.038}$ & ${\small 0.082}$ &  & ${\small 0.022}$ & ${\small %
								0.074}$ & ${\small 0.144}$ \\
							${\small c=3}$ & ${\small 0.178}$ & ${\small 0.402}$ & ${\small 0.55}$ &  & $%
							{\small 0.162}$ & ${\small 0.374}$ & ${\small 0.532}$ &  & ${\small 0.314}$
							& ${\small 0.516}$ & ${\small 0.654}$ \\
							& ${\small 0.282}$ & ${\small 0.56}$ & ${\small 0.694}$ &  & ${\small 0.136}$
							& ${\small 0.346}$ & ${\small 0.468}$ &  & ${\small 0.19}$ & ${\small 0.37}$
							& ${\small 0.542}$ \\
							${\small c=6}$ & ${\small 0.846}$ & ${\small 0.968}$ & ${\small 0.984}$ &  &
							${\small 0.796}$ & ${\small 0.95}$ & ${\small 0.98}$ &  & ${\small 0.89}$ & $%
							{\small 0.976}$ & ${\small 0.986}$ \\
							& ${\small 0.982}$ & ${\small 0.998}$ & ${\small 1}$ &  & ${\small 0.776}$ &
							${\small 0.934}$ & ${\small 0.974}$ &  & ${\small 0.852}$ & ${\small 0.946}$
							& ${\small 0.982}$\\
							\hline\hline
						\end{tabular}
						\caption{Rejection probabilities of SARARMA(0,1,0) using bootstrap
							test $\fancyt^{\ast }$ at 1, 5, 10\% levels, power series (\textbf{PS}), trigonometric (\textbf{Trig}) and B-spline (\textbf{B-s}) bases.}\label{table:newsims1}}
				\end{table}

				\begin{table}
					\begin{tabular}{cccccccccccc}
						\hline\hline
						$\fancyt^{a \ast }$ & \multicolumn{11}{c}{SARARMA(0,1,0)} \\
						\hline
						&  & \textbf{PS} &  &  &  & \textbf{Trig} &  &  &  & \textbf{B-s} &
						\\
						\hline
						{\small $n=60$} & {\small 0.01} & {\small 0.05} & {\small 0.10} &  & {\small %
							0.01} & {\small 0.05} & {\small 0.10} &  & {\small 0.01} & {\small 0.05} &
						{\small 0.1} \\
						\hline
						${\small c=0}$ & ${\small 0.008}$ & ${\small 0.032}$ & ${\small 0.084}$ &  &
						${\small 0.004}$ & ${\small 0.04}$ & ${\small 0.092}$ &  & ${\small 0.01}$ &
						${\small 0.07}$ & ${\small 0.132}$ \\
						& ${\small 0.004}$ & ${\small 0.038}$ & ${\small 0.096}$ &  & ${\small 0.004}
						$ & ${\small 0.038}$ & ${\small 0.094}$ &  & ${\small 0.004}$ & ${\small %
							0.038}$ & ${\small 0.096}$ \\
						${\small c=3}$ & ${\small 0.036}$ & ${\small 0.154}$ & ${\small 0.296}$ &  &
						${\small 0.09}$ & ${\small 0.274}$ & ${\small 0.384}$ &  & ${\small 0.164}$
						& ${\small 0.376}$ & ${\small 0.558}$ \\
						& ${\small 0.154}$ & ${\small 0.414}$ & ${\small 0.62}$ &  & ${\small 0.056}$
						& ${\small 0.22}$ & ${\small 0.376}$ &  & ${\small 0.036}$ & ${\small 0.152}$
						& ${\small 0.288}$ \\
						${\small c=6}$ & ${\small 0.22}$ & ${\small 0.544}$ & ${\small 0.748}$ &  & $%
						{\small 0.444}$ & ${\small 0.794}$ & ${\small 0.906}$ &  & ${\small 0.56}$ &
						${\small 0.892}$ & ${\small 0.956}$ \\
						& ${\small 0.844}$ & ${\small 0.992}$ & ${\small 1}$ &  & ${\small 0.312}$ &
						${\small 0.714}$ & ${\small 0.87}$ &  & ${\small 0.174}$ & ${\small 0.532}$
						& ${\small 0.732}$ \\
						\hline
						${\small n=100}$ &  &  &  &  &  &  &  &  &  &  &  \\
						\hline
						${\small c=0}$ & ${\small 0.006}$ & ${\small 0.044}$ & ${\small 0.098}$ &  &
						${\small 0.004}$ & ${\small 0.038}$ & ${\small 0.092}$ &  & ${\small 0.012}$
						& ${\small 0.048}$ & ${\small 0.112}$ \\
						& ${\small 0.012}$ & ${\small 0.046}$ & ${\small 0.096}$ &  & ${\small 0.006}
						$ & ${\small 0.036}$ & ${\small 0.106}$ &  & ${\small 0.01}$ & ${\small 0.056%
						}$ & ${\small 0.106}$ \\
						${\small c=3}$ & ${\small 0.294}$ & ${\small 0.578}$ & ${\small 0.72}$ &  & $%
						{\small 0.214}$ & ${\small 0.504}$ & ${\small 0.63}$ &  & ${\small 0.28}$ & $%
						{\small 0.564}$ & ${\small 0.72}$ \\
						& ${\small 0.37}$ & ${\small 0.662}$ & ${\small 0.824}$ &  & ${\small 0.194}$
						& ${\small 0.45}$ & ${\small 0.632}$ &  & ${\small 0.196}$ & ${\small 0.466}$
						& ${\small 0.64}$ \\
						${\small c=6}$ & ${\small 0.95}$ & ${\small 0.99}$ & ${\small 0.996}$ &  & $%
						{\small 0.900}$ & ${\small 0.99}$ & ${\small 0.998}$ &  & ${\small 0.932}$ &
						${\small 0.992}$ & ${\small 1}$ \\
						& ${\small 0.992}$ & ${\small 0.998}$ & ${\small 1}$ &  & ${\small 0.856}$ &
						${\small 0.988}$ & ${\small 1}$ &  & ${\small 0.86}$ & ${\small 0.984}$ & $%
						{\small 0.998}$ \\
						\hline
						$\small n=200$ &  &  &  &  &  &  &  &  &  &  &  \\
						\hline
						${\small c=0}$ & ${\small 0.006}$ & ${\small 0.038}$ & ${\small 0.104}$ &  &
						${\small 0.012}$ & ${\small 0.046}$ & ${\small 0.114}$ &  & ${\small 0.014}$
						& ${\small 0.048}$ & ${\small 0.132}$ \\
						& ${\small 0.006}$ & ${\small 0.048}$ & ${\small 0.088}$ &  & ${\small 0.016}
						$ & ${\small 0.042}$ & ${\small 0.08}$ &  & ${\small 0.022}$ & ${\small 0.07}
						$ & ${\small 0.14}$ \\
						${\small c=3}$ & ${\small 0.178}$ & ${\small 0.402}$ & ${\small 0.55}$ &  & $%
						{\small 0.162}$ & ${\small 0.38}$ & ${\small 0.524}$ &  & ${\small 0.282}$ &
						${\small 0.476}$ & ${\small 0.608}$ \\
						& ${\small 0.282}$ & ${\small 0.56}$ & ${\small 0.694}$ &  & ${\small 0.134}$
						& ${\small 0.35}$ & ${\small 0.466}$ &  & ${\small 0.198}$ & ${\small 0.37}$
						& ${\small 0.514}$ \\
						${\small c=6}$ & ${\small 0.846}$ & ${\small 0.968}$ & ${\small 0.984}$ &  &
						${\small 0.802}$ & ${\small 0.952}$ & ${\small 0.978}$ &  & ${\small 0.848}$
						& ${\small 0.95}$ & ${\small 0.982}$ \\
						& ${\small 0.982}$ & ${\small 0.998}$ & ${\small 1}$ &  & ${\small 0.774}$ &
						${\small 0.934}$ & ${\small 0.972}$ &  & ${\small 0.84}$ & ${\small 0.932}$
						& ${\small 0.97}$\\
						\hline\hline
					\end{tabular}
					\caption{Rejection probabilities of {SARARMA(0,1,0)} using bootstrap
						test $\fancyt^{a \ast}$ at 1, 5, 10\% levels, power series (\textbf{PS}), trigonometric (\textbf{Trig}) and B-spline (\textbf{B-s}) bases.}\label{table:newsims2}
				\end{table}
				
				\begin{table}
					\begin{tabular}{cccccccccccc}
						\hline\hline
						$\fancyt^{\ast }$ & \multicolumn{11}{c}{SARARMA(1,0,1)} \\
						\hline
						&  & \textbf{PS} &  &  &  & \textbf{Trig} &  &  &  & \textbf{B-s} &
						\\
						\hline
						{\small $n=60$} & {\small 0.01} & {\small 0.05} & {\small 0.10} &  & {\small %
							0.01} & {\small 0.05} & {\small 0.10} &  & {\small 0.01} & {\small 0.05} &
						{\small 0.1} \\
						\hline
						${\small c=0}$ & ${\small 0.006}$ & ${\small 0.054}$ & ${\small 0.08}$ &  & $%
						{\small 0.012}$ & ${\small 0.062}$ & ${\small 0.106}$ &  & ${\small 0.016}$
						& ${\small 0.044}$ & ${\small 0.086}$ \\
						& ${\small 0.016}$ & ${\small 0.062}$ & ${\small 0.118}$ &  & ${\small 0.026}
						$ & ${\small 0.09}$ & ${\small 0.138}$ &  & ${\small 0.016}$ & ${\small 0.048%
						}$ & ${\small 0.088}$ \\
						${\small c=3}$ & ${\small 0.08}$ & ${\small 0.264}$ & ${\small 0.402}$ &  & $%
						{\small 0.082}$ & ${\small 0.256}$ & ${\small 0.406}$ &  & ${\small 0.08}$ &
						${\small 0.288}$ & ${\small 0.475}$ \\
						& ${\small 0.132}$ & ${\small 0.41}$ & ${\small 0.578}$ &  & ${\small 0.096}$
						& ${\small 0.222}$ & ${\small 0.354}$ &  & ${\small 0.048}$ & ${\small 0.192}
						$ & ${\small 0.282}$ \\
						${\small c=6}$ & ${\small 0.266}$ & ${\small 0.588}$ & ${\small 0.748}$ &  &
						${\small 0.266}$ & ${\small 0.616}$ & ${\small 0.782}$ &  & ${\small 0.218}$
						& ${\small 0.604}$ & ${\small 0.772}$ \\
						& ${\small 0.444}$ & ${\small 0.804}$ & ${\small 0.894}$ &  & ${\small 0.204}
						$ & ${\small 0.474}$ & ${\small 0.658}$ &  & ${\small 0.198}$ & ${\small %
							0.496}$ & ${\small 0.612}$ \\
						\hline
						${\small n=100}$ &  &  &  &  &  &  &  &  &  &  &  \\
						\hline
						${\small c=0}$ & ${\small 0.006}$ & ${\small 0.054}$ & ${\small 0.116}$ &  &
						${\small 0.012}$ & ${\small 0.046}$ & ${\small 0.114}$ &  & ${\small 0.014}$
						& ${\small 0.042}$ & ${\small 0.09}$ \\
						& ${\small 0.02}$ & ${\small 0.056}$ & ${\small 0.112}$ &  & ${\small 0.012}$
						& ${\small 0.044}$ & ${\small 0.088}$ &  & ${\small 0.034}$ & ${\small 0.058}
						$ & ${\small 0.118}$ \\
						${\small c=3}$ & ${\small 0.134}$ & ${\small 0.366}$ & ${\small 0.496}$ &  &
						${\small 0.132}$ & ${\small 0.346}$ & ${\small 0.514}$ &  & ${\small 0.162}$
						& ${\small 0.46}$ & ${\small 0.59}$ \\
						& ${\small 0.222}$ & ${\small 0.556}$ & ${\small 0.732}$ &  & ${\small 0.242}
						$ & ${\small 0.542}$ & ${\small 0.698}$ &  & ${\small 0.08}$ & ${\small 0.234%
						}$ & ${\small 0.372}$ \\
						${\small c=6}$ & ${\small 0.566}$ & ${\small 0.832}$ & ${\small 0.916}$ &  &
						${\small 0.59}$ & ${\small 0.888}$ & ${\small 0.96}$ &  & ${\small 0.548}$ &
						${\small 0.898}$ & ${\small 0.952}$ \\
						& ${\small 0.732}$ & ${\small 0.964}$ & ${\small 0.986}$ &  & ${\small 0.476}
						$ & ${\small 0.846}$ & ${\small 0.918}$ &  & ${\small 0.432}$ & ${\small %
							0.796}$ & ${\small 0.874}$ \\
						\hline
						${\small n=200}$ &  &  &  &  &  &  &  &  &  &  &  \\
						\hline
						${\small c=0}$ & ${\small 0.04}$ & ${\small 0.086}$ & ${\small 0.11}$ &  & $%
						{\small 0.026}$ & ${\small 0.076}$ & ${\small 0.108}$ &  & ${\small 0.02}$ &
						${\small 0.06}$ & ${\small 0.09}$ \\
						& ${\small 0.03}$ & ${\small 0.078}$ & ${\small 0.114}$ &  & ${\small 0.032}$
						& ${\small 0.074}$ & ${\small 0.118}$ &  & ${\small 0.038}$ & ${\small 0.086}
						$ & ${\small 0.112}$ \\
						${\small c=3}$ & ${\small 0.186}$ & ${\small 0.4}$ & ${\small 0.524}$ &  & $%
						{\small 0.242}$ & ${\small 0.432}$ & ${\small 0.526}$ &  & ${\small 0.29}$ &
						${\small 0.516}$ & ${\small 0.626}$ \\
						& ${\small 0.402}$ & ${\small 0.636}$ & ${\small 0.754}$ &  & ${\small 0.244}
						$ & ${\small 0.42}$ & ${\small 0.542}$ &  & ${\small 0.184}$ & ${\small 0.36}
						$ & ${\small 0.458}$ \\
						${\small c=6}$ & ${\small 0.718}$ & ${\small 0.904}$ & ${\small 0.962}$ &  &
						${\small 0.78}$ & ${\small 0.942}$ & ${\small 0.982}$ &  & ${\small 0.73}$ &
						${\small 0.948}$ & ${\small 0.978}$ \\
						& ${\small 0.872}$ & ${\small 0.98}$ & ${\small 0.998}$ &  & ${\small 0.794}$
						& ${\small 0.948}$ & ${\small 0.98}$ &  & ${\small 0.772}$ & ${\small 0.914}$
						& ${\small 0.94}$\\
						\hline\hline
					\end{tabular}
					\caption{Rejection probabilities of {SARARMA(1,0,1)} using bootstrap
						test $\fancyt^{\ast }$ at 1, 5, 10\% levels, power series (\textbf{PS}), trigonometric (\textbf{Trig}) and B-spline (\textbf{B-s}) bases.}\label{table:newsims3}
				\end{table}
				
				\begin{table}
					\begin{tabular}{cccccccccccc}
						\hline\hline
						$\fancyt^{a \ast }$ & \multicolumn{11}{c}{SARARMA(1,0,1)} \\
						\hline
						&  & \textbf{PS} &  &  &  & \textbf{Trig} &  &  &  & \textbf{B-s} &
						\\
						\hline
						{\small $n=60$} & {\small 0.01} & {\small 0.05} & {\small 0.10} &  & {\small %
							0.01} & {\small 0.05} & {\small 0.10} &  & {\small 0.01} & {\small 0.05} &
						{\small 0.1} \\
						\hline
						${\small c=0}$ & ${\small 0.006}$ & ${\small 0.052}$ & ${\small 0.084}$ &  &
						${\small 0.014}$ & ${\small 0.064}$ & ${\small 0.096}$ &  & ${\small 0.012}$
						& ${\small 0.044}$ & ${\small 0.104}$ \\
						& ${\small 0.012}$ & ${\small 0.068}$ & ${\small 0.114}$ &  & ${\small 0.024}
						$ & ${\small 0.088}$ & ${\small 0.13}$ &  & ${\small 0.018}$ & ${\small 0.038%
						}$ & ${\small 0.068}$ \\
						${\small c=3}$ & ${\small 0.092}$ & ${\small 0.27}$ & ${\small 0.396}$ &  & $%
						{\small 0.08}$ & ${\small 0.25}$ & ${\small 0.406}$ &  & ${\small 0.118}$ & $%
						{\small 0.382}$ & ${\small 0.56}$ \\
						& ${\small 0.164}$ & ${\small 0.408}$ & ${\small 0.596}$ &  & ${\small 0.102}
						$ & ${\small 0.242}$ & ${\small 0.37}$ &  & ${\small 0.046}$ & ${\small 0.15}
						$ & ${\small 0.23}$ \\
						${\small c=6}$ & ${\small 0.268}$ & ${\small 0.596}$ & ${\small 0.752}$ &  &
						${\small 0.248}$ & ${\small 0.61}$ & ${\small 0.792}$ &  & ${\small 0.23}$ &
						${\small 0.56}$ & ${\small 0.808}$ \\
						& ${\small 0.518}$ & ${\small 0.824}$ & ${\small 0.9}$ &  & ${\small 0.206}$
						& ${\small 0.484}$ & ${\small 0.658}$ &  & ${\small 0.176}$ & ${\small 0.43}$
						& ${\small 0.56}$ \\
						\hline
						${\small n=100}$ &  &  &  &  &  &  &  &  &  &  &  \\
						\hline
						${\small c=0}$ & ${\small 0.008}$ & ${\small 0.058}$ & ${\small 0.122}$ &  &
						${\small 0.01}$ & ${\small 0.046}$ & ${\small 0.116}$ &  & ${\small 0.004}$
						& ${\small 0.04}$ & ${\small 0.82}$ \\
						& ${\small 0.024}$ & ${\small 0.062}$ & ${\small 0.118}$ &  & ${\small 0.014}
						$ & ${\small 0.044}$ & ${\small 0.096}$ &  & ${\small 0.028}$ & ${\small %
							0.056}$ & ${\small 0.074}$ \\
						${\small c=3}$ & ${\small 0.14}$ & ${\small 0.36}$ & ${\small 0.494}$ &  & $%
						{\small 0.122}$ & ${\small 0.354}$ & ${\small 0.52}$ &  & ${\small 0.186}$ &
						${\small 0.4}$ & ${\small 0.524}$ \\
						& ${\small 0.252}$ & ${\small 0.566}$ & ${\small 0.73}$ &  & ${\small 0.272}$
						& ${\small 0.568}$ & ${\small 0.696}$ &  & ${\small 0.04}$ & ${\small 0.148}$
						& ${\small 0.214}$ \\
						${\small c=6}$ & ${\small 0.536}$ & ${\small 0.818}$ & ${\small 0.914}$ &  &
						${\small 0.554}$ & ${\small 0.884}$ & ${\small 0.948}$ &  & ${\small 0.58}$
						& ${\small 0.914}$ & ${\small 0.95}$ \\
						& ${\small 0.786}$ & ${\small 0.958}$ & ${\small 0.974}$ &  & ${\small 0.478}
						$ & ${\small 0.834}$ & ${\small 0.916}$ &  & ${\small 0.328}$ & ${\small %
							0.586}$ & ${\small 0.678}$ \\
						\hline
						${\small n=200}$ &  &  &  &  &  &  &  &  &  &  &  \\
						\hline
						${\small c=0}$ & ${\small 0.04}$ & ${\small 0.08}$ & ${\small 0.116}$ &  & $%
						{\small 0.03}$ & ${\small 0.076}$ & ${\small 0.102}$ &  & ${\small 0.016}$ &
						${\small 0.036}$ & ${\small 0.072}$ \\
						& ${\small 0.026}$ & ${\small 0.064}$ & ${\small 0.108}$ &  & ${\small 0.028}
						$ & ${\small 0.06}$ & ${\small 0.122}$ &  & ${\small 0.008}$ & ${\small 0.014%
						}$ & ${\small 0.02}$ \\
						${\small c=3}$ & ${\small 0.176}$ & ${\small 0.382}$ & ${\small 0.516}$ &  &
						${\small 0.22}$ & ${\small 0.438}$ & ${\small 0.526}$ &  & ${\small 0.262}$
						& ${\small 0.45}$ & ${\small 0.55}$ \\
						& ${\small 0.41}$ & ${\small 0.632}$ & ${\small 0.738}$ &  & ${\small 0.256}$
						& ${\small 0.428}$ & ${\small 0.538}$ &  & ${\small 0.06}$ & ${\small 0.124}$
						& ${\small 0.164}$ \\
						${\small c=6}$ & ${\small 0.704}$ & ${\small 0.894}$ & ${\small 0.948}$ &  &
						${\small 0.746}$ & ${\small 0.934}$ & ${\small 0.976}$ &  & ${\small 0.69}$
						& ${\small 0.916}$ & ${\small 0.974}$ \\
						& ${\small 0.914}$ & ${\small 0.986}$ & ${\small 0.996}$ &  & ${\small 0.776}
						$ & ${\small 0.93}$ & ${\small 0.976}$ &  & ${\small 0.482}$ & ${\small 0.612%
						}$ & ${\small 0.66}$\\
						\hline\hline
					\end{tabular}
					\caption{Rejection probabilities of {SARARMA(1,0,1)} using bootstrap
						test $\fancyt^{a \ast}$ at 1, 5, 10\% levels, power series (\textbf{PS}), trigonometric (\textbf{Trig}) and B-spline (\textbf{B-s}) bases.}\label{table:newsims4}
				\end{table}
			}
			\begin{table}
				\centering{
					\begin{tabular}{l|ll|ll|ll|lll}
						\hline \hline
						& \multicolumn{2}{|c}{$r=2$} & \multicolumn{2}{|c}{$r=3$} &
						\multicolumn{2}{|c}{$r=4$}&
						\multicolumn{3}{|c}{$r=5$} \\ \hline
						$n=150$ & $p=10$ & \multicolumn{1}{|l|}{$p=15$} & $p=10$ & \multicolumn{1}{|l|}{$p=15$} & $p=10$ &
						\multicolumn{1}{|l}{$p=15$}& $p=10$ &
						\multicolumn{1}{|l}{$p=15$}&
						\multicolumn{1}{|l}{$p=20$} \\ \hline
						$c=0$ & 0.0860 & 0.2020 & 0.1180 & 0.2060 & 0.1420 & 0.2240 \\
						$c=3$ & 0.3320 & 0.6340 & 0.3700 & 0.6380 & 0.3760 & 0.6700 \\
						$c=6$ & 0.9060 & 0.9920 & 0.9180 & 0.9940 & 0.9220 & 0.9960 \\
						\cmidrule{1-7}
						$n=300$ &  &  &  &  &  & \\
						\cmidrule{1-7}
						$c=0$ & 0.0820 & 0.0960 & 0.0880 & 0.1080 & 0.1060 & 0.1100 \\
						$c=3$ & 0.2680 & 0.5980 & 0.2600 & 0.6120 & 0.2780 & 0.6220 \\
						$c=6$ & 0.8140 & 0.9980 & 0.8160 & 0.9980 & 0.8220 & 0.9980 \\
						\cmidrule{1-7}
						$n=500$ &  &  &  &  &  & \\
						\cmidrule{1-7}
						$c=0$ & 0.0280 & 0.0420 & 0.0260 & 0.0400 & 0.0360 & 0.0480 \\
						$c=3$ & 0.2320 & 0.6660 & 0.2400 & 0.6620 & 0.2460 & 0.6680 \\
						$c=6$ & 0.8920 & 1 & 0.9040 & 1 & 0.9000 & 1 \\
						\cmidrule{1-7}
						$n=600$ &  &  &  &  &  & \\
						\cmidrule{1-7}
						$c=0$ & 0.0320 & 0.0500 & 0.0340 & 0.0540 & 0.0360 & 0.0540 \\
						$c=3$ & 0.3140 & 0.6480 & 0.3080 & 0.6280 & 0.3120 & 0.6460 \\
						$c=6$ & 0.9220 & 1 & 0.9180 & 1 & 0.9180 & 1 \\
						\hline
						$n=700$ &  &  &  &  &  & \\
						\hline
						$c=0$ &  0.0260 & 0.0300 & 0.0280 & 0.0380 & 0.0280 & 0.0380 & 0.0280 & 0.0420 & 0.0580 \\
						$c=3$ &  0.2420 & 0.5540 & 0.2400 & 0.5480 & 0.2520 & 0.5500 & 0.2420 & 0.5600 & 0.6920 \\
						$c=6$ &  0.9580 & 0.9980 & 0.9560 & 0.9980 & 0.9600 & 0.9980 & 0.9500 & 0.9980 & 1 \\
						\hline \hline
					\end{tabular}%
				}
				\caption{Rejection probabilities of $\fancyt$ at 5\% asymptotic level, nonparametric spatial error structure.}\label{table:sims5}
			\end{table}
			
			\begin{table}
				\centering{
					{\small
						\begin{tabular}{l|llll|llll}
							\hline
							& \multicolumn{4}{|l}{1998} & \multicolumn{4}{|l}{Pooled} \\ \hline
							& $H_{0}$ & p-value & $H_{1}$ & p-value & $H_{0}$ & p-value & $H_{1}$ &  \\
							\hline
							& \multicolumn{8}{|c}{SARARMA(2,1,0)} \\ \hline
							$W^A y$ & \multicolumn{1}{|c}{-0.005} & $<$0.001 & -0.003 &
							$<$0.001 & 0.013 & $<$0.001 & 0.013 & $<$%
							0.001 \\
							$W^E y$ & \multicolumn{1}{|c}{0.130} & $<$0.001 & 0.129 &
							$<$0.001 & 0.121 & $<$0.001 & 0.121 & $<$%
							0.001 \\
							$W^d$ & \multicolumn{1}{|c}{-0.159} & 0.281 & -0.225 & $<$0.001
							& -0.086 & 0.033 & -0.086 & 0.033 \\
							$\fancyt$ &  &  & -1.921 & 0.973 &  &  & -2.531 & 0.994 \\
							$\fancyt^{\ast }$ &  &  &  & 0.840 &  &  &  & 0.940 \\
							$\fancyt^{a}$ & \multicolumn{1}{|c}{} &  & -1.918 & 0.972 &  &  & -2.547 &
							0.995 \\
							$\fancyt^{a\ast }$ &  &  &  & 0.870 &  &  &  & 0.730 \\ \hline
							& \multicolumn{8}{|c}{SARARMA(2,0,1)} \\ \hline
							$W^A y$ & {\small 0.001} & {\small $<$0.01} & {\small 0.011}
							& {\small $<$0.01} & {\small 0.013} & {\small $<$0.01}
							& {\small 0.013} & {\small $<$0.01} \\
							$W^E y$ & {\small 0.127} & {\small $<$0.01} & {\small 0.122}
							& {\small $<$0.01} & {\small 0.121} & {\small $<$0.01}
							& {\small 0.121} & {\small $<$0.01} \\
							$W^d$ & {\small -0.153} & {\small $<$0.01} & {\small -0.050} &
							{\small $<$0.01} & {\small -0.086} & {\small $<$0.01}
							& {\small -0.086} & {\small 0.025} \\
							$\fancyt$ &  &  & {\small -1.763} & {\small 0.961} &  &  & {\small -2.421} &
							{\small 0.992} \\
							$\fancyt^{\ast }$ &  &  &  & {\small 0.900} &  &  &  & {\small 0.990} \\
							$\fancyt^{a}$ &  &  & {\small -2.349} & {\small 0.991} &  &  & {\small -2.423}
							& {\small 0.992} \\
							\multicolumn{1}{l|}{$\fancyt^{a\ast }$} &  &  &  & {\small 0.850} &  &  &  &
							{\small 0.790} \\ \hline
							& \multicolumn{8}{|c}{Nonparametric} \\ \hline
							$W^A y$ & {\small -0.052} & {\small $<$0.001} & {\small -0.011%
							} & {\small $<$0.001} & {\small 0.033} & {\small $<$%
								0.001} & {\small 0.033} & {\small $<$0.001} \\
							$W^E y$ & {\small 0.149} & {\small $<$0.001} & {\small 0.133}
							& {\small $<$0.001} & {\small 0.110} & {\small $<$0.001%
							} & {\small 0.109} & {\small $<$0.001} \\
							$W^d$ &  &  &  &  &  &  &  &  \\
							$\fancyt$ &  &  & {\small -1.294} & {\small 0.902} &  &  & {\small -2.314} &
							{\small 0.990} \\
							$\fancyt^{\ast }$ &  &  &  & {\small 0.830} &  &  &  & {\small 0.850} \\
							$\fancyt^{a}$ &  &  & {\small -1.898} & {\small 0.971} &  &  & {\small -2.530}
							& {\small 0.994} \\
							\multicolumn{1}{l|}{$\fancyt^{a\ast }$} &  &  &  & {\small 0.660} &  &  &  &
							{\small 0.910} \\ \hline \hline
						\end{tabular}%
					}
					
					\caption{The estimates and test statistics for the conflict data. $^*$ denotes the bootstrap p-value.}\label{table:conflict}
				}
			\end{table}

			\small{
				\begin{table}
					\centering{

						\begin{tabular}{l|cl|llll}
							\hline \hline
							Variables & \multicolumn{2}{|l}{FE} & \multicolumn{4}{|l}{SARARMA(0,1,0), $%
								W_{TEC}$}  \\ \hline
							& \multicolumn{1}{c}{} & p-value & $H_{0}$ & p-value & $H_{1}$ & p-value \\ \hline
							$\ln (Spsic)$ & \multicolumn{1}{c}{-0.005} & 0.649 & \multicolumn{1}{|c}{
								0.007} & 0.574 & 0.015 & 0.166  \\
							$\ln (Sptec)$ & \multicolumn{1}{c}{0.191} & $<$0.001 &
							\multicolumn{1}{|c}{0.006} & 0.850 & -0.001 & 0.998  \\
							$\ln (Lab.)$ & \multicolumn{1}{c}{0.636} & $<$0.001 &
							\multicolumn{1}{|c}{0.572} & $<$0.001 &  &    \\
							$\ln (Cap.)$ & \multicolumn{1}{c}{0.154} & $<$0.001 & 0.336 &
							$<$0.001 &  &    \\
							$\ln (R\&D)$ & \multicolumn{1}{c}{0.043} & $<$0.001 & 0.081 &
							$<$0.001 &  &    \\
							$W_{TEC}$ & \multicolumn{1}{|l}{} &  & 0.835 & $<$0.001 & 0.829 &
							$<$0.001 \\
							$\fancyt$ &  &  & \multicolumn{1}{|c}{} &  & 15.528 & $<$0.001
							\\
							$\fancyt^{*}$  & \multicolumn{1}{|l}{} &  &  &  &  & 0.050 \\ \hline \hline
						\end{tabular}
						
						\bigskip
						
						\begin{tabular}{l|llll}
							\hline \hline
							Variables &  \multicolumn{4}{|l}{SARARMA(0,1,0), $W_{SIC}$} \\ \hline
							&  $H_{0}$ & p-value & $H_{1}$ & p-value \\ \hline
							$\ln (Spsic)$ &  0.008 & 0.620 & 0.017 & 0.193 \\
							$\ln (Sptec)$ &  0.039 & 0.157 & 0.020
							& 0.336 \\
							$\ln (Lab.)$ &  0.571 & $
							<$0.001 &  &  \\
							$\ln (Cap.)$ &  0.318 & $<$0.001 &  &  \\
							$\ln (R\&D)$ &  0.082 & $<$0.001 &  &  \\
							$W_{SIC}$ &  0.722 & $<$%
							0.001 & 0.724 & $<$0.001 \\
							$\fancyt$ &    & & 10.451 & $<$0.001 \\
							$\fancyt^{*}$  &   &  &  & {$<$}0.001 \\ \hline \hline
						\end{tabular}
						\caption{The estimates and test statistics for the R\&D data, SARARMA(0,1,0). $^*$ denotes the bootstrap p-value. The price index as well as a dummy variable for missing value in R\&D are included, but we only report the coefficients reported in \cite{Bloom2013}.}\label{table:R&D1}}
				\end{table}
				
				\begin{table}
					\centering{
						\begin{tabular}{l|llll}
							\hline \hline
							{\small Variables} & \multicolumn{4}{|l}{SARARMA(0,2,0)} \\ \hline
							& ${\small H}_{0}$ & {\small p-value} & ${\small H}_{1}$ & {\small p-value}   \\ \hline
							$\ln {\small (Spsic)}$ & {\small 0.009} & {\small 0.587} & {\small 0.018} &
							{\small 0.170} \\
							$\ln {\small (Sptec)}$ & \multicolumn{1}{|c}{\small 0.044} & {\small 0.112} &
							{\small 0.026} & {\small 0.236}  \\
							$\ln {\small (Lab.)}$ & \multicolumn{1}{|c}{\small 0.573} & {\small
								$<$0.001}   \\
							$\ln {\small (Cap.)}$ & \multicolumn{1}{|c}{\small 0.315} & {\small
								$<$0.001}   \\
							$\ln {\small (R\&D)}$ & {\small 0.082} & {\small $<$0.001}  \\
							${\small W}_{SIC}$ & {\small 0.696} & {\small $<$0.001} & {\small %
								0.693} & {\small $<$0.001}
							\\
							${\small W}_{TEC}$ & {\small 0.157} & {\small 0.092} & {\small 0.164} &
							{\small 0.079}
							\\
							$\fancyt$ & \multicolumn{1}{|c}{} &  & {\small 10.485} & {\small
								$<$0.001}  \\
							$\fancyt^{*}$  &  &  &  & {\small 0.060} \\
							\hline
							\hline
						\end{tabular}
						
						\bigskip
						
						\begin{tabular}{l|llll}
							\hline \hline
							{\small Variables} &
							\multicolumn{4}{|l}{SARARMA(0,0,2)} \\ \hline
							& ${\small H}_{0}$ & {\small p-value} & ${\small H}_{1}$ & {\small p-value}   \\ \hline
							$\ln {\small (Spsic)}$ &  {\small -0.0002} & {\small 0.991} & {\small 0.013} &
							{\small 0.266}  \\
							$\ln {\small (Sptec)}$ &  \multicolumn{1}{|c}{\small 0.033} &
							{\small 0.200} & {\small 0.017} & {\small 0.434} \\
							$\ln {\small (Lab.)}$ &  \multicolumn{1}{|c}{\small 0.565} & {\small
								$<$0.01} \\
							$\ln {\small (Cap.)}$ & \multicolumn{1}{|c}{\small 0.334} & {\small
								$<$0.01}   \\
							$\ln {\small (R\&D)}$ &  {\small 0.076} & {\small $<$0.01}   \\
							${\small W}_{SIC}$ &  {\small 0.624} & {\small $<$0.01} & {\small 0.728} & {\small $<$0.001} \\
							${\small W}_{TEC}$ & {\small 0.312} & {\small 0.123} & {\small 0.321} & {\small
								0.112}
							\\
							$\fancyt$ &  &  & {\small 15.144} & {\small
								$<$0.001} \\
							$\fancyt^{*}$   &  &  && {\small 0.020} \\
							\hline
							\hline
						\end{tabular}
						
						\bigskip
						\pagebreak
						\newpage
						\begin{tabular}{l|llll}
							\hline \hline
							{\small Variables} &  \multicolumn{4}{|l}{\small %
								Error MESS(2)} \\ \hline
							& ${\small H}_{0}$ & {\small p-value} & ${\small H}_{1}$ & {\small p-value}  \\ \hline
							$\ln {\small (Spsic)}$ &  {\small 0.002} & {\small 0.788} & {\small 0.014} & {\small %
								0.040} \\
							$\ln {\small (Sptec)}$ & \multicolumn{1}{|c}%
							{\small 0.045} & {\small 0.025} & {\small 0.027} & {\small 0.088} \\
							$\ln {\small (Lab.)}$ &  \multicolumn{1}{|c}{\small 0.569} & {\small
								$<$0.001} &  &  \\
							$\ln {\small (Cap.)}$ & \multicolumn{1}{|c}{\small 0.323} & {\small
								$<$0.001} &  &  \\
							$\ln {\small (R\&D)}$  & {\small 0.077} &
							{\small $<$0.001} &  &  \\
							${\small W}_{SIC}$ &  {\small 0.775} &
							{\small $<$0.001} & {\small 0.836} & {\small $<$0.001}
							\\
							${\small W}_{TEC}$ & {\small 0.338} & {\small 0.010} & {\small 0.380} & {\small 0.004}
							\\
							$\fancyt$ &  & & {\small 12.776} & {\small
								$<$0.001} \\
							$\fancyt^{*}$   &  &
							&  & {\small 0.050}\\
							\hline
							\hline
						\end{tabular}
						
						\caption{The estimates and test statistics for the R\&D data, SARARMA(0,2,0) and Error MESS(2). $^*$ denotes the bootstrap p-value. The price index as well as a dummy variable for missing value in R\&D are included, but we only report the coefficients reported in \cite{Bloom2013}.}\label{table:R&D2}}
			\end{table}}
			
			\begin{table}
				\centering{
					\begin{tabular}{lcc|cc}
						\hline \hline
						Variable & \multicolumn{2}{c}{$w_{ij}^{\ast }=d_{ij}^{-2}$\ for $i\neq j$} &
						\multicolumn{2}{|c}{$w_{ij}^{\ast }=e^{-2d_{ij}}$\ for $i\neq j$} \\ \hline
						& estimate & p-value & estimate & p-value \\ \hline
						Constant & $1.0711$ & $0.608$ & $0.5989$ & $0.798$ \\
						$\ln (s)$ & $0.8256$ & $<0.001$ & $0.7938$ & $<0.001$ \\
						$\ln (n_{p}+0.05)$ & $-1.4984$ & $0.008$ & $-1.4512$ & $0.009$ \\
						$W\ln (s)$ & $-0.3159$ & $0.075$ & $-0.3595$ & $0.020$ \\
						$W\ln (n_{p}+0.05)$ & $0.5633$ & $0.498$ & $0.1283$ & $0.856$ \\
						$Wy$ & $0.7360$ & $<0.001$ & $0.6510$ & $<0.001$ \\
						\hline
						$\fancyt$ & $-1.88$ & 0.970 & $-2.08$ & 0.981 \\
						$\fancyt^{*}$  &  & 0.850 &  & 0.900 \\
						$\fancyt^{a}$ & $-1.90$ & 0.971 & $-2.05$ & 0.980 \\
						$\fancyt^{a*}$  &  & 0.820 &  & 0.810 \\ \hline
						Restricted regression &  &  &  &  \\
						\hline
						Constant & $2.1411$ & $<0.001$ & $2.9890$ & $<0.001$ \\
						$\ln (s)-\ln (n+0.05)$ & $0.8426$ & $<0.001$ & $0.8195$ & $<0.001$ \\
						$W[\ln (s)-\ln (n_{p}+0.05)]$ & $-0.2675$ & $0.122$ & $-0.2589$ & $0.098$ \\
						$W\ln (y)$ & $0.7320$ & $<0.001$ & $0.6380$ & $<0.001$ \\ \hline
						$\fancyt$ & $0.30$ & 0.382 & $4.04$ & $<0.001$ \\
						$\fancyt^{*}$  &  & 0.500 &  & $<0.001$ \\
						$\fancyt^{a}$ & $0.10$ & 0.460 & $4.50$ & $<0.001$ \\
						$\fancyt^{a*}$  &  & 0.560 &  & $0.040$ \\ \hline \hline
					\end{tabular}
					
					\caption{The estimates and test statistics of the linear SAR model for the growth data. $^*$ denotes the bootstrap p-value.}\label{table:growth}
				}
			\end{table}

			\clearpage
			\begin{center}
				{\textbf{\Large{Appendix}}}
			\end{center}
			\appendix
			\section{Proofs of theorems and propositions}
			\begin{proof}[Proof of Proposition \ref{prop:Sigma_diff_bound}:] Because the map $\Sigma:\Gamma^o\rightarrow \mathcal{M}^{n\times n}$ is Fr\'echet-differentiable on $\Gamma^o$, it is also G\^ateaux-differentiable and the two derivative maps coincide. Thus by Theorem 1.8 of \cite{Ambrosetti1995},  $\left\Vert\Sigma\left(\gamma_1\right)-\Sigma\left(\gamma_2\right)\right\Vert\leq \sup_{\gamma\in \ell\left[\gamma_1,\gamma_2\right]}\left\Vert D\Sigma(\gamma)\right\Vert\left\Vert\gamma_1-\gamma_2\right\Vert,$
				where $\ell\left[\gamma_1,\gamma_2\right]=\left\{t\gamma_1+(1-t)\gamma_2:t\in[0,1]\right\}$. The claim now follows by (\ref{Sigma_fre_der_bdd}) in Assumption 8.
			\end{proof}
			
			\begin{proof}[Proof of Theorem \ref{thm:consistency}]
				This a particular case of the proof of Theorem \ref{thm:cons_SAR} with $\lambda=0$, and so $S(\lambda)=I_n$.
			\end{proof}
			\begin{proof}[Proof of Theorem \ref{thm:stat_appr}]
				In the supplementary appendix.
			\end{proof}
			\begin{proof}[Proof of Theorem \ref{thm:appr_clt}]
				We would like to establish the asymptotic unit normality of
				\begin{equation}
					\frac{\sigma _{0}^{-2}\varepsilon ^{\prime }\fancyv\varepsilon -p}{\sqrt{2p}}.  \label{clt_tgt2}
				\end{equation}%
				Writing $q=\sqrt{2p}$, the ratio in (\ref{clt_tgt2}) has zero mean and variance equal to one, and
				may be written as $\sum_{s=1}^{\infty }w_{s}$, where
				$w_{s}=\sigma _{0}^{-2}q^{-1}v_{ss}\left( \varepsilon _{s}^{2}-\sigma
				_{0}^{2}\right) +2\sigma _{0}^{-2}q^{-1}\mathbf{1}(s\geq 2)\varepsilon
				_{s}\sum_{t<s}v_{st}\varepsilon _{t},$
				with $v_{st}$ the typical element of $\fancyv$, with $s,t=1,2,\ldots ,$. We first
				show that
				\begin{equation}
					w_{\ast }\overset{p}{\longrightarrow }0,  \label{clt_first_target}
				\end{equation}%
				where $w_{\ast }=w-w_{S}$, $w_{S}=\sum_{s=1}^{S}w_{s}$ and $S=S_{n}$ is a
				positive integer sequence that is increasing in $n$. All expectations in the
				sequel are taken conditional on $X$. By Chebyshev's inequality proving
				\begin{equation}
					\mathcal{E}w_{\ast }^{2}\overset{p}{\rightarrow }0  \label{pmle_clt_tail_neg}
				\end{equation}%
				is sufficient to establish (\ref{clt_first_target}). Notice that
				$\mathcal{E}w_{s}^{2}\leq Cq^{-2}v_{ss}^{2}+Cq^{-2}\mathbf{1}(s\geq
				2)\sum_{t<s}v_{st}^{2}\leq Cq^{-2}\sum_{t\leq s}v_{st}^{2},
				$
				so that, writing $\mathscr{M}=\Sigma ^{-{1}}\Psi \lbrack \Psi ^{\prime }\Sigma
				^{-1}\Psi ]^{-1}\Psi ^{\prime }\Sigma ^{-{1}}$,
				\begin{eqnarray}
					&&\sum_{s=S+1}^{\infty }\mathcal{E}w_{s}^{2}\leq Cq^{-2}\sum_{s=S+1}^{\infty
					}\sum_{t\leq s}v_{st}^{2}\leq Cq^{-2}\sum_{s=S+1}^{\infty }b_{s}^{\prime
					}M\sum_{t\leq s}b_{t}b_{t}^{\prime }\mathscr{M}b_{s}  \notag \\
					&\leq &Cq^{-2}\left \Vert \Sigma \right \Vert \sum_{s=S+1}^{\infty
					}b_{s}^{\prime }\mathscr{M}^{2}b_{s}\leq Cq^{-2}\sum_{s=S+1}^{\infty
					}\sum_{i,j,k=1}^{n}b_{is}b_{kt}m_{ij}m_{kj}  \notag \\
					&\leq &Cq^{-2}\sum_{s=S+1}^{\infty }\sum_{i,k=1}^{n} \left \vert b_{is}^{\ast }\right
					\vert \left \vert b_{ks}^{\ast }\right \vert  \sum_{j=1}^{n}\left
					( m_{kj}^2+m_{ij}^2\right )  ,  \label{pmle_clt_tail_neg_2}
				\end{eqnarray}%
				where $m_{ij}$ is the $(i,j)$-th element of $\mathscr{M}$ and we have used the inequality $|ab|\leq \left(a^2+b^2\right)/2$ in the last step. Now, denote by $h_i'$ the $i$-th row of the $n\times p$ matrix $\Sigma^{-1}\Psi$. Denoting the elements of $\Sigma^{-1}$ by $\Sigma^{-1}_{ij}$ and $\psi_{jl}=\psi\left(x_{jl}\right)$, $h_i$ has entries $h_{il}=\sum_{j=1}^n\Sigma^{-1}_{ij}\psi_{jl}$, $l=1,\ldots,p$. We have $\left\vert h_{il}\right\vert= O_p\left(\sum_{j=1}^n\left\vert\Sigma^{-1}_{ij}\right\vert\right)=O_p\left(\left\Vert \Sigma^{-1}\right\Vert_R\right)=O_p(1)$, uniformly, by Assumptions \ref{ass:rsums_no_coll} and \ref{ass:psi_mom}. Thus, we have $\left\Vert h_i\right\Vert=O_p\left(\sqrt{p}\right)$, uniformly in $i$. As a result,
				\begin{equation}
					\left \vert m_{ij}\right \vert =n^{-1}\left\vert h_i'\left(n^{-1}\Psi ^{\prime }\Sigma^{-1}\Psi \right)^{-1}h_j\right\vert =O_{p}\left( n^{-1}\left\Vert h_i\right\Vert\left\Vert h_j\right\Vert\right)= O_{p}\left( pn^{-1}\right),  \label{pbound}
				\end{equation}%
				uniformly in $i,j$, by Assumption \ref{ass:rsums_no_coll}. Similarly, note that \begin{eqnarray}
					\sum_{j=1}^n m_{ij}^2&=&n^{-1}h_i'\left(n^{-1}\Psi ^{\prime }\Sigma^{-1}\Psi \right)^{-1}\left(n^{-1}\Psi ^{\prime }\Sigma^{-2}\Psi \right)\left(n^{-1}\Psi ^{\prime }\Sigma^{-1}\Psi \right)^{-1}h_i\nonumber\\
					&\leq& n^{-1}\left\Vert h_i\right\Vert^2\left\Vert\left(n^{-1}\Psi ^{\prime }\Sigma^{-1}\Psi \right)^{-1} \right\Vert^2\left\Vert n^{-1}\Psi ^{\prime }\Sigma^{-2}\Psi \right\Vert\nonumber\\
					&=&O_p\left(pn^{-2}\left\Vert \Psi\right\Vert^2\left\Vert \Sigma^{-1}\right\Vert^2\right)=O_p\left(pn^{-1}\right)\label{m_square_sums},
				\end{eqnarray}
				uniformly in $i$. Thus (\ref{pmle_clt_tail_neg_2}) is
				\begin{equation}
					O_{p}\left( q^{-2}pn^{-1}\sum_{i=1}^{n}\sum_{s=S+1}^{\infty }\left \vert
					b_{is}^{\ast }\right \vert \sum_{t=1}^{n}\left \vert b_{ks}^{\ast }\right
					\vert \right) =O_{p}\left( q^{-2}p\sup_{i=1,\ldots
						,n}\sum_{s=S+1}^{\infty }\left \vert b_{is}^{\ast }\right \vert \right) ,
					\label{pmle_clte_tail_neg_3}
				\end{equation}%
				by Assumption \ref{ass:errors_epsilon_3}. By the same assumption, there
				exists $S_{in}$ such that $\sum_{s=S_{in}+1}^{\infty }\left \vert
				b_{is}^{\ast }\right \vert \leq \epsilon_{n}$ for any decreasing
				sequence $\epsilon_{n}\rightarrow 0$ as $n\rightarrow \infty $. Choosing
				$S=\max_{i=1,\ldots ,n}S_{in}$ in $w_{S}$, we deduce that (\ref%
				{pmle_clte_tail_neg_3}) is $O_{p}\left( q^{-2}p\epsilon
				_{n}\right) =O_{p}\left( \epsilon _{n}\right) =o_{p}(1)$, proving (\ref%
				{pmle_clt_tail_neg}). Thus we need only focus on $w_{S}$, and seek to
				establish that
				\begin{equation}
					w_{S}\longrightarrow _{d}N(0,1),\text{ as }n\rightarrow \infty .
					\label{CLT_trunc_dist_target}
				\end{equation}%
				From \cite{Scott1973}, (\ref{CLT_trunc_dist_target}) follows if
				\begin{equation}
					\sum_{s=1}^{S}\mathcal{E}w_{s}^{4}\overset{p}{\longrightarrow }0,\text{ as }%
					n\rightarrow \infty ,  \label{pmle_lyap}
				\end{equation}%
				and
				\begin{equation}
					\sum_{s=1}^{S}\left[ \mathcal{E}\left( w_{s}^{2}\left. {}\right \vert
					\varepsilon _{t},t<s\right) -\mathcal{E}\left( w_{s}^{2}\right) \right] \overset{%
						p}{\longrightarrow }0,\text{ as }n\rightarrow \infty .  \label{pmle_var}
				\end{equation}%
				We show (\ref{pmle_lyap}) first. Evaluating the expectation and using (\ref{m_square_sums}) yields
				\begin{eqnarray*}
					\mathcal{E}w_{s}^{4} &\leq &Cq^{-4}v_{ss}^{4}+Cq^{-4}\sum_{t<s}v_{st}^{4}\leq
					Cq^{-4}\left( \sum_{t\leq s}v_{st}^{2}\right) ^{2} \leq Cq^{-4}\left( b_{s}^{\prime }\mathscr{M}\sum_{t\leq s}b_{t}b_{t}^{\prime
					}\mathscr{M}b_{s}\right) ^{2}\\
					&\leq& Cq^{-4}\left( b_{s}^{\prime }\mathscr{M}^{2}b_{s}\right)
					^{2}= Cq^{-4}\sum_{i,j,k=1}^n b_{is}b_{ks}m_{ij}m_{kj}\leq Cq^{-4}\sum_{i,k=1}^n \left \vert b^*_{is}\right\vert \left\vert b^*_{ks}\right\vert\sum_{j=1}^n \left(m^2_{ij}+m^2_{kj}\right) \\
					&=&O_{p}\left( q^{-4}pn^{-1}\left(\sum_{i=1}^n \left \vert b^*_{is}\right\vert\right)^2\right) ,
				\end{eqnarray*}%
				whence
				\begin{eqnarray*}
					\sum_{s=1}^{S}\mathcal{E}w_{s}^{4} &=&O_{p}\left(
					q^{-4}pn^{-1}\sum_{s=1}^{S}\left(\sum_{i=1}^n \left \vert b^*_{is}\right\vert\right)^2\right)
					=O_{p}\left( q^{-4}pn^{-1}\sum_{s=1}^{S}\left( \sum_{i=1}^{n}\left\vert b_{is}^{\ast}\right\vert\right) \right) =O_{p}\left( q^{-4}p\right) ,
				\end{eqnarray*}%
				by Assumption \ref{ass:errors_epsilon_3}. Thus (\ref{pmle_lyap}) is established. Notice that $\mathcal{E}\left( \left. w_{s}^{2}\right \vert \epsilon
				_{t},t<s\right) $ equals
				\begin{equation*}
					4q^{-2}\sigma _{0}^{-4}\left \{ \left( \mu _{4}-\sigma _{0}^{4}\right)
					v_{ss}^{2}+2\mu _{3}\mathbf{1}(s\geq 2)\sum_{t<s}v_{st}v_{ss}\varepsilon
					_{t}\right \} +4q^{-2}\sigma _{0}^{-2}\mathbf{1}(s\geq 2)\left(
					\sum_{t<s}v_{st}\varepsilon _{t}\right) ^{2},
				\end{equation*}%
				and $
				\mathcal{E}w_{s}^{2}=4q^{-2}\sigma _{0}^{-4}\left( \mu _{4}-\sigma
				_{0}^{4}\right) v_{ss}^{2}+4q^{-2}\mathbf{1}(s\geq 2)\sum_{t<s}v_{st}^{2},
				$
				so that (\ref{pmle_var}) is bounded by a constant times%
				\begin{equation}
					q^{-2}\sum_{s=2}^{S}\sum_{t<s}v_{st}v_{ss}\varepsilon _{t}+\left \{
					\sum_{s=2}^{S}\left( \sum_{t<s}v_{st}\varepsilon _{t}\right) ^{2}-\sigma
					_{0}^{2}\sum_{t<s}v_{st}^{2}\right \} .  \label{pmle_var_11_1}
				\end{equation}%
				By transforming the range of summation, the square of the first term in (\ref%
				{pmle_var_11_1}) has expectation bounded by
				\begin{equation}
					Cq^{-4}\mathcal{E}\left( \sum_{t=1}^{S-1}\sum_{s=t+1}^{S}v_{st}v_{ss}\varepsilon
					_{t}\right) ^{2}\leq Cq^{-4}\sum_{t=1}^{S-1}\left(
					\sum_{s=t+1}^{S}v_{st}v_{ss}\right) ^{2},  \label{pmle_var_11_2}
				\end{equation}%
				where the factor in parentheses on the RHS of (\ref{pmle_var_11_2}) is
				\begin{eqnarray*}
					&&\sum_{s,r=t+1}^{S}b_{s}^{\prime }\mathscr{M}b_{s}b_{s}^{\prime }\mathscr{M}b_{t}b_{r}^{\prime
					}\mathscr{M}b_{r}b_{r}^{\prime }\mathscr{M}b_{t} \leq\sum_{s,r=t+1}^{S}\left \vert b_{s}^{\prime }\mathscr{M}b_{s}b_{r}^{\prime
					}\mathscr{M}b_{r}\right \vert \left \vert b_{s}^{\prime }\mathscr{M}b_{t}\right \vert \left
					\vert b_{r}^{\prime }\mathscr{M}b_{t}\right \vert \\
					&\leq & C\sum_{s,r=t+1}^{S}\sum_{i,j,k,l=1}^{n}\left \vert b_{is}\right \vert
					\left \vert m_{ij}\right \vert \left \vert b_{jr}\right \vert \left \vert
					b_{ks}\right \vert \left \vert m_{lk}\right \vert \left \vert b_{kr}\right
					\vert \left \vert b_{s}^{\prime }\mathscr{M}b_{t}\right \vert \left \vert
					b_{r}^{\prime }\mathscr{M}b_{t}\right \vert \\
					&\leq &C\left( \sup_{i,j}\left \vert m_{ij}\right \vert \right) ^{2}\left(
					\sup_{s\geq 1}\sum_{i=1}^{n}\left \vert b_{is}^{\ast }\right \vert \right)
					^{4}\sum_{s,r=t+1}^{S}\left \vert b_{s}^{\prime }\mathscr{M}b_{t}\right \vert \left
					\vert b_{r}^{\prime }\mathscr{M}b_{t}\right \vert \\
					&=&O_{p}\left( p^{2}n^{-2}\left( \sum_{s=t+1}^{S}\left \vert b_{t}^{\prime
					}\mathscr{M}b_{s}\right \vert \right) ^{2}\right) =O_{p}\left( p^{2}n^{-2}\left(
					\sum_{s=t+1}^{S}\sum_{i,j=1}^{n}\left \vert b_{it}^{\ast }\right \vert \left
					\vert m_{ij}\right \vert \left \vert b_{js}^{\ast }\right \vert \right)
					^{2}\right) ,
				\end{eqnarray*}%
				where we used Assumptions \ref{ass:errors_epsilon_3} and (\ref{pbound}). Now
				Assumptions \ref{ass:errors_epsilon_3}, \ref{ass:rsums_no_coll} and (\ref%
				{pbound}) imply that
				\begin{eqnarray*}
					\sum_{s=t+1}^{S}\sum_{i,j=1}^{n}\left \vert b_{it}^{\ast }\right \vert \left
					\vert m_{ij}\right \vert \left \vert b_{js}^{\ast }\right \vert
					&=&O_p\left(\sup_{i,j}\left \vert m_{ij}\right \vert\sup_t\sum_{i=1}^{n}\left \vert b_{it}^{\ast }\right \vert \sum_{j=1}^{n}\sum_{s=t+1}^{S}\left \vert b_{js}^{\ast }\right
					\vert \right)=O_{p}\left( p\sup_t\sum_{i=1}^{n}\left \vert b_{it}^{\ast }\right \vert \right) ,
				\end{eqnarray*}
				so (\ref{pmle_var_11_2}) is $O_{p}\left( q^{-4}p^{4}n^{-2}\sup_{t}\left( \sum_{i=1}^{n}\left \vert b_{it}^{\ast }\right \vert \right) \left(
				\sum_{i=1}^{n}\left( \sum_{t=1}^{S-1}\left \vert b_{it}^{\ast }\right \vert
				\right) \right) \right)$. By Assumption \ref{ass:errors_epsilon_3} the
				latter is $O_{p}\left( q^{-4}p^{4}n^{-1}\right) $ and therefore the first term in (\ref{pmle_var_11_1}) is $O_{p}\left( p^2n^{-1}\right) $, which is negligible.
				
				Once again transforming the  summation range and using the inequality $%
				|a+b|^2\leq C\left(a^2+b^2\right)$, we can bound the square of the second
				term in (\ref{pmle_var_11_1}) by a constant times
				\begin{equation}  \label{pmle_var_11_3}
					\left(\sum_{t=1}^{S-1}\sum_{s=t+1}^S v_{st}^2
					\left(\varepsilon_t^2-\sigma_0^2\right)\right)^2+\left(2\sum_{t=1}^{S-1}%
					\sum_{r=1}^{t-1}\sum_{s=t+1}^S v_{st} v_{sr} \varepsilon_t\varepsilon_r\right)^2.
				\end{equation}
				
				Using Assumption \ref{ass:errors_epsilon_3}, the expectations of the two terms in (\ref%
				{pmle_var_11_3}) are bounded by a constant times $\alpha _{1}$ and a
				constant times $\alpha _{2}$, respectively, where
				$
				\alpha _{1} =\sum_{t=1}^{S-1}\left( \sum_{s=t+1}^{S}v_{st}^{2}\right) ^{2},
				\alpha _{2} =\sum_{t=1}^{S-1}\sum_{r=1}^{t-1}\left(
				\sum_{s=t+1}^{S}v_{st}v_{sr}\right) ^{2}.  $
				Thus (\ref{pmle_var_11_3}) is $O_{p}\left( \alpha _{1}+\alpha _{2}\right) $.
				Now by (\ref{pbound}), Assumptions \ref{ass:errors_epsilon_3}, \ref%
				{ass:rsums_no_coll} and elementary inequalities $\alpha _{2}$ is bounded by
				\begin{eqnarray*}
					&&\sum_{t=1}^{S-1}\sum_{r=1}^{t-1}\sum_{s=t+1}^{S}\sum_{u=t+1}^{S}b_{s}^{%
						\prime }\mathscr{M}b_{t}b_{s}^{\prime }\mathscr{M}b_{r}b_{u}^{\prime }\mathscr{M}b_{t}b_{u}^{\prime }\mathscr{M}b_{r}
					\\
					&=&O_{p}\left( q^{-4}\sum_{s,r,t,u=1}^{S}\sum_{i,j=1}^{n}\left \vert
					b_{ir}^{\ast }\right \vert \left \vert m_{ij}\right \vert \left \vert
					b_{js}^{\ast }\right \vert \sum_{i,j=1}^{n}\left \vert b_{ir}^{\ast }\right
					\vert \left \vert m_{ij}\right \vert \left \vert b_{ju}^{\ast }\right \vert
					\sum_{i,j=1}^{n}\left \vert b_{it}^{\ast }\right \vert \left \vert
					m_{ij}\right \vert \left \vert b_{js}^{\ast }\right \vert\sum_{i,j=1}^{n}\left \vert b_{it}^{\ast }\right \vert \left \vert
					m_{ij}\right \vert \left \vert b_{ju}^{\ast }\right \vert \right) \\
					&=&O_{p}\left( q^{-4}pn^{-1}\sum_{s,r,t=1}^{S}\left(
					\sum_{i,j=1}^{n}\left \vert b_{ir}^{\ast }\right \vert \left \vert
					m_{ij}\right \vert \left \vert b_{js}^{\ast }\right \vert \right) \left(
					\sum_{i,j=1}^{n}\left \vert b_{ir}^{\ast }\right \vert \left \vert
					m_{ij}\right \vert \sum_{u=1}^{S}\left \vert b_{ju}^{\ast }\right \vert
					\right) \right. \\
					&\times &\left. \sum_{i,j=1}^{n}\left \vert b_{it}^{\ast }\right \vert \left
					\vert m_{ij}\right \vert \left \vert b_{js}^{\ast }\right \vert
					\sum_{i=1}^{n}\left \vert b_{it}^{\ast }\right \vert\sup_u \sum_{j=1}^{n}\left
					\vert b_{ju}^{\ast }\right \vert \right) \\
					&=&O_{p}\left( q^{-4}p^{2}n^{-2}\sum_{s,r=1}^{S}\left( \sum_{i,j=1}^{n}\left
					\vert b_{ir}^{\ast }\right \vert \left \vert m_{ij}\right \vert \left \vert
					b_{js}^{\ast }\right \vert \right) \sum_{i=1}^{n}\left \vert b_{ir}^{\ast
					}\right \vert\sum_{j=1}^{n}\left(\sum_{u=1}^S\left
					\vert b_{ju}^{\ast }\right \vert\right)\left( \sum_{i,j=1}^{n}\sum_{t=1}^{S}\left \vert b_{it}^{\ast
					}\right \vert \left \vert m_{ij}\right \vert \left \vert b_{js}^{\ast
					}\right \vert \right) \right) \\
					&=&O_{p}\left( q^{-4}p^{2}n^{-1}\sum_{i,j=1}^{n}\left( \sum_{r=1}^{S}\left
					\vert b_{ir}^{\ast }\right \vert \right) \left \vert m_{ij}\right \vert
					\left( \sum_{s=1}^{S}\left \vert b_{js}^{\ast }\right \vert \right)
					\left(\sup_{j} \sum_{i=1}^{n}\left \vert m_{ij}\right \vert \right)
					\sum_{j=1}^{n}\left \vert b_{js}^{\ast }\right \vert \right) \\
					&=&O_p\left(q^{-4}p^{2}n^{-1}\sup_{k}\sum_{i,j=1}^{n}\left \vert m_{ij}\right \vert \sum_{i=1}^{n}\left \vert m_{ik}\right \vert\right)=O_p\left(q^{-4}p^{2}n^{-1}\sup_{k}\sum_{i,j,\ell=1}^{n}\left \vert m_{ij}\right \vert \left \vert m_{\ell k}\right \vert\right)\\
					&=&O_p\left(q^{-4}p^{2}n^{-1}\sup_{k}\sum_{i,j,\ell=1}^{n}\left(m_{ij}^2+ m_{\ell k}^2\right)\right)=O_p\left(q^{-4}p^{2}n^{-1}\sum_{i,j,\ell=1}^{n}\left(m_{ij}^2+ m_{\ell j}^2\right)\right)\\&=&O_p\left(q^{-4}p^{2}n^{-1}\sum_{i,j=1}^{n}m_{ij}^2\right)=O_p\left(q^{-4}p^{2}\sup_{j}\sum_{i=1}^{n}m_{ij}^2\right)=O_p\left(pn^{-1}\right),
				\end{eqnarray*}%
				where we used (\ref{m_square_sums}) in the last step. A similar use of the conditions of the theorem and (\ref{pbound}) implies that $\alpha
				_{1}$ is \allowdisplaybreaks%
				\begin{eqnarray*}
					&&O_{p}\left( q^{-4}\sum_{t=1}^{S-1}\left \{ \sum_{s=t+1}^{S}\left(
					\sum_{i,j=1}^{n}\left \vert m_{ij}\right \vert \left \vert b_{jt}^{\ast
					}\right \vert \left \vert b_{is}^{\ast }\right \vert \right) ^{2}\right \}
					^{2}\right) \\
					&=&O_{p}\left( q^{-4}\left(\sup_{i,j}\left\vert m_{ij}\right\vert\right)^4\sum_{t=1}^{S-1}\left
					\{ \sum_{s=t+1}^{S}\left(\sum_{i=1}^{n}\left\vert b_{is}^{\ast }\right\vert \sum_{j=1}^{n}\left \vert
					b_{jt}^{\ast }\right \vert \right) ^{2}\right \} ^{2}\right) \\
					&=&O_{p}\left( q^{-4}p^{4}n^{-4}\sum_{t=1}^{S-1}\left
					\{ \sum_{s=t+1}^{S}\left(\sum_{i=1}^{n}\left\vert b_{is}^{\ast }\right\vert\right)^2\left( \sum_{j=1}^{n}\left \vert
					b_{jt}^{\ast }\right \vert \right) ^{2}\right \} ^{2}\right)\\
					&=&O_{p}\left( q^{-4}p^{4}n^{-4}\sum_{t=1}^{S-1}\left( \sum_{s=t+1}^{S}\left(\sum_{i=1}^{n}\left\vert b_{is}^{\ast }\right\vert\right)^2\right)^2\left( \sum_{j=1}^{n}\left \vert
					b_{jt}^{\ast }\right \vert \right) ^{4}\right)\\
					&=&O_{p}\left( q^{-4}p^{4}n^{-4}\left( \sum_{t=1}^{S-1}\sum_{j=1}^{n}\left \vert
					b_{jt}^{\ast }\right \vert \right) \left(\sum_{s=t+1}^{S}\sum_{i=1}^{n}\left\vert b_{is}^{\ast }\right\vert\right)^2\sup_s\left( \sum_{i=1}^{n}\left \vert
					b_{is}^{\ast }\right \vert \right) ^{2}\sup_t\left( \sum_{j=1}^{n}\left \vert
					b_{jt}^{\ast }\right \vert \right) ^{3}\right)\\
					&=&O_{p}\left( q^{-4}p^{4}n^{-1}\right)=O_{p}\left( p^{2}n^{-1}\right)
				\end{eqnarray*}%
				proving (\ref{pmle_var}), as $p^{2}/n\rightarrow 0$ by the conditions of the theorem.
			\end{proof}
			
			\begin{proof}[Proof of Theorem \ref{thm:stat_properties}]
				In supplementary appendix.
			\end{proof}
			\begin{proof}[Proof of Theorem \ref{thm:cons_SAR}]
				Due to the similarity with proofs in \cite{Delgado2015} and \cite{Gupta2018}, the details are in the supplementary appendix.
			\end{proof}
			\begin{proof}[Proof of Theorem \ref{thm:stat_appr_SAR}]Denote $\theta ^{\ast }$ as the solution of $\min_{\theta
				}\mathcal{E}\left(y_{i}-\sum_{j=1}^{d_{\lambda }}\lambda
				_{j}w_{i,j}'y-\theta (x_{i})\right)^{2}$. Put $\theta _{i}^{\ast }=\theta
				^{\ast }(x_{i})$, $\theta _{0i}=\theta _{0}(x_{i})$, $\widehat{\theta }%
				_{i}=\psi _{i}^{\prime }\widehat{\beta }$
				, $\widehat{f}_{i}=f(x_{i},\widehat{\alpha })$, $f_{i}^{\ast
				}=f(x_{i},\alpha ^{\ast })$. Then $\widehat{u}_{i}=y_{i}-\sum_{j=1}^{d_{%
						\lambda }}\widehat{\lambda }_{j}w_{i,j}'y-f(x_{i},\widehat{\alpha }%
				)=u_{i}+\theta _{0i}+\sum_{j=1}^{d_{\lambda }}(\lambda _{j_{0}}-\widehat{%
					\lambda }_{j})w_{i,j}'y-\widehat{f}_{i}$. Proceeding as in the proof of Theorem \ref{thm:stat_appr}, we obtain $n\widehat{m}_{n}
				=\widehat{\sigma }^{-2}u^{\prime }\Sigma \left( \widehat{\gamma }\right)
				^{-1}\Psi \lbrack \Psi ^{\prime }\Sigma \left( \widehat{\gamma }\right)
				^{-1}\Psi ]^{-1}\Psi ^{\prime }\Sigma \left( \widehat{\gamma }\right) ^{-1}u+\widehat\sigma^{-2}\sum_{j=1}^7 A_{j}$.
				Thus, compared to the test statistic with no spatial lag, cf. the proof of Theorem \ref{thm:stat_appr}, we have the additional terms%
				\begin{eqnarray*}
					A_{5}&=&\sum_{j=1}^{d_{\lambda }}(\lambda _{j_{0}}-\widehat{\lambda }%
					_{j})y^{\prime }W_{j}^{\prime }\Sigma \left( \widehat{\gamma }\right)
					^{-1}\Psi \lbrack \Psi ^{\prime }\Sigma \left( \widehat{\gamma }\right)
					^{-1}\Psi ]^{-1}\Psi ^{\prime }\Sigma \left( \widehat{\gamma }\right)
					^{-1}\sum_{j=1}^{d_{\lambda }}(\lambda _{j_{0}}-\widehat{\lambda }%
					_{j})W_{j}y, \\
					A_{6}&=&\sum_{j=1}^{d_{\lambda }}(\lambda _{j_{0}}-\widehat{\lambda }%
					_{j})y^{\prime }W_{j}^{\prime }\Sigma \left( \widehat{\gamma }\right)
					^{-1}\Psi \lbrack \Psi ^{\prime }\Sigma \left( \widehat{\gamma }\right)
					^{-1}\Psi ]^{-1}\Psi ^{\prime }\Sigma \left( \widehat{\gamma }\right)
					^{-1}(u+\theta _{0}-\widehat{f}), \\
					A_{7}&=&\left( \Psi \left( \Psi ^{\prime }\Sigma \left( \widehat{\gamma }\right)
					^{-1}\Psi \right) ^{-1}\Psi ^{\prime }\Sigma \left( \widehat{\gamma }\right)
					^{-1}\left( u\mathbf{+}e\right) -e+\theta _{0}-\widehat{f}\right)
					^{\prime }\Sigma \left( \widehat{\gamma }\right) ^{-1}\sum_{j=1}^{d_{\lambda
					}}(\lambda _{j_{0}}-\widehat{\lambda }_{j})W_{j}y.
				\end{eqnarray*}%
				We now show that $A_{\ell}=o_{p}(\sqrt{p}), \ell>4$, so the leading
				term in $n\widehat{m}_{n}$ is the same as before. First $\left\Vert y\right\Vert=O_{p}(\sqrt{n}%
				)$ from $y=(I_{n}-\sum_{j=1}^{d_{\lambda }}\lambda _{j_{0}}W_{j})^{-1}\left(
				\theta _{0}+u\right) $. Then, with $\left\Vert\lambda _{_{0}}-\widehat{\lambda }\right\Vert=O_{p}\left(\sqrt{d_\gamma/n}\right)$ by Lemma \ref{lemma:gamma_order}, we have%
				\begin{eqnarray*}
					\left\vert A_{5}\right\vert
					&\leq &\left \Vert \lambda _{_{0}}-\widehat{\lambda }%
					\right \Vert ^{2}\sum_{j=1}^{d_{\lambda }}\left\Vert W_j\right\Vert^2\sup_{\gamma ,j}\left \Vert \Sigma \left(
					\gamma \right) ^{-1}\frac{1}{n}\Psi \left( \frac{1}{n}\Psi ^{\prime }\Sigma
					\left( {\gamma }\right) ^{-1}\Psi \right) ^{-1}\Psi ^{\prime }\Sigma
					\left( \gamma \right) ^{-1}\right \Vert \left \Vert y\right \Vert ^{2} \\
					&=&O_{p}\left(d_\gamma/n\right) O_{p}(1) O_{p}(n)=O_{p}\left(d_\gamma\right)=o_{p}(\sqrt{p}).
				\end{eqnarray*}%
				Uniformly in $\gamma$ and $j$,
				\begin{eqnarray*}
					&&\mathcal{E}\left ( u^{\prime }S^{-1\prime }W_{j}^{\prime }\Sigma \left( \gamma
					\right) ^{-1}\Psi \lbrack \Psi ^{\prime }\Sigma \left( \gamma \right)
					^{-1}\Psi ]^{-1}\Psi ^{\prime }\Sigma \left( \gamma \right)
					^{-1}u\right )  \\
					&=&\mathcal{E}tr\left( \left( \frac{1}{n}\Psi ^{\prime }\Sigma \left( {\gamma }%
					\right) ^{-1}\Psi \right) ^{-1}\frac{1}{n}\Psi ^{\prime }\Sigma \left(
					\gamma \right) ^{-1}\Sigma S^{-1\prime }W_{j}^{\prime }\Sigma \left( \gamma
					\right) ^{-1}\Psi \right) =O_p(p)
				\end{eqnarray*}%
				and%
				\begin{eqnarray*}
					&&\mathcal{E}\left (\theta _{0}^{\prime }S^{-1\prime }W_{j}^{\prime }\Sigma \left(
					\gamma \right) ^{-1}\Psi \lbrack \Psi ^{\prime }\Sigma \left( \gamma \right)
					^{-1}\Psi ]^{-1}\Psi ^{\prime }\Sigma \left( \gamma \right)
					^{-1}u\right ) ^{2} \\
					&= &O_p\left( \left \Vert S^{-1}\right \Vert ^{2}\sup_{\gamma }\left \Vert
					\Sigma \left( \gamma \right) ^{-1}\right \Vert ^{4}\left \Vert \frac{1}{n}\Psi
					\left( \frac{1}{n}\Psi ^{\prime }\Sigma \left( \gamma \right) ^{-1}\Psi
					\right) ^{-1}\Psi ^{\prime }\right \Vert ^{2}\sup_{j}\left \Vert
					W_{j}\right \Vert ^{2}\left \Vert \Sigma \right \Vert\left\Vert \theta _{0}\right\Vert^2\right) =O_p(n).
				\end{eqnarray*}%
				Similarly, $\theta _{0}^{\prime }S^{-1\prime }W_{j}^{\prime }\Sigma \left(\gamma \right) ^{-1}\Psi \lbrack \Psi ^{\prime }\Sigma \left( \gamma \right)
				^{-1}\Psi ]^{-1}\Psi ^{\prime }\Sigma \left( \gamma \right) ^{-1}W_{j}\theta
				_{0} =O_{p}(n),$ uniformly.
				Therefore,%
				\begin{eqnarray*}
					&&\left \vert \sum_{j=1}^{d_{\lambda }}(\lambda _{j_{0}}-\widehat{\lambda }%
					_{j})y^{\prime }W_{j}^{\prime }\Sigma \left( \widehat{\gamma }\right)
					^{-1}\Psi \lbrack \Psi ^{\prime }\Sigma \left( \widehat{\gamma }\right)
					^{-1}\Psi ]^{-1}\Psi ^{\prime }\Sigma \left( \widehat{\gamma }\right)
					^{-1}u\right \vert  \\
					&=&\left \vert \sum_{j=1}^{d_{\lambda }}(\lambda _{j_{0}}-\widehat{\lambda }%
					_{j})\left( \theta _{0}+u\right) ^{\prime }S^{-1\prime }W_{j}^{\prime
					}\Sigma \left( \widehat{\gamma }\right) ^{-1}\Psi \lbrack \Psi ^{\prime
					}\Sigma \left( \widehat{\gamma }\right) ^{-1}\Psi ]^{-1}\Psi ^{\prime
					}\Sigma \left( \widehat{\gamma }\right) ^{-1}u\right \vert  \\
					&\leq &d_{\lambda }\left \Vert \lambda _{_{0}}-\widehat{\lambda }\right \Vert
					\sup_{\gamma ,j}\left \vert \theta _{0}^{\prime }S^{-1\prime }W_{j}^{\prime
					}\Sigma \left( \gamma \right) ^{-1}\Psi \lbrack \Psi ^{\prime }\Sigma \left(
					\gamma \right) ^{-1}\Psi ]^{-1}\Psi ^{\prime }\Sigma \left( \gamma \right)
					^{-1}u\right \vert  \\
					&&+d_{\lambda }\left \Vert \lambda _{_{0}}-\widehat{\lambda }\right \Vert
					\sup_{\gamma ,j}\left \vert u^{\prime }S^{-1\prime }W_{j}^{\prime }\Sigma
					\left( \gamma \right) ^{-1}\Psi \lbrack \Psi ^{\prime }\Sigma \left( \gamma
					\right) ^{-1}\Psi ]^{-1}\Psi ^{\prime }\Sigma \left( \gamma \right)
					^{-1}u\right \vert  \\
					&=&O_{p}\left(\sqrt{d_\gamma/n}\right) O_{p}(\sqrt{n})+O_{p}\left(\sqrt{d_\gamma/n}\right)
					O_{p}(p)=O_{p}\left(\sqrt{d_\gamma}\right)=o_{p}\left(\sqrt{p}\right),
				\end{eqnarray*}
				and
				\begin{eqnarray*}
					&&\left \vert \sum_{j=1}^{d_{\lambda }}(\lambda _{j_{0}}-\widehat{\lambda }%
					_{j})y^{\prime }W_{j}^{\prime }\Sigma \left( \widehat{\gamma }\right)
					^{-1}\Psi \lbrack \Psi ^{\prime }\Sigma \left( \widehat{\gamma }\right)
					^{-1}\Psi ]^{-1}\Psi ^{\prime }\Sigma \left( \widehat{\gamma }\right)
					^{-1}(\theta _{0}-\widehat{f})\right \vert  \\
					&\leq &d_{\lambda }\left \Vert \lambda _{_{0}}-\widehat{\lambda }\right \Vert
					\left \Vert y\right \Vert \sup_{j}\left \Vert W_{j}\right \Vert \sup_{\gamma
					}\left \Vert \frac{1}{n}\Psi \left( \frac{1}{n}\Psi ^{\prime }\Sigma \left(
					\gamma \right) ^{-1}\Psi \right) ^{-1}\Psi \right \Vert \sup_{\gamma
					}\left \Vert \Sigma \left( \gamma \right) ^{-1}\right \Vert ^{2}\left\Vert\theta _{0}-%
					\widehat{f}\right\Vert \\
					&=&O_{p}\left(\sqrt{d_\gamma/n}\right) O_{p}\left( \sqrt{n}\right)  O_{p}\left(
					p^{1/4}\right) =O_{p}\left( \sqrt{d_\gamma}p^{1/4}\right) =o_{p}(\sqrt{p}),
				\end{eqnarray*}%
				so that $A_{6}=o_p(\sqrt{p})$. Finally,
				\begin{eqnarray*}
					&&\left \vert \sum_{j=1}^{d_{\lambda }}(\lambda _{j_{0}}-\widehat{\lambda }%
					_{j})y^{\prime }W_{j}^{\prime }\Sigma \left( \widehat{\gamma }\right)
					^{-1}\Psi \lbrack \Psi ^{\prime }\Sigma \left( \widehat{\gamma }\right)
					^{-1}\Psi ]^{-1}\Psi ^{\prime }\Sigma \left( \widehat{\gamma }\right)
					^{-1}e\right \vert  \\
					&\leq &d_{\lambda }\left \Vert \lambda _{_{0}}-\widehat{\lambda }\right \Vert
					\left \Vert y\right \Vert \sup_{j}\left \Vert W_{j}\right \Vert \sup_{\gamma
					}\left \Vert \frac{1}{n}\Psi \left( \frac{1}{n}\Psi ^{\prime }\Sigma \left(
					\gamma \right) ^{-1}\Psi \right) ^{-1}\Psi \right \Vert \sup_{\gamma
					}\left \Vert \Sigma \left( \gamma \right) ^{-1}\right \Vert ^{2}\left\Vert e\right\Vert \\
					&=&O_{p}\left(\sqrt{d_\gamma/n}\right) O_{p}\left( \sqrt{n}\right)  O_{p}\left(p^{-\mu
					}\sqrt{n}\right)=O_{p}\left(\sqrt{d_\gamma}p^{-\mu }\sqrt{n}\right)=o_{p}(\sqrt{p}),
				\end{eqnarray*}%
				and%
				\begin{eqnarray*}
					&&\left \vert (e+\theta _{0}-\widehat{f})^{\prime }\Sigma \left( \widehat{%
						\gamma }\right) ^{-1}\sum_{j=1}^{d_{\lambda }}(\lambda _{j_{0}}-\widehat{%
						\lambda }_{j})W_{j}y\right \vert  \\
					&\leq &d_{\lambda }\left \Vert \lambda _{_{0}}-\widehat{\lambda }\right \Vert
					\left( \left \Vert e\right \Vert +\left\Vert\theta _{0}-\widehat{f}\right\Vert\right)
					\sup_{\gamma }\left \Vert \Sigma \left( \gamma \right) ^{-1}\right \Vert
					\sup_{j}\left \Vert W_{j}\right \Vert \left \Vert y\right \Vert  \\
					&=&O_{p}\left(\sqrt{d_\gamma/n}\right) O_{p}\left(p^{-\mu }\sqrt{n}+p^{1/4}\right) O_{p}\left(
					\sqrt{n}\right) =O_{p}\left(\sqrt{d_\gamma}p^{-\mu }\sqrt{n}+\sqrt{d_\gamma}p^{1/4}\right)=o_{p}(\sqrt{p}),
				\end{eqnarray*}
				implying that $A_{7}=o_p(\sqrt{p}).$
			\end{proof}
			\begin{proof}[Proof of Theorem \ref{thm:stat_properties_SAR}]
				Omitted as it is similar to the proof of Theorem \ref{thm:stat_properties}.
			\end{proof}
			\begin{proof}[Proof of Proposition \ref{prop:Sigma_diff_bound_np}:] Because the map $\Sigma:\mathcal{T}^o\rightarrow \mathcal{M}^{n\times n}$ is Fr\'echet-differentiable on $\mathcal{T}^o$, it is also G\^ateaux-differentiable and the two derivative maps coincide. Thus by Theorem 1.8 of \cite{Ambrosetti1995},
				\begin{equation}\label{frech_MVT_tau}
					\left\Vert\Sigma(t_1)-\Sigma(t_1)\right\Vert\leq  \sup_{t\in\mathcal{T}^o}\left\Vert D\Sigma(t)\right\Vert_{\mathscr{L}\left(\mathcal{T}^o,\mathcal{M}^{n\times n}\right)}\left(\left\Vert\gamma_1-\gamma_2\right\Vert+\sum_{\ell=1}^{d_\zeta}\left\Vert \left(\delta_{\ell1}-\delta_{\ell2}\right)'\varphi_{\ell}\right\Vert_{w}\right),
				\end{equation}
				where
				\begin{eqnarray*}
					\sum_{\ell=1}^{d_\zeta}\left\Vert \left(\delta_{\ell1}-\delta_{\ell2}\right)'\varphi_{\ell}\right\Vert_{w}&=&\sum_{\ell=1}^{d_\zeta}\sup_{z\in\mathcal{Z}}\left\vert \left(\delta_{\ell1}-\delta_{\ell2}\right)'\varphi_{\ell}\right\vert\left(1+\left\Vert z\right\Vert^2\right)^{-w/2}\\
					&\leq&\sum_{\ell=1}^{d_\zeta}\left\Vert \delta_{\ell1}-\delta_{\ell2}\right\Vert\sup_{z\in\mathcal{Z}}\left\Vert\varphi_{\ell}\right\Vert\left(1+\left\Vert z\right\Vert^2\right)^{-w/2}\\
					&\leq&C\varsigma(r)\sum_{\ell=1}^{d_\zeta}\left\Vert \delta_{\ell1}-\delta_{\ell2}\right\Vert\leq C\varsigma(r)\left\Vert t_1-t_2\right\Vert.
				\end{eqnarray*}
				The claim now follows by (\ref{Sigma_semi_fre_der_bdd}) in Assumption \ref{ass:Sigma_frech_der_np}, because $\left\Vert\gamma_1-\gamma_2\right\Vert\leq C\varsigma(r) \left\Vert t_1-t_2\right\Vert$ for some suitably chosen $C$.
			\end{proof}
			\begin{proof}[Proof of Theorem \ref{thm:consistency_np}]
				The proof is omitted as it is entirely analogous to that of Theorem \ref{thm:cons_SAR}, with the exception of one difference when proving equicontinuity. In the setting of Section \ref{sec:nonpar_ext}, we obtain
				via Proposition \ref{prop:Sigma_diff_bound_np} that $\left\Vert \Sigma(\tau)-\Sigma\left(\tau^*\right) \right\Vert =O_p\left(\varepsilon\right)$, the $\varsigma(r)$ factor being omitted because only finitely many neighborhoods contribute due to compactness of $\mathcal{T}$.
			\end{proof}
			\begin{proof}[Proof of Theorem \ref{thm:stat_appr_np}]
				Writing, $\delta(z)=\left(\widehat\delta_1'\varphi_1(z),\ldots,\widehat\delta_{d_\zeta}'\varphi_{d_\zeta}(z)\right)'$ and taking $t_1=\left(\widehat\gamma',\hat\delta(z)'\right)'$ and $t_2=\left(\gamma_0',\zeta_0(z)'\right)'$ in Proposition \ref{prop:Sigma_diff_bound_np} implies (we suppress the argument $z$)
				\begin{eqnarray*}
					\left \Vert \Sigma \left( \widehat{\tau }\right) -\Sigma
					\right \Vert =O_{p}\left(\varsigma(r)\left(\left \Vert \widehat{\gamma }-\gamma
					_{0}\right \Vert +\left \Vert \widehat{\delta }-\zeta
					_{0}\right \Vert\right) \right)&=&O_{p}\left( \varsigma(r)\left(\left \Vert \widehat{\tau }-\tau
					_{0}\right \Vert+\left\Vert \nu\right\Vert\right) \right)\\
					&=&O_p\left(\varsigma(r)\max\left\{\sqrt{d_\tau/n}, \sqrt{\sum_{\ell=1}^{d_\zeta}r_\ell^{-2\kappa_\ell}}\right\}\right),
				\end{eqnarray*}
				uniformly on $\mathcal{Z}$.
				Thus we have
				\begin{equation*}
					\left \Vert \Sigma \left( \widehat{\tau }\right) ^{-1}-\Sigma
					^{-1}\right \Vert \leq \left \Vert \Sigma \left( \widehat{\tau }\right)
					^{-1}\right \Vert \left \Vert \Sigma \left( \widehat{\tau }\right) -\Sigma
					\right \Vert \left \Vert \Sigma ^{-1}\right \Vert =O_p\left(\varsigma(r)\max\left\{\sqrt{d_\tau/n}, \sqrt{\sum_{\ell=1}^{d_\zeta}r_\ell^{-2\kappa_\ell}}\right\}\right).
				\end{equation*}%
				
				And similarly,%
				\begin{eqnarray*}
					&&\left \Vert \left( \frac{1}{n}\Psi ^{\prime }\Sigma \left( \widehat{\tau }%
					\right) ^{-1}\Psi \right) ^{-1}-\left( \frac{1}{n}\Psi ^{\prime }\Sigma
					^{-1}\Psi \right) ^{-1}\right \Vert  \\
					&\leq &\left \Vert \left( \frac{1}{n}\Psi ^{\prime }\Sigma \left( \widehat{%
						\tau }\right) ^{-1}\Psi \right) ^{-1}\right \Vert \left \Vert \frac{1}{n}\Psi
					^{\prime }\left( \Sigma \left( \widehat{\tau }\right) ^{-1}-\Sigma
					^{-1}\right) \Psi \right \Vert \left \Vert \left( \frac{1}{n}\Psi ^{\prime
					}\Sigma ^{-1}\Psi \right) ^{-1}\right \Vert  \\
					&=&O_{p}\left( \left \Vert \Sigma \left( \widehat{\tau }\right) ^{-1}-\Sigma
					^{-1}\right \Vert \right) =O_{p}\left( \varsigma (r)\max \left \{ \sqrt{%
						d_{\tau }/n},\sqrt{\sum_{\ell=1}^{d_\zeta}r_\ell^{-2\kappa_\ell}}\right \}
					\right) .
				\end{eqnarray*}%
				As in the proof of Theorem \ref{thm:stat_appr},
				$n\widehat{m}_{n}=\widehat{\sigma }^{-2}{u}^{\prime }\Sigma \left( \widehat{%
					\tau }\right) ^{-1}\Psi \lbrack \Psi ^{\prime }\Sigma \left( \widehat{\tau }
				\right) ^{-1}\Psi ]^{-1}\Psi ^{\prime }\Sigma \left( \widehat{\tau }\right)
				^{-1}{u}+\widehat{\sigma }^{-2}\sum_{k=1}^4 A_{k},
				$
				where $\gamma $ in the parametric setting is changed to $\tau $ in this
				nonparametric setting. Then, by the MVT,
				\begin{eqnarray*}
					&&\left \vert u^{\prime }\left( \Sigma \left( \widehat{\tau }\right)
					^{-1}\Psi \lbrack \Psi ^{\prime }\Sigma \left( \widehat{\tau }\right)
					^{-1}\Psi ]^{-1}\Psi ^{\prime }\Sigma \left( \widehat{\tau }\right)
					^{-1}-\Sigma^{-1}\Psi \lbrack \Psi ^{\prime }\Sigma^{-1}\Psi
					]^{-1}\Psi ^{\prime }\Sigma^{-1}\right) u\right \vert  \\
					&\leq &2\left(\sup_{t }\left \Vert \frac{1}{\sqrt{n}}u^{\prime }\Sigma \left(
					t \right) ^{-1}\Psi \right \Vert \left \Vert \left( \frac{1}{n}\Psi
					^{\prime }\Sigma \left( t \right) ^{-1}\Psi \right) ^{-1}\right \Vert\right)
					\sum_{j=1}^{d_{\tau }}\left \Vert \frac{1}{\sqrt{n}}\Psi ^{\prime }\left(
					\Sigma \left( \widetilde\tau \right) ^{-1}\Sigma _{j}\left( \widetilde\tau \right) \Sigma \left(\widetilde
					\tau \right) ^{-1}\right) u\right \Vert\\
					&\times& \left \vert \widetilde{\tau }_{j}-\tau
					_{j0}\right \vert
					+2\sup_{t }\left \Vert \frac{1}{\sqrt{n}}u^{\prime }\Sigma \left( t
					\right) ^{-1}\Psi \right \Vert \left \Vert \left( \frac{1}{n}\Psi ^{\prime
					}\Sigma \left( t \right) ^{-1}\Psi \right) ^{-1}\right \Vert \left \Vert
					\frac{1}{\sqrt{n}}\Psi ^{\prime }\left( \Sigma_0-\Sigma\right) u\right \Vert  \\
					&&+\left \Vert \frac{1}{\sqrt{n}}u^{\prime }\Sigma^{-1}\Psi \right \Vert
					^{2}\left \Vert \left( \frac{1}{n}\Psi ^{\prime }\Sigma \left( \widehat{\tau }%
					\right) ^{-1}\Psi \right) ^{-1}-\left( \frac{1}{n}\Psi ^{\prime }\Sigma^{-1}\Psi \right) ^{-1}\right \Vert  \\
					&=&O_{p}(\sqrt{p})O_{p}(d_{\tau }\sqrt{p}\varsigma (r)/\sqrt{n})+O_{p}(\sqrt{%
						p})O_{p}\left(\sqrt{p}\varsigma (r)\sqrt{\sum_{\ell=1}^{d_\zeta}r_\ell^{-2\kappa_\ell}}\right)\\
					&+&O_{p}(p)O_{p}\left( \varsigma (r)\max \left \{ \sqrt{%
						d_{\tau }/n},\sqrt{\sum_{\ell=1}^{d_\zeta}r_\ell^{-2\kappa_\ell}}\right \}
					\right)  \\
					&=&O_{p}\left( p\varsigma (r)\max \left \{ d_{\tau }/\sqrt{n},\sqrt{\sum_{\ell=1}^{d_\zeta}r_\ell^{-2\kappa_\ell}}\right \} \right) =o_{p}(\sqrt{%
						p}),
				\end{eqnarray*}%
				where the last equality holds under the conditions of the theorem. Next, it remains to show $A_{k}=o_{p}(p^{1/2}), k=1,\ldots,4$. The order of $A_{k}$, $k\leq 3$, is the same as
				the parametric case:%
				\begin{eqnarray*}
					\left \vert A_{1}\right \vert  &=&\left \vert {u}^{\prime }\Sigma \left(
					\widehat{\tau }\right) ^{-1}\left( {\theta }_{0}-\widehat{{f}}\right)
					\right \vert  \leq\sup_{\alpha ,t }\left \Vert u^{\prime }\Sigma \left( t\right) ^{-1}\frac{\partial {f}(x,{\alpha })}{\partial \alpha _{j}}%
					\right \Vert \left \vert \alpha _{j}^{\ast }-\widetilde{\alpha }%
					_{j}\right \vert +\frac{p^{1/4}}{n^{1/2}}\sup_{t }\left \Vert u^{\prime
					}\Sigma \left( t \right) ^{-1}h\right \Vert  \\
					&=&O_{p}(\sqrt{n})O_{p}(\frac{1}{\sqrt{n}})+O(\frac{p^{1/4}}{n^{1/2}})O_{p}(%
					\sqrt{n})=O_{p}(p^{1/4})=o_{p}(p^{1/2}),\\
					|A_{2}| &=&\left \vert (u\mathbf{+}\theta _{0}-\widehat{f})^{\prime
					}\left( \Sigma \left( \widehat{\tau }\right) ^{-1}-\Sigma \left( \widehat{%
						\tau }\right) ^{-1}\Psi \lbrack \Psi ^{\prime }\Sigma \left( \widehat{\tau }%
					\right) ^{-1}\Psi ]^{-1}\Psi ^{\prime }\Sigma \left( \widehat{\tau }\right)
					^{-1}\right) e\right \vert  \\
					&\leq &\sup_{t }|u^{\prime }\Sigma \left( t \right) ^{-1}e|+\sup_{t
					}\left \vert u^{\prime }\Sigma \left( t \right) ^{-1}\Psi \lbrack \Psi
					^{\prime }\Sigma \left( t \right) ^{-1}\Psi ]^{-1}\Psi ^{\prime }\Sigma
					\left( t \right) ^{-1}e\right \vert  \\
					&&+\left \Vert {\theta }_{0}-\widehat{{f}}\right \Vert \sup_{t }\left(
					\left \Vert \Sigma \left( t \right) ^{-1}\right \Vert +\left \Vert \Sigma
					\left( t \right) ^{-1}\Psi \lbrack \Psi ^{\prime }\Sigma \left( t
					\right) ^{-1}\Psi ]^{-1}\Psi ^{\prime }\Sigma \left( t \right)
					^{-1}\right \Vert \right) \left\Vert e\right\Vert \\
					&=&O_{p}(p^{-\mu }n^{1/2})+O_{p}(p^{-\mu +1/4}n^{1/2})=O_{p}(p^{-\mu
						+1/4}n^{1/2})=o_{p}(\sqrt{p}),\\
					\left \vert A_{3}\right \vert  &=&\left \vert {u}^{\prime }\Sigma \left(
					\widehat{\tau }\right) ^{-1}\Psi \left( \Psi ^{\prime }\Sigma \left(
					\widehat{\tau }\right) ^{-1}\Psi \right) ^{-1}\Psi ^{\prime }\Sigma \left(
					\widehat{\tau }\right) ^{-1}({\theta }_{0}-\widehat{{f}})\right \vert  \\
					&\leq &\sup_{\alpha,t }\sum_{j=1}^{d_{\alpha }}\left \Vert u^{\prime
					}\Sigma \left( t \right) ^{-1}\Psi \left( \Psi ^{\prime }\Sigma \left(
					t \right) ^{-1}\Psi \right) ^{-1}\Psi ^{\prime }\Sigma \left( t
					\right) ^{-1}\frac{\partial {f}(x,{\alpha })}{\partial \alpha _{j}}%
					\right \Vert \left \vert \alpha _{j}^{\ast }-\widetilde{\alpha }%
					_{j}\right \vert  \\
					&&+\frac{p^{1/4}}{n^{1/2}}\sup_{t }\left \Vert u^{\prime }\Sigma \left(
					t \right) ^{-1}\Psi \left( \Psi ^{\prime }\Sigma \left( t\right)
					^{-1}\Psi \right) ^{-1}\Psi ^{\prime }\Sigma \left( t \right)
					^{-1}h\right \Vert  \\
					&=&O_{p}(1)+O_{p}(p^{1/4})=O_{p}(p^{1/4})=o_{p}(p^{1/2}).
				\end{eqnarray*}%
				However, $A_{4}$ has a different order. Under $H_{\ell }$,
				\begin{eqnarray*}
					A_{4} &=&\left( {\theta }_{0}-\widehat{{f}}\right) ^{\prime }\Sigma
					\left( \widehat{\gamma }\right) ^{-1}\left( {\theta }_{0}-\widehat{{f}}%
					\right)  \\
					&=&\left( {\theta }_{0}-\widehat{{f}}\right) ^{\prime }\Sigma_0
					^{-1}\left( {\theta }_{0}-\widehat{{f}}\right) +\left( {\theta }_{0}-%
					\widehat{{f}}\right) ^{\prime }\left( \Sigma \left( \widehat{\tau }%
					\right) ^{-1}-\Sigma ^{-1}\right) \left( {\theta }_{0}-\widehat{{f}}%
					\right)  \\
					&=&\frac{p^{1/2}}{n}h^{\prime }\Sigma _{0}^{-1}h+o_{p}(1)+O_{p}\left(
					p^{1/2}\right) O_{p}\left( \varsigma (r)\max \left \{ \sqrt{d_{\tau }/n},%
					\sqrt{\sum_{\ell=1}^{d_\zeta}r_\ell^{-2\kappa_\ell}}\right \} \right)  \\
					&=&\frac{p^{1/2}}{n}h^{\prime }\Sigma _{0}^{-1}h+o_{p}(\sqrt{p}),
				\end{eqnarray*}
				where the last equality holds under the conditions of the theorem. Combining these together, we have $
				n\widehat{m}_{n}=\widehat{\sigma }^{-2}\widehat{{v}}^{\prime }\Sigma \left(
				\widehat{\tau }\right) ^{-1}\widehat{{u}}={\sigma _{0}^{-2}}%
				\varepsilon ^{\prime }\fancyv \varepsilon +\left({p^{1/2}}/{n}\right){h}^{\prime
				}\Sigma _{0}^{-1}{h}+o_{p}(\sqrt{p}),$
				under $H_{\ell }$ and the same expression holds with $h=0$ under $H_{0}$.
			\end{proof}
			\begin{proof}[Proof of Theorem \ref{thm:stat_properties_np}]
				Omitted as it is similar to the proof of Theorem \ref{thm:stat_properties}.
			\end{proof}

\allowdisplaybreaks
\setcounter{section}{0}
\renewcommand\thesection{S.\Alph{section}}
\setcounter{equation}{0}
\renewcommand\theequation{S.\Alph{section}.\arabic{equation}}
\setcounter{table}{0}
\renewcommand\thetable{OT.\arabic{table}}
\setcounter{lemma}{0}
\renewcommand\thelemma{LS.\arabic{lemma}}
\setcounter{theorem}{0}
\renewcommand\thetheorem{TS.\arabic{theorem}}
\setcounter{corollary}{0}
\renewcommand\thecorollary{CS.\arabic{corollary}}
\begin{center}
	\Huge {Supplementary online appendix to `Consistent specification testing under spatial dependence'}\\
	\bigskip
	\Large{Abhimanyu Gupta and Xi Qu}
	\\\today
\end{center}
\section{Additional simulation results: Unboundedly supported regressors and asymptotic critical values}
This section provides additional simulation results using the same design as in Section \ref{sec:mc} of the main body of the paper. Recall that the paper reports only bootstrap results for the compactly supported regressors case.  Here we include results using asymptotic critical values for both the compactly and unbounded supported regressor cases, as well as bootstrap results for the latter, focusing on the SARARMA(0,1,0) model. The results are in Tables \ref{table:newsimsapp1}-\ref{table:newsimsapp4} and our findings match those in the main text, with the bootstrap typically offering better size control.

\section{Proofs of Theorems \ref{thm:stat_appr} and \ref{thm:stat_properties}}
\begin{proof}[Proof of Theorem \ref{thm:stat_appr}]
	
	From Corollary \ref{cor:Sigma_equic} and Lemma \ref{lemma:gamma_order}, $\left \Vert \Sigma \left( \widehat{%
		\gamma }\right) -\Sigma \right \Vert =O_{p}\left( \left \Vert \widehat{\gamma}%
	-\gamma _{0}\right \Vert \right) =\sqrt{d_{\gamma }/n}$, so we have, from
	Assumption \ref{ass:Sigma_spec_norm},
	\begin{equation}
		\left \Vert \Sigma \left( \widehat{\gamma }\right) ^{-1}-\Sigma ^{-1}\right
		\Vert \leq \left \Vert \Sigma \left( \widehat{\gamma }\right) ^{-1}\right
		\Vert \left \Vert \Sigma \left( \widehat{\gamma }\right) -\Sigma \right
		\Vert \left \Vert \Sigma ^{-1}\right \Vert =O_{p}\left( \left \Vert \widehat{%
			\gamma}-\gamma _{0}\right \Vert \right) =\sqrt{d_{\gamma }/n}.
		\label{sigma_inv_diff_bound}
	\end{equation}%
	Similarly,
	\begin{eqnarray*}
		&&\left \Vert \left( \frac{1}{n}\Psi ^{\prime }\Sigma \left( \widehat{\gamma
		}\right) ^{-1}\Psi \right) ^{-1}-\left( \frac{1}{n}\Psi ^{\prime }\Sigma
		^{-1}\Psi \right) ^{-1}\right \Vert \\
		&\leq &\left \Vert \left( \frac{1}{n}\Psi ^{\prime }\Sigma \left( \widehat{%
			\gamma }\right) ^{-1}\Psi \right) ^{-1}\right \Vert \left \Vert \frac{1}{n}%
		\Psi ^{\prime }\left( \Sigma \left( \widehat{\gamma }\right) ^{-1}-\Sigma
		^{-1}\right) \Psi \right \Vert \left \Vert \left( \frac{1}{n}\Psi ^{\prime
		}\Sigma ^{-1}\Psi \right) ^{-1}\right \Vert \\
		&\leq &\sup_{\gamma \in \Gamma }\left \Vert \left( \frac{1}{n}\Psi ^{\prime
		}\Sigma \left( \gamma \right) ^{-1}\Psi \right) ^{-1}\right \Vert \left
		\Vert \Sigma \left( \widehat{\gamma }\right) ^{-1}-\Sigma ^{-1}\right \Vert
		\left \Vert \frac{1}{\sqrt{n}}\Psi \right \Vert ^{2}=O_{p}\left( \left \Vert
		\widehat{\gamma}-\gamma _{0}\right \Vert \right) =\sqrt{d_{\gamma }/n}.
	\end{eqnarray*}
	
	By Assumption \ref{ass:alpha_order}, we have $\widehat{\alpha }-\alpha
	^{\ast }=O_{p}(1/\sqrt{n})$. Denote by $\theta^{\ast }(x)=\psi
	(x)^{\prime }\beta^{\ast }$, where $\beta ^{\ast }=\operatorname*{arg\,min}
	_{\beta }\mathcal{E}[y_{i}-\psi (x_{i})^{\prime }\beta )]^{2}$, and set $\theta
	_{ni}=\theta (x_{i})$, $\theta _{0i}=\theta _{0}(x_{i})$, $\widehat{%
		\theta }_{i}=\psi _{i}^{\prime }\widehat{\beta }$, $\widehat{f}%
	_{i}=f(x_{i},\widehat{\alpha })$, $f_{i}^{\ast }=f(x_{i},\alpha
	^{\ast })$. Then $\widehat{u}_{i}=y_{i}-f(x_{i},\widehat{\alpha }%
	)=u_{i}+\theta _{0i}-\widehat{f}_{i}$. Let ${\theta _{0}}=(\theta
	_{0}\left( x_{1}\right) ,\ldots,\theta _{0}\left( x_{n}\right) )^{\prime }$ as
	before, with similar component-wise notation for the $n$-dimensional vectors ${\theta ^{\ast }}$, $\widehat{f}$, and $u$. As the approximation
	error is ${e}={\theta }_{0}-{\theta }^{\ast }={\theta }_{0}-\Psi \beta
	^{\ast }$,%
	\begin{eqnarray*}
		\widehat{{\theta }}-{\theta }^{\ast } &=&\Psi (\widehat{\beta}-\beta
		^{\ast })=\Psi \left( \Psi ^{\prime }\Sigma \left( \widehat{\gamma }%
		\right) ^{-1}\Psi \right) ^{-1}\Psi ^{\prime }\Sigma \left( \widehat{\gamma }%
		\right) ^{-1}({u+\theta }_{0}-\Psi \beta ^{\ast }) \\
		&=&\Psi \left( \Psi ^{\prime }\Sigma \left( \widehat{\gamma }\right)
		^{-1}\Psi \right) ^{-1}\Psi ^{\prime }\Sigma \left( \widehat{\gamma }\right)
		^{-1}({u+e}),
	\end{eqnarray*}%
	so that
	\begin{eqnarray*}
		n\widehat{m}_{n} &=&\widehat{\sigma }^{-2}\widehat{{v}}^{\prime }\Sigma
		\left( \widehat{\gamma }\right) ^{-1}\widehat{{u}}=\widehat{\sigma }^{-2}\left(
		\widehat{{\theta }}-\widehat{f}\right)^{\prime }\Sigma \left( \widehat{%
			\gamma }\right) ^{-1}\left(y-\widehat{f}\right) \\
		&=&\widehat{\sigma }^{-2}\left( \widehat{{\theta }}-{\theta }^{\ast
		}+{\theta }^{\ast }-{\theta }_{0}+{\theta }_{0}-\widehat{{f}}\right)
		^{\prime }\Sigma \left( \widehat{\gamma }\right) ^{-1}\left( {u+\theta }_{0}-%
		\widehat{{f}}\right)  \\
		&=&\widehat{\sigma }^{-2}\left[ \Psi \left( \Psi ^{\prime }\Sigma \left(
		\widehat{\gamma }\right) ^{-1}\Psi \right) ^{-1}\Psi ^{\prime }\Sigma \left(
		\widehat{\gamma }\right) ^{-1}({u+e}{)}-{e}+{\theta }_{0}-\widehat{{f}}%
		\right] ^{\prime }\Sigma \left( \widehat{\gamma }\right) ^{-1}\left( {%
			u+\theta }_{0}-\widehat{{f}}\right)  \\
		&=&\widehat{\sigma }^{-2}{u}^{\prime }\Sigma \left( \widehat{\gamma }\right)
		^{-1}\Psi \lbrack \Psi ^{\prime }\Sigma \left( \widehat{\gamma }\right)
		^{-1}\Psi ]^{-1}\Psi ^{\prime }\Sigma \left( \widehat{\gamma }\right) ^{-1}{u%
		}+\widehat{\sigma }^{-2}{{u}^{\prime }}\Sigma \left( \widehat{\gamma }%
		\right) ^{-1}{\left({\theta }_{0}-\widehat{{f}}\right)} \\
		&{-}&\widehat{\sigma }^{-2}\left({u+\theta }_{0}-\widehat{{f}}\right)^{\prime
		}\Sigma \left( \widehat{\gamma }\right) ^{-1}\left( I-\Psi \lbrack \Psi
		^{\prime }\Sigma \left( \widehat{\gamma }\right) ^{-1}\Psi ]^{-1}\Psi
		^{\prime }\Sigma \left( \widehat{\gamma }\right) ^{-1}\right) {e} \\
		&+&\widehat{\sigma }^{-2}\left({\theta }_{0}-\widehat{{f}}\right)^{\prime }\Sigma
		\left( \widehat{\gamma }\right) ^{-1}\Psi \lbrack \Psi ^{\prime }\Sigma
		\left( \widehat{\gamma }\right) ^{-1}\Psi ]^{-1}\Psi ^{\prime }\Sigma \left(
		\widehat{\gamma }\right) ^{-1}{u} \\
		&+&\widehat{\sigma }^{-2}({\theta }_{0}-\widehat{{f}})^{\prime }\Sigma
		\left( \widehat{\gamma }\right) ^{-1}({\theta }_{0}-\widehat{{f}}) \\
		&=&\widehat{\sigma }^{-2}{u}^{\prime }\Sigma \left( \widehat{\gamma }\right)
		^{-1}\Psi \lbrack \Psi ^{\prime }\Sigma \left( \widehat{\gamma }\right)
		^{-1}\Psi ]^{-1}\Psi ^{\prime }\Sigma \left( \widehat{\gamma }\right) ^{-1}u+\widehat\sigma^{-2}\left(
		A_{1}+A_{2}+A_{3}+A_{4}\right),
	\end{eqnarray*}%
	say. First, for any vector $g$ comprising of conditioned random variables,
	\begin{equation*}
		\mathcal{E}\left[ \left( u^{\prime }\Sigma (\gamma )^{-1}{g}\right) ^{2}\right]
		=g^{\prime }\Sigma (\gamma )^{-1}\Sigma \Sigma (\gamma )^{-1}{g} \leq \sup_{\gamma \in \Gamma }\left \Vert \Sigma (\gamma
		)^{-1}\right \Vert ^{2}\left \Vert \Sigma \right \Vert \left \Vert g\right \Vert
		^{2}=O_p\left(\left \Vert g\right \Vert ^{2}\right),
	\end{equation*}%
	uniformly in $\gamma\in\Gamma$, where the expectation is taken conditional on $g$.
	Similarly,%
	\begin{eqnarray*}
		&&\mathcal{E}\left[ \left( u^{\prime }\Sigma (\gamma )^{-1}\Psi \left( \Psi ^{\prime
		}\Sigma (\gamma )^{-1}\Psi \right) ^{-1}\Psi ^{\prime }\Sigma (\gamma )^{-1}{%
			g}\right) ^{2}\right]  \\
		&=&g^{\prime }\Sigma (\gamma )^{-1}\Psi \left( \Psi ^{\prime }\Sigma (\gamma
		)^{-1}\Psi \right) ^{-1}\Psi ^{\prime }\Sigma (\gamma )^{-1}\Sigma \Sigma
		(\gamma )^{-1}\Psi \left( \Psi ^{\prime }\Sigma (\gamma )^{-1}\Psi \right)
		^{-1}\Psi ^{\prime }\Sigma (\gamma )^{-1}{g} \\
		&\leq &\sup_{\gamma \in \Gamma }\left \Vert \Sigma (\gamma )^{-1}\right \Vert
		^{4}\left \Vert \Sigma \right \Vert \left \Vert \frac{1}{n}\Psi \left( \frac{1}{%
			n}\Psi ^{\prime }\Sigma (\gamma )^{-1}\Psi \right) ^{-1}\Psi ^{\prime
		}\right \Vert ^{2}\left \Vert g\right \Vert ^{2}=O_p\left(\left \Vert g\right \Vert ^{2}\right),
	\end{eqnarray*}%
	uniformly and, for any $j=1$, \ldots, $d_{\gamma }$,%
	\begin{eqnarray*}
		\mathcal{E}\left[ \left( u^{\prime }\Sigma (\gamma )^{-1}\Sigma _{j}\left( \gamma
		\right) \Sigma \left( \gamma \right) ^{-1}{g}\right) ^{2}\right] &=&g^{\prime }\Sigma (\gamma )^{-1}\Sigma _{j}\left( \gamma \right) \Sigma
		\left( \gamma \right) ^{-1}\Sigma \Sigma (\gamma )^{-1}\Sigma _{j}\left(
		\gamma \right) \Sigma \left( \gamma \right) ^{-1}{g}\\
		&\leq& \sup_{\gamma \in
			\Gamma }\left \Vert \Sigma (\gamma )^{-1}\right \Vert ^{4}\left \Vert \Sigma
		_{j}\left( \gamma \right) \right \Vert ^{2}\left \Vert \Sigma \right \Vert
		\left \Vert g\right \Vert ^{2}=O_p\left(\left \Vert g\right \Vert ^{2}\right).
	\end{eqnarray*}%
	Let $\Psi _{k}$ be the $k$-th column of $%
	\Psi $, $k=1,\ldots,p$. Then, we have $\left\Vert\Psi _{k}/\sqrt{n}\right\Vert=O_p(1)$
	and for any $\gamma \in \Gamma $,%
	\begin{eqnarray*}
		\mathcal{E}\left \Vert \frac{1}{\sqrt{n}}u^{\prime }\Sigma \left( \gamma \right)
		^{-1}\Psi \right \Vert^2  &\leq &{\sum_{k=1}^{p}\mathcal{E}\left( u^{\prime }\Sigma
			\left( \gamma \right) ^{-1}\frac{1}{\sqrt{n}}\Psi _{k}\right) ^{2}}=O_p\left({p
		}\right), \\
		\mathcal{E}\left \Vert \frac{1}{\sqrt{n}}u^{\prime }\Sigma \left( \gamma \right)
		^{-1}\Sigma _{j}\left( \gamma \right) \Sigma \left( \gamma \right) ^{-1}\Psi
		\right \Vert^2  &\leq &{\sum_{k=1}^{p}\mathcal{E}\left( u^{\prime }\Sigma \left(
			\gamma \right) ^{-1}\Sigma _{j}\left( \gamma \right) \Sigma \left( \gamma
			\right) ^{-1}\frac{1}{\sqrt{n}}\Psi _{k}\right) ^{2}}=O({p}).
	\end{eqnarray*}%
	Therefore, by Chebyshev's inequality,
	\begin{equation*}
		\sup_{\gamma \in \Gamma }\left \Vert \frac{1}{\sqrt{n}}u^{\prime }\Sigma
		\left( \gamma \right) ^{-1}\Psi \right \Vert =O_{p}(\sqrt{p})\text{ \  \ and \
			\ }\sup_{\gamma \in \Gamma }\left \Vert \frac{1}{\sqrt{n}}u^{\prime }\Sigma
		\left( \gamma \right) ^{-1}\Sigma _{j}\left( \gamma \right) \Sigma \left(
		\gamma \right) ^{-1}\Psi \right \Vert =O_{p}(\sqrt{p}).
	\end{equation*}%
	By the decomposition%
	\begin{eqnarray*}
		&&u^{\prime }\left( \Sigma \left( \widehat{\gamma }\right) ^{-1}\Psi \lbrack
		\Psi ^{\prime }\Sigma \left( \widehat{\gamma }\right) ^{-1}\Psi ]^{-1}\Psi
		^{\prime }\Sigma \left( \widehat{\gamma }\right) ^{-1}-\Sigma ^{-1}\Psi
		\lbrack \Psi ^{\prime }\Sigma ^{-1}\Psi ]^{-1}\Psi ^{\prime }\Sigma
		^{-1}\right) u \\
		&=&u^{\prime }\left( \Sigma \left( \widehat{\gamma }\right) ^{-1}+\Sigma
		^{-1}\right) \Psi \lbrack \Psi ^{\prime }\Sigma \left( \widehat{\gamma }%
		\right) ^{-1}\Psi ]^{-1}\Psi ^{\prime }\left(
		\sum_{i=1}^{n}e_{in}e_{in}^{\prime }\right) \left( \Sigma \left( \widehat{%
			\gamma }\right) ^{-1}-\Sigma ^{-1}\right) u \\
		&&+u^{\prime }\Sigma ^{-1}\Psi \left( \lbrack \Psi ^{\prime }\Sigma \left(
		\widehat{\gamma }\right) ^{-1}\Psi ]^{-1}-[\Psi ^{\prime }\Sigma ^{-1}\Psi
		]^{-1}\right) \Psi ^{\prime }\Sigma ^{-1}u \\
		&=&u^{\prime }\left( \Sigma \left( \widehat{\gamma }\right) ^{-1}+\Sigma
		^{-1}\right) \Psi \lbrack \Psi ^{\prime }\Sigma \left( \widehat{\gamma }%
		\right) ^{-1}\Psi ]^{-1}\Psi ^{\prime }\left(
		\sum_{i=1}^{n}e_{in}e_{in}^{\prime }\right) \sum_{j=1}^{d_{\gamma }}\left(
		\Sigma \left( \widetilde{\gamma }\right) ^{-1}\Sigma _{j}\left( \widetilde{%
			\gamma }\right) \Sigma \left( \widetilde{\gamma }\right) ^{-1}\right)\\
		&\times& u (%
		\widetilde{\gamma }_{j}-\gamma _{j0})
		+u^{\prime }\Sigma ^{-1}\Psi \left( \lbrack \Psi ^{\prime }\Sigma \left(
		\widehat{\gamma }\right) ^{-1}\Psi ]^{-1}-[\Psi ^{\prime }\Sigma ^{-1}\Psi
		]^{-1}\right) \Psi ^{\prime }\Sigma ^{-1}u,
	\end{eqnarray*}%
	where $e_{in}$ is an $n\times 1$ vector with $i$-th entry one and zeros elsewhere, so $\sum_{i=1}^{n}e_{in}e_{in}^{\prime }=I_{n}$,
	and%
	\begin{eqnarray*}
		e_{in}^{\prime }\left( \Sigma \left( \widehat{\gamma }\right) ^{-1}-\Sigma
		^{-1}\right) u &=&\sum_{j=1}^{d_{\gamma }}e_{in}^{\prime }\left( \Sigma
		\left( \widetilde{\gamma }\right) ^{-1}\Sigma _{j}\left( \widetilde{\gamma }%
		\right) \Sigma \left( \widetilde{\gamma }\right) ^{-1}\right) u(\widetilde{%
			\gamma }_{j}-\gamma _{j0}) \\
		&=&e_{in}^{\prime }\sum_{j=1}^{d_{\gamma }}\left( \Sigma \left( \widetilde{%
			\gamma }\right) ^{-1}\Sigma _{j}\left( \widetilde{\gamma }\right) \Sigma
		\left( \widetilde{\gamma }\right) ^{-1}\right) u(\widetilde{\gamma }%
		_{j}-\gamma _{j0})
	\end{eqnarray*}%
	where $\widetilde{\gamma }$ is a value between $\widehat{\gamma }$ and $%
	\gamma _{0}$ due to the mean value theorem. We have
	\begin{eqnarray*}
		&&\left \vert u^{\prime }\left( \Sigma \left( \widehat{\gamma }\right)
		^{-1}\Psi \lbrack \Psi ^{\prime }\Sigma \left( \widehat{\gamma }\right)
		^{-1}\Psi ]^{-1}\Psi ^{\prime }\Sigma \left( \widehat{\gamma }\right)
		^{-1}-\Sigma ^{-1}\Psi \lbrack \Psi ^{\prime }\Sigma ^{-1}\Psi ]^{-1}\Psi
		^{\prime }\Sigma ^{-1}\right) u\right \vert  \\
		&\leq &2\sup_{\gamma \in \Gamma }\left \Vert \frac{1}{\sqrt{n}}u^{\prime
		}\Sigma \left( \gamma \right) ^{-1}\Psi \right \Vert \left \Vert \left( \frac{1%
		}{n}\Psi ^{\prime }\Sigma \left( \gamma \right) ^{-1}\Psi \right)
		^{-1}\right \Vert \sum_{j=1}^{d_{\gamma }}\left \Vert \frac{1}{\sqrt{n}}\Psi
		^{\prime }\left( \Sigma \left( \gamma \right) ^{-1}\Sigma _{j}\left( \gamma
		\right) \Sigma \left( \gamma \right) ^{-1}\right) u\right \Vert\\
		&\times& \left \vert
		\widetilde{\gamma }_{j}-\gamma _{j0}\right \vert
		+\left \Vert \frac{1}{\sqrt{n}}u^{\prime }\Sigma ^{-1}\Psi \right \Vert
		^{2}\left \Vert \left( \frac{1}{n}\Psi ^{\prime }\Sigma \left( \widehat{%
			\gamma }\right) ^{-1}\Psi \right) ^{-1}-\left( \frac{1}{n}\Psi ^{\prime
		}\Sigma ^{-1}\Psi \right) ^{-1}\right \Vert  \\
		&=&O_{p}(\sqrt{p}) O_{p}(d_{\gamma }\sqrt{p}/\sqrt{n})+O_{p}(p)
		O_{p}(\sqrt{d_{\gamma }}/\sqrt{n})=O_{p}(d_{\gamma }p/\sqrt{n})=o_{p}(\sqrt{p%
		}),
	\end{eqnarray*}%
	where the last equality holds under the conditions of the theorem.
	
	It remains to show that
	\begin{equation}
		A_{i}=o_{p}\left( {p^{1/2}}\right) ,i=1,\ldots ,4.  \label{As_neg}
	\end{equation}%
	It is convenient to perform the calculations under $H_{\ell }$, which covers
	$H_{0}$ as a particular case.
	Using the mean value theorem and either $H_{0}$ or $H_{\ell }$, we can
	express%
	\begin{equation}
		{\theta }_{0i}-\widehat{{f}}_{i}={f}_{i}^{\ast }-\widehat{{f}}%
		_{i}-(p^{1/4}/n^{1/2}){h_{i}}=\sum_{j=1}^{d_{\alpha }}\frac{\partial {f}%
			(x_{i},\widetilde{\alpha })}{\partial \alpha _{j}}(\alpha _{j}^{\ast }-%
		\widetilde{\alpha }_{j})-\frac{p^{1/4}}{n^{1/2}}{h_{i},}
		\label{theta_f_diff}
	\end{equation}%
	where $\widetilde{\alpha }_{j}$ is a value between $\alpha _{j}^{\ast }$
	and $\widehat{\alpha }_{j}$. Then, for any $j=1,\ldots,d_{\alpha }$, $%
	\left
	\vert \alpha _{j}^{\ast }-\widetilde{\alpha }_{j}\right \vert {=}%
	O_{p}(1/\sqrt{n})$. Based on
	\begin{equation*}
		\sup_{\gamma \in \Gamma }\left \vert u^{\prime }\Sigma (\gamma )^{-1}\Psi
		\left( \Psi ^{\prime }\Sigma (\gamma )^{-1}\Psi \right) ^{-1}\Psi ^{\prime
		}\Sigma (\gamma )^{-1}{g}\right \vert {=}O_{p}\left( \left \Vert g\right
		\Vert \right) \text{ and }\sup_{\gamma \in \Gamma }\left \vert u^{\prime
		}\Sigma (\gamma )^{-1}g\right \vert =O_{p}\left( \left \Vert g\right \Vert
		\right)
	\end{equation*}%
	for any $\gamma \in \Gamma $ and any conditioned vector $g$, if we take $g={\partial {f}(x,{\alpha })}/{\partial \alpha _{j}}$ or $g=h$, then both satisfy $O_p\left(\left \Vert g\right \Vert\right )=O_p\left(\sqrt{n}\right)$ and it follows that
	\begin{eqnarray*}
		\left \vert A_{1}\right \vert &=&\left \vert {u}^{\prime }\Sigma \left(
		\widehat{\gamma }\right) ^{-1}\left( {\theta }_{0}-\widehat{{f}}\right)
		\right \vert \leq\sup_{\gamma ,\alpha }\sum_{j=1}^{d_{\alpha }}\left \Vert u^{\prime
		}\Sigma (\gamma )^{-1}\frac{\partial {f}(x,{\alpha })}{\partial \alpha _{j}}%
		\right \Vert \left \vert \alpha _{j}^{\ast }-\widetilde{\alpha }_{j}\right
		\vert +\frac{p^{1/4}}{n^{1/2}}\sup_{\gamma }\left \Vert u^{\prime }\Sigma
		(\gamma )^{-1}h\right \Vert \\
		&=&O_{p}(\sqrt{n})O_{p}\left(\frac{1}{\sqrt{n}}\right)+O\left(\frac{p^{1/4}}{n^{1/2}}\right)O_{p}(%
		\sqrt{n})=O_{p}(p^{1/4})=o_{p}(p^{1/2}).
	\end{eqnarray*}%
	Similarly,%
	\begin{eqnarray*}
		\left \vert A_{3}\right \vert &=&\left \vert {u}^{\prime }\Sigma \left(
		\widehat{\gamma }\right) ^{-1}\Psi \left( \Psi ^{\prime }\Sigma \left(
		\widehat{\gamma }\right) ^{-1}\Psi \right) ^{-1}\Psi ^{\prime }\Sigma \left(
		\widehat{\gamma }\right) ^{-1}({\theta }_{0}-\widehat{{f}})\right \vert
		\\
		&\leq &\sup_{\gamma ,\alpha }\sum_{j=1}^{d_{\alpha }}\left \Vert u^{\prime
		}\Sigma \left( \widehat{\gamma }\right) ^{-1}\Psi \left( \Psi ^{\prime
		}\Sigma \left( \widehat{\gamma }\right) ^{-1}\Psi \right) ^{-1}\Psi ^{\prime
		}\Sigma \left( \widehat{\gamma }\right) ^{-1}\frac{\partial {f}(x,{\alpha })%
		}{\partial \alpha _{j}}\right \Vert \left \vert \alpha _{j}^{\ast }-%
		\widetilde{\alpha }_{j}\right \vert \\
		&&+\frac{p^{1/4}}{n^{1/2}}\sup_{\gamma }\left \Vert u^{\prime }\Sigma \left(
		\widehat{\gamma }\right) ^{-1}\Psi \left( \Psi ^{\prime }\Sigma \left(
		\widehat{\gamma }\right) ^{-1}\Psi \right) ^{-1}\Psi ^{\prime }\Sigma \left(
		\widehat{\gamma }\right) ^{-1}h\right \Vert \\
		&=&O_{p}(1)+O_{p}(p^{1/4})=O_{p}(p^{1/4})=o_{p}(p^{1/2}).
	\end{eqnarray*}%
	Also, by Assumptions \ref{ass:alpha_order} and \ref{ass:alpha_der}, we have
	\begin{equation}
		\left \Vert {\theta }_{0}-\widehat{{f}}\right \Vert \leq \sup_{\alpha
		}\sum_{j=1}^{d_{\alpha }}\left \Vert \frac{\partial {f}(x,{\alpha })}{%
			\partial \alpha _{j}}\right \Vert \left \vert \alpha _{j}^{\ast }-%
		\widetilde{\alpha }_{j}\right \vert +\left \Vert h\right \Vert \frac{p^{1/4}%
		}{n^{1/2}}=O_{p}(p^{1/4}).  \label{theta_f_diff_norm}
	\end{equation}%
	By (\ref{appr_error}), we have $\left\Vert e\right\Vert=O(p^{-\mu }n^{1/2})$ and%
	\begin{eqnarray*}
		|A_{2}| &=&\left \vert (u\mathbf{+}\theta _{0}-\widehat{f})^{\prime
		}\left( \Sigma \left( \widehat{\gamma }\right) ^{-1}-\Sigma \left( \widehat{%
			\gamma }\right) ^{-1}\Psi \lbrack \Psi ^{\prime }\Sigma \left( \widehat{%
			\gamma }\right) ^{-1}\Psi ]^{-1}\Psi ^{\prime }\Sigma \left( \widehat{\gamma
		}\right) ^{-1}\right) e\right \vert \\
		&\leq &\sup_{\gamma }|u^{\prime }\Sigma \left( \gamma \right)
		^{-1}e|+\sup_{\gamma }\left \vert u^{\prime }\Sigma \left( \gamma \right)
		^{-1}\Psi \lbrack \Psi ^{\prime }\Sigma \left( \gamma \right) ^{-1}\Psi
		]^{-1}\Psi ^{\prime }\Sigma \left( \gamma \right) ^{-1}e\right \vert \\
		&&+\left \Vert {\theta }_{0}-\widehat{{f}}\right \Vert \sup_{\gamma
		}\left( \left \Vert \Sigma \left( \gamma \right) ^{-1}\right \Vert +\left
		\Vert \Sigma \left( \gamma \right) ^{-1}\Psi \lbrack \Psi ^{\prime }\Sigma
		\left( \gamma \right) ^{-1}\Psi ]^{-1}\Psi ^{\prime }\Sigma \left( \gamma
		\right) ^{-1}\right \Vert \right) \left\Vert e\right\Vert \\
		&=&O_{p}(p^{-\mu }n^{1/2})+O_{p}(p^{-\mu +1/4}n^{1/2})=O_{p}(p^{-\mu
			+1/4}n^{1/2})=o_{p}(\sqrt{p}).
	\end{eqnarray*}%
	where the last equality holds under the conditions of the theorem. Finally, under $H_{\ell }$,
	\begin{eqnarray*}
		A_{4} &=&\left( {\theta }_{0}-\widehat{{f}}\right) ^{\prime }\Sigma
		\left( \widehat{\gamma }\right) ^{-1}\left( {\theta }_{0}-\widehat{{f}}%
		\right) \\
		&=&\left( {\theta }_{0}-\widehat{{f}}\right) ^{\prime }\Sigma
		^{-1}\left( {\theta }_{0}-\widehat{{f}}\right) +\left( {\theta }_{0}-%
		\widehat{{f}}\right) ^{\prime }\left( \Sigma \left( \widehat{\gamma }%
		\right) ^{-1}-\Sigma ^{-1}\right) \left( {\theta }_{0}-\widehat{{f}}%
		\right) \\
		&=&\frac{p^{1/2}}{n}h^{\prime }\Sigma^{-1}h+o_{p}(1)+O_{p}\left( p^{1/2}d_{\gamma
		}^{1/2}/n^{1/2}\right) =\frac{p^{1/2}}{n}h^{\prime}\Sigma^{-1}h+o_{p}(\sqrt{p}).
	\end{eqnarray*}
	
	Combining these together, we have%
	\begin{equation*}
		n\widehat{m}_{n}=\widehat{\sigma }^{-2}\widehat{{v}}^{\prime }\Sigma \left(
		\widehat{\gamma }\right) ^{-1}\widehat{{u}}=\frac{1}{\sigma _{0}^{2}}%
		\varepsilon ^{\prime }\fancyv \varepsilon {+}\frac{p^{1/2}}{n}{h}^{\prime
		}\Sigma ^{-1}{h}+o_{p}(\sqrt{p}),
	\end{equation*}%
	under $H_{\ell }$ and the same expression holds with $h=0$ under $H_{0}$.

\end{proof}
\begin{proof}[Proof of Theorem \ref{thm:stat_properties}]
	(1) Follows from Theorems \ref{thm:stat_appr} and \ref{thm:appr_clt}.
	(2)
	Following reasoning analogous to the proofs of Theorems \ref{thm:stat_appr} and \ref{thm:appr_clt}, it can be
	shown that under $H_{1}$,
	$\widehat{m}_{n}=n^{-1}{\sigma }^{*-2}(\theta _{0}-f_{}^{\ast })^{\prime }\Sigma \left(
	\gamma ^{\ast }\right) ^{-1}(\theta _{0}-f^{\ast })+o_{p}(1).$ Then,
	\[
	\fancyt=\left(n\widehat{m}_{n}-p\right)/{\sqrt{2p}}=\left({n}/{\sqrt{p}}\right){(\theta _{0}-f^{\ast })^{\prime }\Sigma \left(
		\gamma ^{\ast }\right) ^{-1}(\theta _{0}-f^{\ast })}/\left({\sqrt{2}n\sigma
		^{\ast 2}}\right)+o_{p}\left({n}/{\sqrt{p}}\right)
	\]
	and for any nonstochastic sequence $\{C_{n}\}$, $C_{n}=o(n/p^{1/2})$, $P(\fancyt>C_{n})\rightarrow 1,$ so that consistency follows.
	(3) Follows from Theorems \ref{thm:stat_appr} and \ref{thm:appr_clt}.
\end{proof}
\section{Proof of Theorem \ref{thm:cons_SAR}}
\begin{proof} We prove the result under $H_1$, which is the more challenging case as it involves nonparametric estimation. The proof under $H_0$ is similar. We will show $\widehat\phi\overset{p}\rightarrow \phi_0$, whence $\widehat{\beta}\overset{p}\rightarrow \beta_0$ and $\widehat{\sigma}^2\overset{p}\rightarrow \sigma^2_0$ follow from (\ref{betapmle_SAR}) and (\ref{sigmapmleconc_SAR}) respectively. First note that
	\begin{equation}\label{cons1}
		\mathcal{L}\left( \phi \right) -\mathcal{L} =\log \overline{\sigma }^{2}\left( \phi
		\right) /\overline{\sigma }^{2}-n^{-1}\log \left\vert T'(\lambda)\Sigma(\gamma)^{-1}T(\lambda)\Sigma\right\vert =\log \overline{\sigma }^{2}\left( \phi \right) /\sigma ^{2}\left(
		\phi \right) -\log \overline{\sigma }^{2}/\sigma _{0}^{2}+\log r(\phi
		),
	\end{equation}
	where recall that
	$
	\sigma ^{2}\left( \phi \right) =n^{-1}\sigma_0^2tr\left(T'(\lambda)\Sigma(\gamma)^{-1}T(\lambda)\Sigma\right),\text{ }\overline{\sigma }^{2}=\overline{\sigma }^{2}\left( \phi
	_{0}\right) =n^{-1}u^{\prime } E'M Eu,
	$
	using (\ref{sigmapmleconc_SAR}) and also  $r(\phi )=n^{-1}tr\left(T'(\lambda)\Sigma(\gamma)^{-1}T(\lambda)\Sigma\right)/\left\vert T'(\lambda)\Sigma(\gamma)^{-1}T(\lambda)\Sigma\right\vert ^{1/n}$.
	
	We have
	$\overline{\sigma }^{2}\left( \phi \right) =n^{-1}\left\{ S^{-1^{\prime
	}}\left( \Psi\beta _{0}+u\right) \right\} ^{\prime }S^{\prime }(\lambda
	) E(\gamma)'M\left( \gamma \right)  E(\gamma)S(\lambda )S^{-1}\left( \Psi\beta _{0}+u\right)=c_1\left( \phi \right) +c_2\left( \phi \right) +c_3\left( \phi \right)$,
	where
	\begin{eqnarray*}
		c_1\left( \phi \right) &=&n^{-1}\beta _{0}^{\prime }\Psi^{\prime }T^{\prime
		}(\lambda ) E(\gamma)'M\left( \gamma \right) E(\gamma) T(\lambda )\Psi\beta _{0}, \\
		\text{\ }c_2\left( \phi \right) &=&n^{-1}\sigma _{0}^{2}tr\left( T^{\prime
		}(\lambda ) E(\gamma)'M\left( \gamma \right) E(\gamma) T(\lambda )\Sigma\right) , \\
		c_3\left( \phi \right) &=&n^{-1}tr\left( T^{\prime }(\lambda ) E(\gamma)'M\left( \gamma
		\right)  E(\gamma)T(\lambda )\left(uu^{\prime }-\sigma _{0}^{2}\Sigma\right)\right)\\
		&+&2n^{-1}\beta
		_{0}^{\prime }\Psi^{\prime }T^{\prime }(\lambda ) E(\gamma)'M\left( \gamma \right)
		E(\gamma)T(\lambda )u.
	\end{eqnarray*}
	Note that in the particular cases of Theorems \ref{thm:consistency} and \ref{thm:consistency_np}, where $T(\lambda)=I_n$, the $c_1$ term vanishes because $M\left( \gamma \right) E(\gamma) \Psi=0$ and $M\left( \tau \right) E(\tau) \Psi=0$. Proceeding with the current, more general proof
	\begin{eqnarray*}
		\log \frac{\overline{\sigma }^{2}\left( \phi \right) }{\sigma ^{2}\left(
			\phi \right) } &=&\log \frac{\overline{\sigma }^{2}\left( \phi \right)
		}{\left( c_1\left( \phi \right) +c_2\left( \phi \right) \right) }+\log \frac{%
			c_1\left( \phi \right) +c_2\left( \phi \right) }{\sigma ^{2}\left( \phi
			\right) } \\
		&=&\log \left( 1+\frac{c_3\left( \phi \right) }{c_1\left( \phi \right)
			+c_2\left( \phi \right) }\right) +\log \left( 1+\frac{c_1\left( \phi \right)
			-f\left( \phi \right) }{\sigma^{2}\left( \phi \right) }\right) ,
	\end{eqnarray*}%
	where $f\left( \phi \right) =n^{-1}\sigma _{0}^{2}tr\left( E'^{-1}T^{\prime }(\lambda
	) E(\gamma)' \left( I_n-M\left( \gamma \right) \right) E(\gamma) T(\lambda ) E^{-1}\right).$
	Then (\ref{cons1}) implies
	\begin{eqnarray*}
		P\left( \left\Vert \widehat{\phi }-\phi _{0}\right\Vert \in \overline{%
			\mathcal{N}}^{\;\phi }\left( \eta \right) \right) &=&P\left( \inf_{\phi
			\in \text{ }\overline{\mathcal{N}}^{\;\phi }\left( \eta \right) }\mathcal{L}\left(
		\phi \right) -\mathcal{L}\leq 0\right) \\
		&\leq &P\left( \log \left( 1+\underset{\phi \in \text{ }\overline{\mathcal{%
					N}}^{\;\phi }\left( \eta \right) }{\sup }\left\vert \frac{c_3\left( \phi
			\right) }{c_1\left( \phi \right) +c_2\left( \phi \right) }\right\vert
		\right) +\left\vert \log \left(\overline{\sigma }^{2}/\sigma _{0}^{2}\right)\right\vert
		\right. \\
		&&\left. \geq \inf_{\phi \in \text{ }\overline{\mathcal{N}}^{\;\phi
			}\left( \eta \right) }\left( \log \left( 1+\frac{c_1\left( \phi \right)
			-f\left( \phi \right) }{\sigma ^{2}\left( \phi \right) }\right) +\log
		r(\phi )\right) \right) ,
	\end{eqnarray*}%
	where recall that $\overline{\mathcal{N}}^{\;\phi }\left( \eta \right) =\Phi
	\backslash \mathcal{N}^{\phi }\left( \eta \right) ,$ $%
	\mathcal{N}^{\phi }\left( \eta \right) =\left\{ \phi :\left\Vert
	\phi -\phi _{0}\right\Vert <\eta\right\}\cap\Phi .$
	Because $\overline{\sigma }^{2}/\sigma
	_{0}^{2}\overset{p}\rightarrow 1,$ the property $\log \left( 1+x\right) =x+o\left(
	x\right) $ as $x\rightarrow 0$ implies that it is sufficient to show that
	\begin{eqnarray}
		\underset{\phi \in \text{ }\overline{\mathcal{N}}^{\;\phi }\left( \eta
			\right) }{\sup }\left\vert \frac{c_3\left( \phi \right) }{c_1\left( \phi
			\right) +c_2\left( \phi \right) }\right\vert &\overset{p}\longrightarrow&\text{ }0, \label{cons2}\\
		\underset{\phi \in \text{ }\overline{\mathcal{N}}^{\;\phi }\left( \eta
			\right) }{\sup }\left\vert \frac{f\left( \phi \right) }{\sigma ^{2}\left(
			\phi \right) }\right\vert &\overset{p}\longrightarrow &\text{ }0\label{cons3}, \\
		P\left( \inf_{\phi \in \text{ }\overline{\mathcal{N}}^{\;\phi }\left(
			\eta \right) }\left\{ \frac{c_1\left( \phi \right) }{\sigma ^{2}\left(
			\phi \right) }+\log r(\phi )\right\}>0\right)  &\longrightarrow &\text{ }1\label{cons4}.
	\end{eqnarray}%
	Because $\overline{\mathcal{N}}^{\;\phi }\left( \eta \right) \subseteq \left\{
	\Lambda \times \overline{\mathcal{N}}^{\;\gamma }\left( \eta /2\right)
	\right\} \cup \left\{ \overline{\mathcal{N}}^{\;\lambda }\left( \eta /2\right)
	\times \Gamma\right\} $, we have
	\begin{eqnarray*}
		P\left(\inf_{\phi \in \text{ }\overline{\mathcal{N}}^{\;\phi }\left( \eta \right)
		}\left\{ \frac{c_1\left( \phi \right) }{\sigma ^{2}\left( \phi \right) }%
		+\log r(\phi )\right\}>0\right) &\geq &P\left(\min \left\{ \underset{\Lambda \times
			\overline{\mathcal{N}}^{\;\gamma }\left( \eta /2\right) }{\inf }\frac{c_1\left(
			\phi \right) }{\sigma ^{2}\left( \phi \right) },\underset{\overline{%
				\mathcal{N}}^{\;\lambda }\left( \eta /2\right) }{\inf }\log r(\phi )\right\}>0\right)
		\\
		&\geq &P\left(\min \left\{ \underset{\Lambda \times \overline{\mathcal{N}}^{\;\gamma
			}\left( \eta /2\right) }{\inf }\frac{c_1\left( \phi \right) }{C},\underset{%
			\overline{\mathcal{N}}^{\lambda }\left( \eta /2\right) }{\inf }\log
		r(\phi )\right\}>0\right) ,
	\end{eqnarray*}%
	from Assumption \ref{ass:sigma_unif_SAR}, whence Assumptions \ref{ass:tau_lambda_ident} and \ref{ass:reg_ident_SAR} imply (\ref{cons4}). Again using Assumption \ref{ass:sigma_unif_SAR}, uniformly in $\phi $, $\left\vert f\left( \phi \right) /\sigma ^{2}\left( \phi
	\right) \right\vert =O_p\left( \left\vert f\left( \phi \right) \right\vert\right)$
	and
	\begin{eqnarray}
		\left\vert f\left( \phi \right) \right\vert &= &O_p\left(tr\left( E'^{-1}T^{\prime
		}(\lambda )\Sigma(\gamma)^{-1}\Psi \left( \Psi^{\prime }\Sigma(\gamma)^{-1}\Psi \right) ^{-1}\Psi^{\prime }\Sigma(\gamma)^{-1}
		T(\lambda ) E^{-1}\right) /n \right)\nonumber\\
		&=&O_p\left( tr\left(  E'^{-1} T^{\prime
		}(\lambda )\Sigma(\gamma)^{-1}\Psi
		\Psi^{\prime }\Sigma(\gamma)^{-1}
		T(\lambda ) E^{-1}\right) /n^{2}\right)=O_p\left(\left\Vert \Psi^{\prime }\Sigma(\gamma)^{-1}
		T(\lambda ) E^{-1}/n\right\Vert_F^2\right)\nonumber\\
		&=& O_p\left(\left\Vert \Psi/n\right\Vert_F^2\overline\varphi^2\left(\Sigma(\gamma)^{-1}\right)\left\Vert
		T(\lambda )\right\Vert^2\left\Vert
		E^{-1}\right\Vert^2\right)= O_p\left(\left\Vert \Psi/n\right\Vert_F^2\left\Vert
		T(\lambda )\right\Vert^2\overline\varphi\left(\Sigma\right)/\underline\varphi^2\left(\Sigma(\gamma)\right)\right)\nonumber\\
		&=&O_p\left(\left\Vert
		T(\lambda )\right\Vert^2/n\right),\label{cons_neg}
	\end{eqnarray}%
	where we have twice made use of the inequality
	\begin{equation}\label{frob_norm_inequality}
		\left\Vert AB\right\Vert_F\leq\left\Vert A\right\Vert_F\left\Vert B\right\Vert
	\end{equation}
	for generic multiplication compatible matrices $A$ and $B$. (\ref{cons3}) now follows by Assumption \ref{ass:spec_norm_G} and compactness of $\Lambda$ because $T(\lambda)=I_n+\sum_{j=1}^{d_{\lambda}}\left(\lambda_{0j}-\lambda_j\right)G_j$. Finally consider (\ref{cons2}). We first prove pointwise convergence. For any
	fixed $\phi \in \overline{\mathcal{N}}^{\;\phi }\left( \eta \right) $ and large enough $n$, Assumptions \ref{ass:sigma_unif_SAR} and \ref{ass:reg_ident_SAR} imply
	\begin{eqnarray}
		\left\{c_1\left( \phi \right)\right\}^{-1} &=&O_p\left( \left\Vert \beta
		_{0}\right\Vert ^{-2}\right)=O_p(1)\label{c_lower_bound}\\
		\left\{c_2\left( \phi \right)\right\}^{-1} &=&O_p(1),\label{d_lower_bound}
	\end{eqnarray}
	because $\left\{n^{-1}\sigma _{0}^{2}tr\left(  T^{\prime }(\lambda )\Sigma(\gamma)^{-1}T(\lambda
	) E^{-1}\right)\right\}^{-1} =O_p(1)$ and, proceeding like in the bound for $\left\vert f (\phi)\right\vert$, $tE'^{-1}r\left( E'^{-1} T^{\prime }(\lambda ) E(\gamma)'\left( I-M\left(
	\gamma \right) \right) E(\gamma) T(\lambda ) E^{-1}\right) =O_p\left(\left\Vert
	T(\lambda )\right\Vert^2/n\right)=O_p\left(1/n\right)$. In fact it is worth noting for the equicontinuity argument presented later that Assumptions \ref{ass:sigma_unif_SAR} and \ref{ass:reg_ident_SAR} actually imply that (\ref{c_lower_bound}) and (\ref{d_lower_bound}) hold uniformly over $\overline{\mathcal{N}}^\phi(\eta)$, a property not needed for the present pointwise arguments.
	Thus
	$c_3\left( \phi \right) /\left( c_1\left( \phi \right) +c_2\left( \phi
	\right) \right) =O_p\left( \left\vert c_3\left( \phi \right) \right\vert
	\right)$ where, writing $\mathfrak{B}(\phi)=T^{\prime }(\lambda ) E(\gamma)'M\left( \gamma \right)
	E(\gamma)T(\lambda )$ with typical element $\mathfrak{b}_{rs}(\phi)$, $r,s=1,\ldots,n$, $c_3\left( \phi \right) $ has mean $0$ and variance%
	\begin{equation}\label{e_var_bound}
		O_p\left( \frac{\left\Vert \mathfrak{B}(\phi)\Sigma\right\Vert _{F}^{2}}{n^2}+\frac{\sum_{r,s,t,v=1}^n\mathfrak{b}_{rs}(\phi)\mathfrak{b}_{tv}(\phi)\kappa_{rstv}}{n^2}+\frac{\left\Vert \beta _{0}^{\prime }\Psi^{\prime }\mathfrak{B}(\phi) E^{-1}\right\Vert ^{2}}{n^2}\right),
	\end{equation}
	with $\kappa_{rstv}$ denoting the fourth cumulant of $u_r, u_s, u_t, u_v$, $r,s,t,v=1,\ldots,n$. Under the linear process assumed in Assumption \ref{ass:errors_epsilon_3} it is known that
	\begin{equation}\label{cumulant_prop_lin_pro}
		\sum_{r,s,t,v=1}^n\kappa^2_{rstv}=O(n).
	\end{equation}
	Using (\ref{frob_norm_inequality}) and Assumptions \ref{ass:spec_norm_G} and \ref{ass:Sigma_spec_norm}, the first term in parentheses in (\ref{e_var_bound}) is
	\begin{eqnarray}
		O_p\left(\left\Vert \mathfrak{B}(\phi)\right\Vert _{F}^{2}\overline\varphi^2\left(\Sigma\right)/n^2\right)&=&O_p\left(\left\Vert T(\lambda)\right\Vert _{F}^{2}\left\Vert  E(\gamma)\right\Vert^{4}\left\Vert M(\gamma)\right\Vert ^{2}\left\Vert T(\lambda)\right\Vert^{2}/n^2\right)\nonumber\\
		&=&O_p\left(\left\Vert T(\lambda)\right\Vert^{4}/n\underline{\varphi}^2\left(\Sigma(\gamma)\right)\right)=O_p\left(\left\Vert T(\lambda)\right\Vert^{4}/n\right),\label{e_var_bound_1}
	\end{eqnarray}
	while the second is similarly
	\begin{equation}\label{e_var_bound_2}
		O_p\left\{\left(\left\Vert\mathfrak{B}(\phi) \right\Vert_F^2/n\right)\left(\sum_{r,s,t,v=1}^n\kappa^2_{rstv}/n^2\right)^{\frac{1}{2}}\right\}=o_p\left(\left\Vert T(\lambda)\right\Vert^{4}\right),
	\end{equation}
	using (\ref{cumulant_prop_lin_pro}). Finally,  the third term in parentheses in (\ref{e_var_bound}) is
	\begin{equation}\label{e_var_bound_3}
		O_p\left(\left\Vert\mathfrak{B}(\phi) \right\Vert^2/n\right)=O_p\left(\left\Vert T(\lambda)\right\Vert^{4}/n\right).
	\end{equation}
	By compactness of $\Lambda$ and Assumption \ref{ass:spec_norm_G}, (\ref{e_var_bound_1}), (\ref{e_var_bound_2}) and (\ref{e_var_bound_3}) are negligible, thus pointwise convergence is established.
	
	Uniform convergence will follow from an equicontinuity argument. First, for arbitrary $\varepsilon>0$ we can find points $\phi_*=\left(\lambda'_*,\gamma'_*\right)'$, possibly infinitely many, such that the neighborhoods $\left\Vert\phi-\phi^*\right\Vert<\varepsilon$ form an open cover of $\overline{\mathcal{N}}^\phi(\eta)$. Since $\Phi$ is compact any open cover has a finite subcover and thus we may in fact choose finitely many $\phi_*=\left(\lambda'_*,\gamma'_*\right)'$, whence it suffices to prove \begin{equation*}
		\underset{\left\Vert \phi -\phi _{_{\ast }}\right\Vert <\varepsilon }{%
			\sup }\left\vert \frac{c_3\left( \phi \right) }{c_1\left( \phi \right)
			+c_2\left( \phi \right) }-\frac{c_3\left( \phi _{\ast }\right) }{c_1\left(
			\phi _{\ast }\right) +c_2\left( \phi _{\ast }\right) }\right\vert
		\overset{p}\longrightarrow 0.
	\end{equation*}%
	Proceeding as in \cite{Gupta2018}, we denote the two components of $c_3\left( \phi \right) $ by $%
	c_{31}\left( \phi \right) ,$ $c_{32}\left( \phi \right) ,$ and are left with establishing the negligibility of
	\begin{eqnarray}
		&&\frac{\left\vert c_{31}\left( \phi \right) -c_{31}\left( \phi _{\ast
			}\right) \right\vert }{c_2\left( \phi \right) }+\frac{\left\vert c_{32}\left(
			\phi \right) -c_{32}\left( \phi _{\ast }\right) \right\vert }{c_1\left(
			\phi \right) }+\frac{\left\vert c_3\left( \phi _{\ast }\right) \right\vert
		}{c_1\left( \phi \right) c_1\left( \phi _{\ast }\right) }\left\vert c_1\left(
		\phi _{\ast }\right) -c_1\left( \phi \right) \right\vert\nonumber\\ &+&\frac{\left\vert
			c_3\left( \phi _{\ast }\right) \right\vert }{c_2\left( \phi \right) c_2\left(
			\phi _{\ast }\right) }\left\vert c_2\left( \phi _{\ast }\right) -c_2\left(
		\phi \right) \right\vert,\label{cons5}
	\end{eqnarray}
	uniformly on $\left\Vert \phi -\phi _{_{\ast }}\right\Vert <\varepsilon$. By the fact that (\ref{c_lower_bound}) and (\ref{d_lower_bound}) hold uniformly over $\Phi$, we first consider only the numerators in the first two terms in (\ref{cons5}). As in the proof of Theorem 1 of \cite{Delgado2015}, (\ref{frob_norm_inequality}) implies that $\mathcal{E}\left(\sup_{\left\Vert \phi -\phi _{_{\ast }}\right\Vert <\varepsilon}\left\vert c_{31}\left( \phi \right) -c_{31}\left( \phi _{\ast
	}\right) \right\vert\right)$ is bounded by \[
	n^{-1}\left(\mathcal{E}\left\Vert u\right\Vert^2+\sigma_0^2tr\Sigma\right)\sup_{\left\Vert \phi -\phi _{_{\ast }}\right\Vert <\varepsilon}\left\Vert\mathfrak{B}(\phi)-\mathfrak{B}(\phi_*) \right\Vert=O_p\left(\sup_{\left\Vert \phi -\phi _{_{\ast }}\right\Vert <\varepsilon}\left\Vert\mathfrak{B}(\phi)-\mathfrak{B}(\phi_*) \right\Vert\right),\]
	because $\mathcal{E}\left\Vert u\right\Vert^2=O(n)$ and $tr\Sigma=O(n)$.  $\mathfrak{B}(\phi)-\mathfrak{B}(\phi_*)$
	can be written as
	\begin{eqnarray}
		&&\left(T(\lambda)-T(\lambda_*)\right)' E(\gamma)'M(\gamma) E(\gamma)T(\lambda)+T(\lambda_{*})' \Sigma'(\gamma_*)M(\gamma_*) E(\gamma_*)\left(T(\lambda)-T(\lambda_*)\right)\nonumber\\
		&+&T'(\lambda_*)\left( E(\gamma)'M(\gamma) E(\gamma)- E(\gamma_*)'M(\gamma_*) E(\gamma_*)\right)T(\lambda),\label{cons6}
	\end{eqnarray}
	which, by the triangle inequality, has spectral norm bounded by
	\begin{eqnarray}
		&&\left\Vert T(\lambda)-T(\lambda_*)\right\Vert\left(\left\Vert E(\gamma)\right\Vert^2\left\Vert T(\lambda)\right\Vert+\left\Vert E(\gamma_*)\right\Vert^2\left\Vert T(\lambda_*)\right\Vert\right)\nonumber\\
		&+&\left\Vert T(\lambda_*)\right\Vert\left\Vert E(\gamma)'M(\gamma) E(\gamma)- E(\gamma_*)'M(\gamma_*) E(\gamma_*)\right\Vert\left\Vert T(\lambda)\right\Vert\nonumber\\
		&=&O_p\left(\left\Vert T(\lambda)-T(\lambda_*)\right\Vert+\left\Vert E(\gamma)'M(\gamma) E(\gamma)- E(\gamma_*)'M(\gamma_*) E(\gamma_*)\right\Vert\right).\nonumber\\\label{cons7}
	\end{eqnarray}
	By Assumption \ref{ass:spec_norm_G} the first term in parentheses on the right side of (\ref{cons7}) is bounded uniformly on $\left\Vert\phi-\phi_{*}\right\Vert<\varepsilon$ by
	\begin{equation}\label{cons8}
		\sum_{j=1}^{d_{\lambda}}\left\vert\lambda_j-\lambda_{*j}\right\vert\left\Vert G_j\right\Vert\leq\max_{j=1,\ldots,d_{\lambda}}\left\Vert G_j\right\Vert\left\Vert\lambda-\lambda_{*}\right\Vert=O_p(\varepsilon),
	\end{equation}
	while because $E(\gamma)'M(\gamma) E(\gamma)=n^{-1}\Sigma(\gamma)^{-1}\Psi\left(n^{-1} \Psi'\Sigma(\gamma)^{-1}\Psi\right)^{-1}\Psi'\Sigma(\gamma)^{-1}$ for any $\gamma\in\Gamma$, the second one can be decomposed into terms with bounds typified by
	\begin{eqnarray*}
		&&n^{-1}\left\Vert\Sigma(\gamma)^{-1}-\Sigma(\gamma_*)^{-1}\right\Vert\left\Vert \Psi\right\Vert^2\left\Vert\left(n^{-1} \Psi'\Sigma(\gamma)^{-1}\Psi\right)^{-1}\right\Vert\left\Vert\Sigma(\gamma)^{-1}\right\Vert^2\\
		&\leq&n^{-1}\left\Vert\Sigma(\gamma)-\Sigma(\gamma_*)\right\Vert\left\Vert \Psi\right\Vert^2\left\Vert\left(n^{-1} \Psi'\Sigma(\gamma)^{-1}\Psi\right)^{-1}\right\Vert\left\Vert\Sigma(\gamma)^{-1}\right\Vert^3\left\Vert\Sigma(\gamma_*)^{-1}\right\Vert\\
		&=&O_p\left(\left\Vert\Sigma(\gamma)-\Sigma(\gamma_*)\right\Vert\right)=O_p(\varepsilon),
	\end{eqnarray*}
	uniformly on $\left\Vert\phi-\phi_*\right\Vert<\varepsilon$, by Assumptions \ref{ass:Sigma_spec_norm} and \ref{ass:Psi_GLS_type}, Proposition \ref{prop:Sigma_diff_bound}  and the inequality $\left\Vert A\right\Vert\leq \left\Vert A\right\Vert_F$ for a generic matrix $A$, so that
	\begin{equation}\label{Bdiffbound}
		\sup_{\left\Vert\phi-\phi_*\right\Vert<\varepsilon}\left\Vert\mathfrak{B}(\phi)-\mathfrak{B}(\phi_*)\right\Vert=O_p(\varepsilon).
	\end{equation}
	Thus equicontinuity of the first term in (\ref{cons5}) follows because $\varepsilon$ is arbitrary. The equicontinuity of the second term in (\ref{cons5}) follows in much the same way. Indeed $\sup_{\left\Vert\phi-\phi_*\right\Vert<\varepsilon}c_{32}\left(
	\phi \right) -c_{32}\left( \phi _{\ast }\right)=2n^{-1}\beta_0'\Psi'\sup_{\left\Vert\phi-\phi_*\right\Vert<\varepsilon}\left(\mathfrak{B}(\phi)-\mathfrak{B}(\phi_*)\right)u=O_p\left(\sup_{\left\Vert\phi-\phi_*\right\Vert<\varepsilon}\left\Vert\mathfrak{B}(\phi)-\mathfrak{B}(\phi_*)\right\Vert\right)=O_p(\varepsilon)$,
	using earlier arguments and (\ref{Bdiffbound}). Because $c_{1}(\phi)$ is bounded and bounded away from zero in probability (see \ref{c_lower_bound}) for sufficiently large $n$ and all $\phi\in\overline{\mathcal{N}}^\phi(\eta)$, the third term in (\ref{cons5}) may be bounded by ${\left\vert c_{3}(\phi_*)\right\vert}/{c_{1}(\phi_*)}\left(1+{c_{1}(\phi_*)}/{c_{1}(\phi)}\right)\overset{p}\longrightarrow0,$
	convergence being uniform on $\left\Vert\phi-\phi_*\right\Vert<\varepsilon$ by pointwise convergence of $c_{3}(\phi)/\left(c_{1}(\phi)+c_{2}(\phi)\right)$, cf. \cite{Gupta2018}. The uniform convergence to zero of the fourth term in (\ref{cons5}) follows in identical fashion, because $c_{2}(\phi)$ is bounded and bounded away from zero (see (\ref{d_lower_bound})) in probability for sufficiently large $n$ and all $\phi\in\overline{\mathcal{N}}^\phi(\eta)$. This concludes the proof.
\end{proof}

\section{Lemmas}\label{app:lemmas}
\begin{lemma}\label{lemma:c1}
	Under the conditions of Theorem \ref{thm:consistency},
	$c_{1}(\gamma )=n^{-1}\beta ^{\prime }\Psi
	^{\prime }E^{\prime }(\gamma )M(\gamma )E(\gamma)\Psi \beta +o_{p}(1).
	$
\end{lemma}
\begin{proof}
	First,
	\[
	c_{1}(\gamma )=n^{-1}\beta ^{\prime }\Psi
	^{\prime }E^{\prime }(\gamma )M(\gamma )E(\gamma)\Psi \beta +c_{12}(\gamma
	)+c_{13}(\gamma ),
	\]
	with $c_{12}(\gamma )=2n^{-1}{e}^{\prime
	}E^{\prime }(\gamma )M(\gamma )E(\gamma)\Psi \beta $ and $c_{13}(\gamma
	)=n^{-1}{e}^{\prime }E^{\prime }(\gamma )M(\gamma )E(\gamma){e}$. It is readily seen that $c_{12}(\gamma )$ and $c_{13}(\gamma )$ are negligible.
\end{proof}
\begin{lemma}\label{lemma:gamma_order}
	Under the conditions of Theorem \ref{thm:stat_appr} or Theorem \ref{thm:stat_appr_SAR}, $
	\left\Vert\widehat\gamma-\gamma_0\right\Vert=O_p\left(\sqrt{d_\gamma/n}\right).
	$
\end{lemma}
\begin{proof}
	We show the details for the setting of Theorem \ref{thm:stat_appr} and omit the details for the setting of Theorem \ref{thm:stat_appr_SAR}. Write $l=\partial L(\beta_0,\gamma_0)/\partial\gamma$. By \cite{Robinson1988}, we have $\left\Vert\widehat\gamma-\gamma_0\right\Vert=O_p\left(\left\Vert l\right\Vert \right)$. Now $l=\left(l_1,\ldots,l_{d_\gamma}\right)'$, with $l_j=n^{-1}tr\left(\Sigma^{-1}\Sigma_j\right)-n^{-1}\sigma_0^{-2}u'\Sigma^{-1}\Sigma_j\Sigma^{-1}u$. Next, $\mathcal{E}\left\Vert l\right\Vert^2=\sum_{j=1}^{d_\gamma}\mathcal{E}\left(l_j^2\right)$ and
	\begin{equation}\label{gamma_order1}
		\mathcal{E}\left(l_j^2\right)=\frac{1}{n^2\sigma_0^4}var\left(u'\Sigma^{-1}\Sigma_j\Sigma^{-1}u\right)=\frac{1}{n^2\sigma_0^4}var\left(\varepsilon'B'\Sigma^{-1}\Sigma_j\Sigma^{-1}B\varepsilon\right)=\frac{1}{n^2\sigma_0^4}var\left(\varepsilon'D_j\varepsilon\right),
	\end{equation}
	say. But, writing $d_{j,st}$ for a typical element of the infinite dimensional matrix $D_j$, we have
	\begin{equation}\label{gamma_order2}
		var\left(\varepsilon'D_j\varepsilon\right)=\left(\mu_4-3\sigma_0^4\right)\sum_{s=1}^{\infty}d_{j,ss}^2+2\sigma_0^4tr\left(D_j^2\right)=\left(\mu_4-3\sigma_0^4\right)\sum_{s=1}^{\infty}d_{j,ss}^2+2\sigma_0^4\sum_{s,t=1}^{\infty}d_{j,st}^2.
	\end{equation}
	Next, by Assumptions \ref{ass:errors_epsilon_3}, \ref{ass:Sigma_spec_norm} and \ref{ass:Sigma_der_spec}
	\begin{equation}\label{djss_bd}
		\sum_{s=1}^{\infty}d_{j,ss}^2=\sum_{s=1}^{\infty}\left(b_s'\Sigma^{-1}\Sigma_j\Sigma^{-1}b_s\right)^2\leq \left(\sum_{s=1}^{\infty}\left\Vert b_s\right\Vert^2\right)\left\Vert \Sigma^{-1}\right\Vert^2\left\Vert\Sigma_j\right\Vert=O\left(\sum_{j=1}^n\sum_{s=1}^\infty b^{*2}_{js}\right)=O(n).
	\end{equation}
	Similarly,
	\begin{equation}\label{djst_sum_bd}
		\sum_{s,t=1}^{\infty}d_{j,st}^2=\sum_{s=1}^\infty b_s'\Sigma^{-1}\Sigma_j\Sigma^{-1}\left(\sum_{t=1}^\infty b_tb_t'\right)\Sigma^{-1}\Sigma_j\Sigma^{-1}b_s=\sum_{s=1}^\infty b_s'\Sigma^{-1}\Sigma_j\Sigma^{-1}\Sigma_j\Sigma^{-1}b_s=O(n).
	\end{equation}
	Using (\ref{djss_bd}) and (\ref{djst_sum_bd}) in (\ref{gamma_order2}) implies that $\mathcal{E}\left(l_j^2\right)=O\left(n^{-1}\right)$, by (\ref{gamma_order1}). Thus we have $\mathcal{E}\left\Vert l\right\Vert^2=O\left(d_\gamma /n\right)$, and thus $\left\Vert l\right\Vert=O_p\left(\sqrt{d_\gamma/n}\right)$, by Markov's inequality, proving the lemma.
\end{proof}
\begin{lemma}\label{lemma:mean_var}
	Under the conditions of Theorem \ref{thm:appr_clt}, $\mathcal{E}\left( {\sigma _{0}^{-2}}\varepsilon ^{\prime }\fancyv\varepsilon \right)=p$ and $Var\left( {\sigma _{0}^{-2}}\varepsilon ^{\prime }\fancyv\varepsilon \right)/2p\rightarrow 1$.
\end{lemma}
\begin{proof}
	As $
	\mathcal{E}\left( {\sigma _{0}^{-2}}\varepsilon ^{\prime }\fancyv\varepsilon \right)
	=tr\left( \mathcal{E}[B^{\prime }\Sigma ^{-1}\Psi (\Psi ^{\prime }\Sigma ^{-1}\Psi
	)^{-1}\Psi ^{\prime }\Sigma ^{-1}B]\right) =p,
	$
	and%
	\begin{equation}
		Var\left( \frac{1}{\sigma _{0}^{2}}\varepsilon ^{\prime }\fancyv\varepsilon
		\right) =\left( \frac{\mu _{4}}{\sigma _{0}^{4}}-3\right) \sum_{s=1}^{\infty
		}\mathcal{E}(v_{ss}^{2})+\mathcal{E}[tr(\fancyv\fancyv^{\prime })+tr(\fancyv^{2})]=\left( \frac{\mu _{4}}{\sigma
			_{0}^{4}}-3\right) \sum_{s=1}^{\infty }v_{ss}^{2}+2p,  \label{VarQ}
	\end{equation}
	it suffices to show that
	\begin{equation}  \label{square_v_neg}
		(2p)^{-1}\sum_{s=1}^{\infty }v_{ss}^{2}\overset{p}{\rightarrow} 0.
	\end{equation}
	Because $v_{ss}=b_s^{\prime }\mathscr{M}b_s$, we have $v_{ss}^2= \left(\sum_{i,j=1}^n b_{is}b_{js}m_{ij}\right)^2$. Thus, using Assumption \ref{ass:errors_epsilon_3} and (\ref{pbound}),
	we have
	\begin{eqnarray}
		\sum_{s=1}^{\infty }v_{ss}^{2} & \leq & \left(\sup_{i,j}\left\vert m_{ij}\right\vert\right)^2\sum_{s=1}^{\infty } \left(\sum_{i,j=1}^n \left\vert b^*_{is}\right\vert\left\vert b^*_{js}\right\vert\right)^2 =O_{p}\left(p^2n^{-2}\left(\sup_{s}\sum_{i=1}^n\left \vert b^*_{is}\right
		\vert \right)^3\sum_{i=1}^n\sum_{s=1}^{\infty }\left \vert b^*_{is}\right
		\vert \right)  \notag \\
		&=& O_{p}\left(p^2n^{-1}\right),  \label{v_squares_sumbound}
	\end{eqnarray}
	establishing (\ref{square_v_neg}) because $p^2/n\rightarrow 0$.
\end{proof}
\begin{lemma}\label{lemma:tau_order}
	Under the conditions of Theorem \ref{thm:stat_appr_np}, $
	\left\Vert\widehat\tau-\tau_0\right\Vert=O_p\left(\sqrt{d_\tau/n}\right).
	$
\end{lemma}
\begin{proof}
	The proof is similar to that of Lemma \ref{lemma:gamma_order} and is omitted.
\end{proof}
Denote $H(\gamma)=I_n+\sum_{j=m_1+1}^{m_1+m_2}\gamma_jW_j$ and $K(\gamma)=I_n-\sum_{j=1}^{m_1}\gamma_jW_j$. Let $G_j(\gamma)=W_jK^{-1}(\gamma)$, $j=1,\ldots,m_1$, $T_j=H^{-1}(\gamma)W_j$, $j=m_1+1,\ldots,m_1+m_2$  and, for a generic matrix $A$, denote $\overline A=A+A'$. Our final conditions may differ according to whether the $W_j$ are of general form or have
`single nonzero diagonal block structure', see e.g \cite{Gupta2013}. To define these, denote by $V$ an $n\times n$ block diagonal matrix with $i$-th block $V_{i}$, a $s_{i}\times s_{i}$ matrix, where $\sum _{i=1}^{m_1+m_2}s_{i}=n$, and for $i=1,...,m_1+m_2$ obtain $W_{j}$ from $V$ by replacing each $V_{j}$, $j\neq
i$, by a matrix of zeros. Thus $V=\sum _{i=1}^{m_1+m_2}W_{j}$.
\begin{lemma}\label{lemma:sem_SARMA_frech_der}
	For the spatial error model with SARMA$(p,q)$ errors, if \begin{equation}\label{sem_SARMA_bounds}
		\sup_{\gamma\in\Gamma^o}\left(\left\Vert K^{-1}(\gamma)\right\Vert+\left\Vert K'^{-1}(\gamma)\right\Vert+\left\Vert H^{-1}(\gamma)\right\Vert+\left\Vert H'^{-1}(\gamma)\right\Vert\right)+\max_{j=1,\ldots,m_1+m_2}\left\Vert W_j\right\Vert < C,
	\end{equation}
	then
	\[
	\left(D\Sigma(\gamma)\right)\left(\gamma^\dag\right)= A^{-1}(\gamma)\left(\sum_{j=1}^{m_1}\gamma^\dag_j\overline{H^{-1}(\gamma)G_j(\gamma)}
	+\sum_{j=m_1+1}^{m_1+m_2}\gamma^\dag_j\overline{T_j(\gamma)}\right)A'^{-1}(\gamma).
	\]
\end{lemma}
\begin{proof}
	We first show that $D\Sigma\in\mathscr{L}\left(\Gamma^o,\mathcal{M}^{n\times n}\right)$. Clearly, $D\Sigma$ is a linear map and (\ref{sem_SARMA_bounds})
	\[
	\left\Vert\left(D\Sigma(\gamma)\right)\left(\gamma^\dag\right)\right\Vert\leq C \left\Vert\gamma^\dag\right\Vert_1,
	\]
	in the general case and
	\[
	\left\Vert\left(D\Sigma(\gamma)\right)\left(\gamma^\dag\right)\right\Vert\leq C \max_{j=1,\ldots,m_1+m_2}\left\vert\gamma_j^\dag\right\vert,
	\]
	in the `single nonzero diagonal block' case.
	Thus $D\Sigma$ is a bounded linear operator between two normed linear spaces, i.e. it is a continuous linear operator.
	
	With $A(\gamma)=H^{-1}(\gamma)K(\gamma)$, we now show that
	\begin{equation}\label{sem_SARMA_tgt}
		\frac{\left\Vert A^{-1}\left(\gamma+\gamma^\dag\right)A'^{-1}\left(\gamma+\gamma^\dag\right)-A^{-1}\left(\gamma\right)A'^{-1}\left(\gamma\right)-\left(D\Sigma(\gamma)\right)\left(\gamma^\dag\right)\right\Vert}{\left\Vert\gamma^\dag\right\Vert_g}\rightarrow 0, \text{ as }\left\Vert\gamma^\dag\right\Vert_g\rightarrow 0,
	\end{equation}
	where $\left\Vert\cdot\right\Vert_g$ is either the 1-norm or the max norm on $\Gamma$. First, note that \begin{eqnarray}
		&&A^{-1}\left(\gamma+\gamma^\dag\right)A'^{-1}\left(\gamma+\gamma^\dag\right)-A^{-1}(\gamma)A'^{-1}(\gamma)\nonumber\\
		&=&A^{-1}\left(\gamma+\gamma^\dag\right)\left(A^{-1}\left(\gamma+\gamma^\dag\right)-A^{-1}(\gamma)\right)'+\left(A^{-1}\left(\gamma+\gamma^\dag\right)-A^{-1}(\gamma)\right)A^{-1}(\gamma)\nonumber\\
		&=&-A^{-1}\left(\gamma+\gamma^\dag\right)A'^{-1}\left(\gamma+\gamma^\dag\right)\left(A\left(\gamma+\gamma^\dag\right)-A(\gamma)\right)'A'^{-1}(\gamma)\nonumber\\
		&-&A^{-1}\left(\gamma+\gamma^\dag\right)\left(A\left(\gamma+\gamma^\dag\right)-A(\gamma)\right)A^{-1}(\gamma)A'^{-1}(\gamma)\label{sem_SARMA_Sdiff}.
	\end{eqnarray}
	Next,
	\begin{eqnarray}
		A\left(\gamma+\gamma^\dag\right)-A(\gamma)&=&H^{-1}\left(\gamma+\gamma^\dag\right)K\left(\gamma+\gamma^\dag\right)-H^{-1}\left(\gamma\right)K\left(\gamma\right)\nonumber\\
		&=&H^{-1}\left(\gamma+\gamma^\dag\right)\left(K\left(\gamma+\gamma^\dag\right)-K(\gamma)\right)\nonumber\\
		&+&H^{-1}\left(\gamma+\gamma^\dag\right)\left(H\left(\gamma\right)-H\left(\gamma+\gamma^\dag\right)\right)H^{-1}\left(\gamma\right)K\left(\gamma\right)\nonumber\\
		&=&-H^{-1}\left(\gamma+\gamma^\dag\right)\left(\sum_{j=1}^{m_1}\gamma_j^\dag W_j+\sum_{j=m_1+1}^{m_1+m_2}\gamma_j^\dag W_jH^{-1}(\gamma)K(\gamma)\right).\nonumber\\\label{sem_SARMA1}
	\end{eqnarray}
	Substituting (\ref{sem_SARMA1}) in (\ref{sem_SARMA_Sdiff}) implies that \begin{equation}\label{sem_SARMA2}
		A^{-1}\left(\gamma+\gamma^\dag\right)A'^{-1}\left(\gamma+\gamma^\dag\right)-A^{-1}(\gamma)A'^{-1}(\gamma)=\Delta_1\left(\gamma,\gamma^\dag\right)+\Delta_2\left(\gamma,\gamma^\dag\right)=\Delta\left(\gamma,\gamma^\dag\right),
	\end{equation}
	say, where
	\begin{eqnarray*}
		\Delta_1\left(\gamma,\gamma^\dag\right)&=&A^{-1}\left(\gamma+\gamma^\dag\right)A'^{-1}\left(\gamma+\gamma^\dag\right)\left(\sum_{j=1}^{m_1}\gamma_j^\dag W'_j+K'(\gamma)H'^{-1}(\gamma)\sum_{j=m_1+1}^{m_1+m_2}\gamma_j^\dag W'_j\right)\\
		&\times&H'^{-1}\left(\gamma+\gamma^\dag\right)A'^{-1}(\gamma),\\
		\Delta_2\left(\gamma,\gamma^\dag\right)&=&A^{-1}\left(\gamma+\gamma^\dag\right)H^{-1}\left(\gamma+\gamma^\dag\right)\left(\sum_{j=1}^{m_1}\gamma_j^\dag W_j+\sum_{j=m_1+1}^{m_1+m_2}\gamma_j^\dag W_jH^{-1}(\gamma)K(\gamma)\right)\\
		&\times&A^{-1}(\gamma)A'^{-1}(\gamma).
	\end{eqnarray*}
	From the definitions above and recalling that $A(\gamma)=H^{-1}(\gamma)K(\gamma)$, we can write
	\begin{equation}\label{sem_SARMA_split}
		\Delta\left(\gamma,\gamma^\dag\right)=A^{-1}\left(\gamma+\gamma^\dag\right)\Upsilon\left(\gamma,\gamma^\dag\right) A'^{-1}\left(\gamma\right),
	\end{equation}
	with
	\begin{eqnarray*}
		\Upsilon\left(\gamma,\gamma^\dag\right)&=&\sum_{j=1}^{m_1}\gamma_j^\dag G'_j\left(\gamma+\gamma^\dag\right)H'^{-1}\left(\gamma+\gamma^\dag\right)+A'^{-1}\left(\gamma+\gamma^\dag\right)A'(\gamma)\sum_{j=m_1+1}^{m_1+m_2}\gamma_j^\dag T'_j\left(\gamma+\gamma^\dag\right)\\
		&+&\sum_{j=1}^{m_1}\gamma_j^\dag H^{-1}\left(\gamma+\gamma^\dag\right)G_j\left(\gamma \right)+\sum_{j=m_1+1}^{m_1+m_2}\gamma_j^\dag T_j\left(\gamma+\gamma^\dag\right).
	\end{eqnarray*}
	Then (\ref{sem_SARMA2}) implies that
	\begin{eqnarray}
		&&A^{-1}\left(\gamma+\gamma^\dag\right)A'^{-1}\left(\gamma+\gamma^\dag\right)-A^{-1}(\gamma)A'^{-1}(\gamma)-\left(D\Sigma(\gamma)\right)\left(\gamma^\dag\right)\nonumber\\
		&=&A^{-1}\left(\gamma+\gamma^\dag\right)A'^{-1}\left(\gamma+\gamma^\dag\right)-A^{-1}(\gamma)A'^{-1}(\gamma)- \Delta\left(\gamma,\gamma^\dag\right)-\left(D\Sigma(\gamma)\right)\left(\gamma^\dag\right)+\Delta\left(\gamma,\gamma^\dag\right)\nonumber\\
		&=&\Delta\left(\gamma,\gamma^\dag\right)-\left(D\Sigma(\gamma)\right)\left(\gamma^\dag\right),
	\end{eqnarray}
	so to prove (\ref{sem_SARMA_tgt}) it is sufficient to show that
	\begin{equation}\label{sem_SARMA3}
		\frac{\left\Vert\Delta\left(\gamma,\gamma^\dag\right)-\left(D\Sigma(\gamma)\right)\left(\gamma^\dag\right)\right\Vert}{\left\Vert \gamma^\dag\right\Vert_g}\rightarrow 0 \text{ as } \left\Vert \gamma^\dag\right\Vert_g\rightarrow 0.
	\end{equation}
	The numerator in (\ref{sem_SARMA3}) can be written as $\sum_{i=1}^{7}\Pi_i\left(\gamma,\gamma^\dag\right) A'^{-1}(\gamma)$ by adding, subtracting and grouping terms, where (omitting the argument $\left(\gamma,\gamma^\dag\right)$)
	\begin{eqnarray*}
		\Pi_1&=&A^{-1}\left(\gamma+\gamma^\dag\right)\sum_{j=1}^{m_1}\gamma_j^\dag G'_j\left(\gamma+\gamma^\dag\right) H'^{-1}(\gamma)\left(H(\gamma)-H\left(\gamma+\gamma^\dag\right)\right)'H'^{-1}\left(\gamma+\gamma^\dag\right),\\
		\Pi_2&=&A^{-1}\left(\gamma+\gamma^\dag\right)\sum_{j=1}^{m_1}\gamma_j^\dag H^{-1}\left(\gamma+\gamma^\dag\right)\left(H(\gamma)-H\left(\gamma+\gamma^\dag\right)\right)H^{-1}(\gamma)G_j\left(\gamma\right),\\
		\Pi_3&=&A^{-1}\left(\gamma+\gamma^\dag\right)\sum_{j=m_1+1}^{m_1+m_2}\gamma_j^\dag \left(A^{-1}\left(\gamma+\gamma^\dag\right)-A^{-1}\left(\gamma\right)\right)T'_j\left(\gamma+\gamma^\dag\right),\\
		\Pi_4&=&\left(A^{-1}\left(\gamma+\gamma^\dag\right)-A^{-1}\left(\gamma\right)\right)\sum_{j=m_1+1}^{m_1+m_2}\gamma_j^\dag\overline{T_j\left(\gamma+\gamma^\dag\right)},\\
		\Pi_5&=&A^{-1}(\gamma)\sum_{j=m_1+1}^{m_1+m_2}\gamma_j^\dag\overline{H^{-1}\left(\gamma+\gamma^\dag\right)\left(H(\gamma)-H\left(\gamma+\gamma^\dag\right)\right)H^{-1}(\gamma)W_j},\\
		\Pi_6&=&\Delta\left(\gamma,\gamma^\dag\right)\sum_{j=1}^{m_1}\gamma_j^\dag W_j'H'^{-1}(\gamma),\\
		\Pi_7&=&\left(A^{-1}\left(\gamma+\gamma^\dag\right)-A^{-1}\left(\gamma\right)\right)\sum_{j=1}^{m_1}\gamma_j^\dag H^{-1}(\gamma)G_j(\gamma).
	\end{eqnarray*}
	By (\ref{sem_SARMA_bounds}), (\ref{sem_SARMA_split}) and replication of earlier techniques, we have
	\begin{equation}\label{sem_SARMA4}
		\max_{i=1,\ldots,7}\sup_{\gamma\in\Gamma^o}\left\Vert \Pi_i \left(\gamma,\gamma^\dag\right)A^{-1}(\gamma)\right\Vert\leq C\left\Vert \gamma^\dag\right\Vert^2_g,
	\end{equation}
	where the norm used on the RHS of (\ref{sem_SARMA4}) depends on whether we are considering the general case or the `single nonzero diagonal block' case. Thus
	\[
	\frac{\left\Vert\Delta\left(\gamma,\gamma^\dag\right)-\left(D\Sigma(\gamma)\right)\left(\gamma^\dag\right)\right\Vert}{\left\Vert \gamma^\dag\right\Vert_g}\leq C \left\Vert \gamma^\dag\right\Vert_g\rightarrow 0 \text{ as } \left\Vert \gamma^\dag\right\Vert_g\rightarrow 0,
	\]
	proving (\ref{sem_SARMA3}) and thus (\ref{sem_SARMA_tgt}).
\end{proof}
\begin{corollary}\label{cor:sem_SAR_frech_der}
	For the spatial error model with SAR$(m_1)$ errors,
	\[
	\left(D\Sigma(\gamma)\right)\left(\gamma^\dag\right)= K^{-1}(\gamma)\sum_{j=1}^{m_1}\gamma^\dag_j\overline{G_j(\gamma)}
	K'^{-1}(\gamma).
	\]
\end{corollary}
\begin{proof}
	Taking $m_2=0$ in Lemma \ref{lemma:sem_SARMA_frech_der}, the elements involving sums from $m_1+1$ to $m_1+m_2$ do not arise and $H(\gamma)=I_n$, proving the claim.
\end{proof}
\begin{corollary}\label{cor:sem_SMA_frech_der}
	For the spatial error model with SMA$(m_2)$ errors,
	\[
	\left(D\Sigma(\gamma)\right)\left(\gamma^\dag\right)= H(\gamma)\sum_{j=1}^{m_2}\gamma^\dag_j\overline{T_j(\gamma)}H'(\gamma).
	\]
\end{corollary}
\begin{proof}
	Taking $m_1=0$ in Lemma \ref{lemma:sem_SARMA_frech_der}, the elements involving sums from $1$ to $m_1$ do not arise and $K(\gamma)=I_n$, proving the claim.
\end{proof}
\begin{lemma}\label{lemma:sem_MESS_frech_der}
	For the spatial error model with MESS$(m_1)$ errors, if
	\begin{equation}\label{MESS_condition_der}
		\max_{j=1,\ldots,m_1}\left(\left\Vert W_j\right\Vert+\left\Vert W'_j\right\Vert\right)<1, \end{equation}
	then
	\[
	\left(D\Sigma(\gamma)\right)\left(\gamma^\dag\right)= \exp\left(\sum_{j=1}^{m_1}\gamma_j\left(W_j+W_j'\right)\right)\sum_{j=1}^{m_1}\gamma_j^\dag\left(W_j+W_j'\right).
	\]
\end{lemma}
\begin{proof}
	Clearly $D\Sigma\in\mathscr{L}\left(\Gamma^o,\mathcal{M}^{n\times n}\right)$. Next,
	\begin{eqnarray}
		&&\left\Vert A^{-1}\left(\gamma+\gamma^\dag\right)A'^{-1}\left(\gamma+\gamma^\dag\right)-A^{-1}(\gamma)A'^{-1}(\gamma)-\left(D\Sigma(\gamma)\right)\left(\gamma^\dag\right)\right\Vert\nonumber\\
		&=&\left\Vert\exp\left(\sum_{j=1}^{m_1}\left(\gamma_j+\gamma^\dag_j\right)\left(W_j+W_j'\right)\right)-\exp\left(\sum_{j=1}^{m_1}\gamma_j\left(W_j+W_j'\right)\right)-\left(D\Sigma(\gamma)\right)\left(\gamma^\dag\right)\right\Vert\nonumber\\
		&=&\left\Vert\exp\left(\sum_{j=1}^{m_1}\gamma_j\left(W_j+W_j'\right)\right)\left(\exp\left(\sum_{j=1}^{m_1}\gamma^\dag_j\left(W_j+W_j'\right)\right)-I_n-\sum_{j=1}^{m_1}\gamma^\dag_j\left(W_j+W_j'\right)\right)\right\Vert\nonumber\\
		&\leq&\left\Vert\exp\left(\sum_{j=1}^{m_1}\gamma_j\left(W_j+W_j'\right)\right)\right\Vert\left\Vert\exp\left(\sum_{j=1}^{m_1}\gamma^\dag_j\left(W_j+W_j'\right)\right)-I_n-\sum_{j=1}^{m_1}\gamma^\dag_j\left(W_j+W_j'\right)\right\Vert\nonumber\\
		&\leq& C\left\Vert I_n+\sum_{j=1}^p\gamma^\dag_j\left(W_j+W_j'\right)+\sum_{k=2}^{\infty}\left\{\sum_{j=1}^{m_1}\gamma^\dag_j\left(W_j+W_j'\right)\right\}^k-I_n-\sum_{j=1}^{m_1}\gamma^\dag_j\left(W_j+W_j'\right)\right\Vert\nonumber\\
		&\leq& C\left\Vert \sum_{k=2}^{\infty}\left\{\sum_{j=1}^{m_1}\gamma^\dag_j\left(W_j+W_j'\right)\right\}^k\right\Vert\leq C \sum_{k=2}^{\infty}\sum_{j=1}^{m_1}\left\vert \gamma^\dag_j\right\vert\left\Vert\left(W_j+W_j'\right)\right\Vert^k\nonumber\\
		&\leq& C\sum_{k=2}^{\infty}\left\Vert\gamma^\dag \right\Vert^k_g,\label{sem_MESS1}
	\end{eqnarray}
	by (\ref{MESS_condition_der}), without loss of generality, and again the norm used in (\ref{sem_MESS1}) depending on whether we are in the general or the `single nonzero diagonal block' case. Thus
	\[
	\frac{\left\Vert A^{-1}\left(\gamma+\gamma^\dag\right)A'^{-1}\left(\gamma+\gamma^\dag\right)-A^{-1}(\gamma)A'^{-1}(\gamma)-\left(D\Sigma(\gamma)\right)\left(\gamma^\dag\right)\right\Vert}{\left\Vert\gamma^\dag \right\Vert_g}\leq C \sum_{k=2}^{\infty}\left\Vert\gamma^\dag \right\Vert^{k-1}_g\rightarrow 0,
	\]
	as $\left\Vert\gamma^\dag \right\Vert_g\rightarrow 0$, proving the claim.
\end{proof}
\begin{theorem}\label{thm:tilde_hat_equiv}
	Under the conditions of Theorem \ref{thm:stat_properties} or \ref{thm:stat_properties_SAR}, $\mathscr{T}_n-\mathscr{T}_n^a=o_p(1)$ as $n\rightarrow\infty$.
\end{theorem}
\begin{proof}
	It suffices to show that $n\widetilde{m}_{n}=n\widehat{m}_{n}+o_{p}(\sqrt{%
		p})$. As $\widehat{\eta }=y-\widehat{{\theta }},$ $\widehat{u}=y-\widehat{f}%
	$, and $\widehat{v}=\widehat{\theta }-\widehat{f}$, we have $%
	\widehat{u}=\widehat{\eta }+\widehat{v}$ and%
	\begin{eqnarray}
		n\widetilde{m}_{n} &=&\widehat{\sigma }^{-2}\left( \widehat{u}^{\prime
		}\Sigma \left( \widehat{\gamma }\right) ^{-1}\widehat{u}-\widehat{\eta }%
		^{\prime }\Sigma \left( \widehat{\gamma }\right) ^{-1}\widehat{\eta }\right)
		=\widehat{\sigma }^{-2}\left( 2\widehat{u}^{\prime }\Sigma \left( \widehat{%
			\gamma }\right) ^{-1}\widehat{v}-\widehat{v}^{\prime }\Sigma \left( \widehat{%
			\gamma }\right) ^{-1}\widehat{v}\right)   \notag \\
		&=&2n\widehat{m}_{n}-\widehat{\sigma }^{-2}\left[ \Psi \left( \Psi ^{\prime
		}\Sigma \left( \widehat{\gamma }\right) ^{-1}\Psi \right) ^{-1}\Psi ^{\prime
		}\Sigma \left( \widehat{\gamma }\right) ^{-1}({u+e}{)}-{e}+{\theta }_{0}-%
		\widehat{{f}}\right] ^{\prime }  \notag \\
		&&\Sigma \left( \widehat{\gamma }\right) ^{-1}\left[ \Psi \left( \Psi
		^{\prime }\Sigma \left( \widehat{\gamma }\right) ^{-1}\Psi \right) ^{-1}\Psi
		^{\prime }\Sigma \left( \widehat{\gamma }\right) ^{-1}({u+e}{)}-{e}+{\theta }%
		_{0}-\widehat{{f}}\right]   \notag \\
		&=&2n\widehat{m}_{n}-\widehat{\sigma }^{-2}u^{\prime }\Sigma \left( \widehat{%
			\gamma }\right) ^{-1}\Psi \left( \Psi ^{\prime }\Sigma \left( \widehat{%
			\gamma }\right) ^{-1}\Psi \right) ^{-1}\Psi ^{\prime }\Sigma \left( \widehat{%
			\gamma }\right) ^{-1}u\mathbf{-}\widehat{\sigma }^{-2}\left( {\theta }_{0}-%
		\widehat{{f}}\right) ^{\prime }\Sigma \left( \widehat{\gamma }\right)
		^{-1}\left( {\theta }_{0}-\widehat{{f}}\right)   \notag \\
		&&\mathbf{+}\widehat{\sigma }^{-2}\left( 2({\theta }_{0}-\widehat{{f}}%
		)-e\right) ^{\prime }\Sigma \left( \widehat{\gamma }\right) ^{-1}\left(
		I-\Psi \lbrack \Psi ^{\prime }\Sigma \left( \widehat{\gamma }\right)
		^{-1}\Psi ]^{-1}\Psi ^{\prime }\Sigma \left( \widehat{\gamma }\right)
		^{-1}\right) e  \notag \\
		&&-2\widehat{\sigma }^{-2}\left( {\theta }_{0}-\widehat{{f}}\right)
		^{\prime }\Sigma \left( \widehat{\gamma }\right) ^{-1}\Psi \left( \Psi
		^{\prime }\Sigma \left( \widehat{\gamma }\right) ^{-1}\Psi \right) ^{-1}\Psi
		^{\prime }\Sigma \left( \widehat{\gamma }\right) ^{-1}u  \notag \\
		&=&2n\widehat{m}_{n}-\left( n\widehat{m}_{n}-\widehat{\sigma }^{-2}\left(A_{1}+A_{2}+A_{3}+A_{4}\right)%
		\right) -\widehat{\sigma }^{-2}A_{4}  \notag \\
		&&\mathbf{+}\widehat{\sigma }^{-2}\left( 2({\theta }_{0}-\widehat{{f}}%
		)-e\right) ^{\prime }\Sigma \left( \widehat{\gamma }\right) ^{-1}\left(
		I-\Psi \lbrack \Psi ^{\prime }\Sigma \left( \widehat{\gamma }\right)
		^{-1}\Psi ]^{-1}\Psi ^{\prime }\Sigma \left( \widehat{\gamma }\right)
		^{-1}\right) e-2\widehat{\sigma }^{-2}A_{3}  \notag \\
		&=&n\widehat{m}_{n}+\widehat{\sigma }^{-2}\left( A_{1}+A_{2}-A_{3}\right)
		\notag \\
		&&+\widehat{\sigma }^{-2}\left( 2({\theta }_{0}-\widehat{{f}})-e\right)
		^{\prime }\Sigma \left( \widehat{\gamma }\right) ^{-1}\left( I-\Psi \left(
		\Psi ^{\prime }\Sigma \left( \widehat{\gamma }\right) ^{-1}\Psi \right)
		^{-1}\Psi ^{\prime }\Sigma \left( \widehat{\gamma }\right) ^{-1}\right) e.
		\label{eqn: alternativetest}
	\end{eqnarray}%
	In the proof of Theorem \ref{thm:stat_appr}, we have shown that%
	\begin{equation*}
		\left \vert \left( {\theta }_{0}-\widehat{{f}}\right) ^{\prime }\Sigma
		\left( \widehat{\gamma }\right) ^{-1}\left( I-\Psi \lbrack \Psi ^{\prime
		}\Sigma \left( \widehat{\gamma }\right) ^{-1}\Psi ]^{-1}\Psi ^{\prime
		}\Sigma \left( \widehat{\gamma }\right) ^{-1}\right) e\right \vert =o_{p}(%
		\sqrt{p})
	\end{equation*}%
	in the process of proving $|A_{2}|=o_{p}(\sqrt{p})$. Along with%
	\begin{eqnarray*}
		&&\left \vert e^{\prime }\Sigma \left( \widehat{\gamma }\right) ^{-1}\left(
		I-\Psi \left( \Psi ^{\prime }\Sigma \left( \widehat{\gamma }\right)
		^{-1}\Psi \right) ^{-1}\Psi ^{\prime }\Sigma \left( \widehat{\gamma }\right)
		^{-1}\right) e\right \vert  \\
		&\leq &\left \vert e^{\prime }\Sigma \left( \widehat{\gamma }\right)
		^{-1}e\right \vert +\left \vert e^{\prime }\Sigma \left( \widehat{\gamma }%
		\right) ^{-1}\Psi \left( \Psi ^{\prime }\Sigma \left( \widehat{\gamma }%
		\right) ^{-1}\Psi \right) ^{-1}\Psi ^{\prime }\Sigma \left( \widehat{\gamma }%
		\right) ^{-1}e\right \vert  \\
		&\leq &\left \Vert e\right \Vert ^{2}\sup_{\gamma \in \Gamma }\left \Vert
		\Sigma \left( \gamma \right) ^{-1}\right \Vert +\left \Vert e\right \Vert
		^{2}\sup_{\gamma \in \Gamma }\left \Vert \Sigma \left( \gamma \right)
		^{-1}\right \Vert ^{2}\left \Vert \frac{1}{n}\Psi \left( \frac{1}{n}\Psi
		^{\prime }\Sigma \left( \gamma \right) ^{-1}\Psi \right) ^{-1}\Psi ^{\prime
		}\right \Vert  \\
		&=&O_{p}\left( \left \Vert e\right \Vert ^{2}\right) =O_{p}\left( p^{-2\mu
		}n\right) =o_{p}(\sqrt{p}),
	\end{eqnarray*}%
	we complete the proof that $n\widetilde{m}_{n}=n\widehat{m}_{n}+o_{p}(\sqrt{p}).$
	In the SAR setting of Section \ref{sec:SAR_ext},%
	\begin{eqnarray*}
		n\widetilde{m}_{n} &=&\widehat{\sigma }^{-2}\left( \widehat{u}^{\prime
		}\Sigma \left( \widehat{\gamma }\right) ^{-1}\widehat{u}-\widehat{\eta }%
		^{\prime }\Sigma \left( \widehat{\gamma }\right) ^{-1}\widehat{\eta }\right)
		=\widehat{\sigma }^{-2}\left( 2\widehat{u}^{\prime }\Sigma \left( \widehat{%
			\gamma }\right) ^{-1}\widehat{v}-\widehat{v}^{\prime }\Sigma \left( \widehat{%
			\gamma }\right) ^{-1}\widehat{v}\right)  \\
		&=&2n\widehat{m}_{n}-\widehat{\sigma }^{-2}\left[ \Psi \left( \Psi ^{\prime
		}\Sigma \left( \widehat{\gamma }\right) ^{-1}\Psi \right) ^{-1}\Psi ^{\prime
		}\Sigma \left( \widehat{\gamma }\right) ^{-1}\left(
		u+e+\sum_{j=1}^{d_{\lambda }}(\lambda _{j_{0}}-\widehat{\lambda }%
		_{j})W_{j}y\right) -e+\theta _{0}-\widehat{f}\right] ^{\prime } \\
		&&\Sigma \left( \widehat{\gamma }\right) ^{-1}\left[ \Psi \left( \Psi
		^{\prime }\Sigma \left( \widehat{\gamma }\right) ^{-1}\Psi \right) ^{-1}\Psi
		^{\prime }\Sigma \left( \widehat{\gamma }\right) ^{-1}\left(
		u+e+\sum_{j=1}^{d_{\lambda }}(\lambda _{j_{0}}-\widehat{\lambda }%
		_{j})W_{j}y\right) -e+\theta _{0}-\widehat{f}\right] .
	\end{eqnarray*}%
	Compared to the expression in (\ref{eqn: alternativetest}), we have
	the additional terms
	\begin{equation*}
		-\widehat{\sigma }^{-2}\sum_{j=1}^{d_{\lambda }}(\lambda _{j_{0}}-\widehat{%
			\lambda }_{j})W_{j}y^{\prime }\Sigma \left( \widehat{\gamma }\right)
		^{-1}\Psi \left( \Psi ^{\prime }\Sigma \left( \widehat{\gamma }\right)
		^{-1}\Psi \right) ^{-1}\Psi ^{\prime }\Sigma \left( \widehat{\gamma }\right)
		^{-1}\sum_{j=1}^{d_{\lambda }}(\lambda _{j_{0}}-\widehat{\lambda }_{j})W_{j}y
	\end{equation*}%
	and%
	\begin{equation*}
		-2\widehat{\sigma }^{-2}\sum_{j=1}^{d_{\lambda }}(\lambda _{j_{0}}-\widehat{%
			\lambda }_{j})W_{j}y^{\prime }\Sigma \left( \widehat{\gamma }\right)
		^{-1}\Psi \left( \Psi ^{\prime }\Sigma \left( \widehat{\gamma }\right)
		^{-1}\Psi \right) ^{-1}\Psi ^{\prime }\Sigma \left( \widehat{\gamma }\right)
		^{-1}\left( u+\theta _{0}-\widehat{f}\right) .
	\end{equation*}%
	Both terms are $o_{p}(\sqrt{p})$ from the orders of $A_{5}$ and $A_{6}$ in
	the proof of Theorem \ref{thm:stat_appr_SAR}. Hence, in the SAR setting, $n\widetilde{m}_{n}=n%
	\widehat{m}_{n}+o_{p}(\sqrt{p})$ also holds.
	
	\sloppy We now present similar calculations that justify the validity of our bootstrap test for the SARARMA($m_{1}$,$m_{2},m_{3}$) model. The bootstrapped test statistic is constructed with
	\begin{equation*}
		n\widehat{m}_{n}^{\ast }=\widehat{{v}}^{\ast \prime }\Sigma \left(
		\widehat{\gamma }^{\ast }\right) ^{-1}\widehat{{u}}^{\ast }=(\widehat{%
			{\theta }}_{n}^{\ast }-{f}(x,\widehat{\alpha }_{n}^{\ast
		}))^{\prime }\Sigma \left( \widehat{\gamma }^{\ast }\right) ^{-1}\left(
		(I_{n}-\sum_{k=1}^{m_{1}}\widehat{\lambda }_{k}^{\ast }W_{1k})y^{\ast }-%
		{f}(x,\widehat{\alpha }_{n}^{\ast })\right).
	\end{equation*}
	
	Let $J_{n}=(I_{n}-\frac{1}{n}l_{n}l_{n}^{\prime })$. As $y=S(\lambda
	)^{-1}(\theta (x)+R(\gamma )\xi )$, we have%
	\begin{eqnarray*}
		\widetilde{\mathbf{\xi }} &=&J_{n}\widehat{\mathbf{\xi }}\\
		&=&J_{n}\left( \left(
		\sum_{l=1}^{m_{3}}\gamma _{3l}W_{3l}+I_{n}\right) ^{-1}+\left(
		\sum_{l=1}^{m_{3}}\gamma _{3l}W_{3l}+I_{n}\right)
		^{-1}\sum_{l=1}^{m_{3}}(\gamma _{3l}-\widehat{\gamma }_{3l})W_{3l}\left(
		\sum_{l=1}^{m_{3}}\widehat{\gamma }_{3l}W_{3l}+I_{n}\right) ^{-1}\right)  \\
		&&\times \left( I_{n}-\sum_{l=1}^{m_{2}}\gamma
		_{2l}W_{2l}+\sum_{l=1}^{m_{2}}(\gamma _{2l}-\widehat{\gamma }%
		_{2l})W_{2l}\right) \left( S(\lambda )y-\theta
		(x)+\sum_{k=1}^{m_{1}}(\lambda _{k}-\widehat{\lambda }_{k})W_{1k}y+\theta
		(x)-\psi ^{\prime }\widehat{\beta }\right)  \\
		&=&\xi -\frac{1}{n}l_{n}l_{n}^{\prime }\xi +J_{n}\left(
		\sum_{l=1}^{m_{3}}\gamma _{3l}W_{3l}+I_{n}\right) ^{-1}\left(
		I_{n}-\sum_{l=1}^{m_{2}}\gamma _{2l}W_{2l}\right) \left(
		\sum_{k=1}^{m_{1}}(\lambda _{k}-\widehat{\lambda }_{k})W_{1k}y+\theta
		(x)-\psi ^{\prime }\widehat{\beta }\right)  \\
		&&+J_{n}\left( \sum_{l=1}^{m_{3}}\gamma _{3l}W_{3l}+I_{n}\right)
		^{-1}\sum_{l=1}^{m_{2}}(\gamma _{2l}-\widehat{\gamma }_{2l})W_{2l}\left(
		S(\lambda )y-\theta (x)+\sum_{k=1}^{m_{1}}(\lambda _{k}-\widehat{\lambda }%
		_{k})W_{1k}y+\theta (x)-\psi ^{\prime }\widehat{\beta }\right)  \\
		&&+J_{n}\left( \sum_{l=1}^{m_{3}}\gamma _{3l}W_{3l}+I_{n}\right)
		^{-1}\sum_{l=1}^{m_{3}}(\gamma _{3l}-\widehat{\gamma }_{3l})W_{3l}\left(
		\sum_{l=1}^{m_{3}}\widehat{\gamma }_{3l}W_{3l}+I_{n}\right) ^{-1} \\
		&&\times \left( I_{n}-\sum_{l=1}^{m_{2}}\gamma
		_{2l}W_{2l}+\sum_{l=1}^{m_{2}}(\gamma _{2l}-\widehat{\gamma }%
		_{2l})W_{2l}\right) \left( S(\lambda )y-\theta
		(x)+\sum_{k=1}^{m_{1}}(\lambda _{k}-\widehat{\lambda }_{k})W_{1k}y+\theta
		(x)-\psi ^{\prime }\widehat{\beta }\right) ,
	\end{eqnarray*}%
	which can be written as
	\begin{equation*}
		\widetilde{\mathbf{\xi }}=\xi +\sum_{j=1}^{r}\zeta
		_{1n,j}p_{nj}+\sum_{j=1}^{s}\zeta _{2n,j}Q_{nj}\xi ,
	\end{equation*}%
	where $p_{nj}$ is an $n$-dimensional vector with bounded elements, $%
	Q_{nj}=[q_{nj,i}]$ is an $n\times n$ matrix with bounded row and column sum
	norms, and $\zeta _{1n,j}$ and $\zeta _{2n,j}$'s are equal to $l_{n}^{\prime
	}\xi /n$, elements of $\lambda _{k}-\widehat{\lambda }_{k}$, $\gamma _{2l}-%
	\widehat{\gamma }_{2l}$, $\theta (x)-\psi ^{\prime }\widehat{\beta }$ or
	their products. This differs from the proof of Lemma 2 in \cite{Jin2015}
	in the term $\theta (x)-\psi ^{\prime }\widehat{\beta }$ and potentially
	increasing order of $d_{\gamma }$. Then, $\zeta _{1n,j}=O_{p}(\sqrt{p^{1/2}/n}%
	\vee \sqrt{d_{\gamma }/n})$ and $\zeta _{2n,j}=O_{p}(\sqrt{p^{1/2}/n}\vee
	\sqrt{d_{\gamma }/n})$, instead of $O_{p}(\sqrt{1/n})$ as in \cite{Jin2015}. Based on this result, the assumptions in Theorem 4 of Su and Qu (2017) hold, so
	the validility of our bootstrap test directly follows.

\end{proof}
\newpage
\begin{table}
	\centering{
		\begin{tabular}{cccccccccccc}
			\hline \hline
			&  & \textbf{PS} &  &  &  & \textbf{Trig} &  &  &  & \textbf{B-s} &
			\\
			& {\small 0.01} & {\small 0.05} & {\small 0.10} &  & {\small 0.01} & {\small %
				0.05} & {\small 0.10} &  & {\small 0.01} & {\small 0.05} & {\small 0.10} \\
			\hline
			{\small $n=60$} &  &  &  &  &  &  &  &  &  &  &  \\ \hline
			${\small c=0}$ & ${\small 0.01}$ & ${\small 0.032}$ & ${\small 0.05}$ &  & $%
			{\small 0.01}$ & ${\small 0.028}$ & ${\small 0.054}$ &  & ${\small 0.02}$ & $%
			{\small 0.042}$ & ${\small 0.064}$ \\
			& ${\small 0.036}$ & ${\small 0.084}$ & ${\small 0.122}$ &  & ${\small 0.02}$
			& ${\small 0.056}$ & ${\small 0.084}$ &  & ${\small 0.044}$ & ${\small 0.008}
			$ & ${\small 0.11}$ \\
			${\small c=3}$ & ${\small 0.07}$ & ${\small 0.156}$ & ${\small 0.194}$ &  & $%
			{\small 0.166}$ & ${\small 0.248}$ & ${\small 0.296}$ &  & ${\small 0.208}$
			& ${\small 0.302}$ & ${\small 0.372}$ \\
			& ${\small 0.454}$ & ${\small 0.58}$ & ${\small 0.658}$ &  & ${\small 0.172}$
			& ${\small 0.29}$ & ${\small 0.358}$ &  & ${\small 0.166}$ & ${\small 0.274}$
			& ${\small 0.346}$ \\
			${\small c=6}$ & ${\small 0.37}$ & ${\small 0.532}$ & ${\small 0.644}$ &  & $%
			{\small 0.688}$ & ${\small 0.806}$ & ${\small 0.854}$ &  & ${\small 0.688}$
			& ${\small 0.82}$ & ${\small 0.884}$ \\
			& ${\small 0.998}$ & ${\small 1}$ & ${\small 1}$ &  & ${\small 0.676}$ & $%
			{\small 0.822}$ & ${\small 0.866}$ &  & ${\small 0.576}$ & ${\small 0.726}$
			& ${\small 0.81}$ \\ \hline
			{\small $n=100$} &  &  &  &  &  &  &  &  &  &  &  \\ \hline
			${\small c=0}$ & ${\small 0.008}$ & ${\small 0.03}$ & ${\small 0.044}$ &  & $%
			{\small 0.006}$ & ${\small 0.012}$ & ${\small 0.028}$ &  & ${\small 0.016}$
			& ${\small 0.028}$ & ${\small 0.042}$ \\
			& ${\small 0.022}$ & ${\small 0.052}$ & ${\small 0.068}$ &  & ${\small 0.004}
			$ & ${\small 0.028}$ & ${\small 0.05}$ &  & ${\small 0.018}$ & ${\small 0.048%
			}$ & ${\small 0.062}$ \\
			${\small c=3}$ & ${\small 0.352}$ & ${\small 0.478}$ & ${\small 0.574}$ &  &
			${\small 0.27}$ & ${\small 0.39}$ & ${\small 0.484}$ &  & ${\small 0.376}$ &
			${\small 0.518}$ & ${\small 0.614}$ \\
			& ${\small 0.54}$ & ${\small 0.666}$ & ${\small 0.744}$ &  & ${\small 0.288}$
			& ${\small 0.412}$ & ${\small 0.508}$ &  & ${\small 0.316}$ & ${\small 0.462}
			$ & ${\small 0.544}$ \\
			${\small c=6}$ & ${\small 0.984}$ & ${\small 0.99}$ & ${\small 0.99}$ &  & $%
			{\small 0.956}$ & ${\small 0.986}$ & ${\small 0.992}$ &  & ${\small 0.98}$ &
			${\small 0.992}$ & ${\small 0.994}$ \\
			& ${\small 0.998}$ & ${\small 0.998}$ & ${\small 0.998}$ &  & ${\small 0.948}
			$ & ${\small 0.99}$ & ${\small 0.992}$ &  & ${\small 0.956}$ & ${\small 0.99}
			$ & ${\small 0.996}$ \\ \hline
			{\small $n=200$} &  &  &  &  &  &  &  &  &  &  &  \\ \hline
			${\small c=0}$ & ${\small 0.002}$ & ${\small 0.016}$ & ${\small 0.034}$ &  &
			${\small 0.002}$ & ${\small 0.014}$ & ${\small 0.034}$ &  & ${\small 0.038}$
			& ${\small 0.074}$ & ${\small 0.102}$ \\
			& ${\small 0.008}$ & ${\small 0.026}$ & ${\small 0.048}$ &  & ${\small 0.012}
			$ & ${\small 0.028}$ & ${\small 0.036}$ &  & ${\small 0.01}$ & ${\small 0.036%
			}$ & ${\small 0.074}$ \\
			${\small c=3}$ & ${\small 0.176}$ & ${\small 0.29}$ & ${\small 0.356}$ &  & $%
			{\small 0.164}$ & ${\small 0.256}$ & ${\small 0.312}$ &  & ${\small 0.388}$
			& ${\small 0.354}$ & ${\small 0.606}$ \\
			& ${\small 0.34}$ & ${\small 0.496}$ & ${\small 0.582}$ &  & ${\small 0.144}$
			& ${\small 0.274}$ & ${\small 0.356}$ &  & ${\small 0.168}$ & ${\small 0.282}
			$ & ${\small 0.376}$ \\
			${\small c=6}$ & ${\small 0.888}$ & ${\small 0.942}$ & ${\small 0.96}$ &  & $%
			{\small 0.818}$ & ${\small 0.898}$ & ${\small 0.934}$ &  & ${\small 0.944}$
			& ${\small 0.974}$ & ${\small 0.986}$ \\
			& ${\small 0.99}$ & ${\small 0.998}$ & ${\small 1}$ &  & ${\small 0.816}$ & $%
			{\small 0.904}$ & ${\small 0.944}$ &  & ${\small 0.862}$ & ${\small 0.932}$
			& ${\small 0.954}$ \\ \hline \hline
		\end{tabular}
		\caption{Rejection probabilities of {\small SARARMA(0,1,0)} using
			asymptotic test ${\fancyt}$ at 1, 5, 10\% levels, power series (\textbf{PS}), trigonometric (\textbf{Trig}) and B-spline (\textbf{B-s}) bases. Compactly
			supported regressors.}\label{table:newsimsapp1}}
\end{table}
\newpage
\begin{table}
	\centering{
		\begin{tabular}{cccccccccccc}
			\hline \hline
			&  & \textbf{PS} &  &  &  & \textbf{Trig} &  &  &  & \textbf{B-s} &
			\\
			& {\small 0.01} & {\small 0.05} & {\small 0.10} &  & {\small 0.01} & {\small %
				0.05} & {\small 0.10} &  & {\small 0.01} & {\small 0.05} & {\small 0.10} \\
			\hline
			{\small $n=60$} &  &  &  &  &  &  &  &  &  &  &  \\ \hline
			${\small c=0}$ & ${\small 0.01}$ & ${\small 0.032}$ & ${\small 0.05}$ &  & $%
			{\small 0.01}$ & ${\small 0.028}$ & ${\small 0.054}$ &  & ${\small 0.06}$ & $%
			{\small 0.01}$ & ${\small 0.016}$ \\
			& ${\small 0.036}$ & ${\small 0.084}$ & ${\small 0.122}$ &  & ${\small 0.02}$
			& ${\small 0.056}$ & ${\small 0.084}$ &  & ${\small 0.044}$ & ${\small 0.008}
			$ & ${\small 0.116}$ \\
			${\small c=3}$ & ${\small 0.07}$ & ${\small 0.156}$ & ${\small 0.194}$ &  & $%
			{\small 0.16}$ & ${\small 0.252}$ & ${\small 0.292}$ &  & ${\small 0.09}$ & $%
			{\small 0.138}$ & ${\small 0.186}$ \\
			& ${\small 0.454}$ & ${\small 0.58}$ & ${\small 0.658}$ &  & ${\small 0.174}$
			& ${\small 0.29}$ & ${\small 0.358}$ &  & ${\small 0.166}$ & ${\small 0.272}$
			& ${\small 0.34}$ \\
			${\small c=6}$ & ${\small 0.37}$ & ${\small 0.532}$ & ${\small 0.644}$ &  & $%
			{\small 0.682}$ & ${\small 0.798}$ & ${\small 0.85}$ &  & ${\small 0.514}$ &
			${\small 0.644}$ & ${\small 0.714}$ \\
			& ${\small 0.998}$ & ${\small 1}$ & ${\small 1}$ &  & ${\small 0.676}$ & $%
			{\small 0.822}$ & ${\small 0.866}$ &  & ${\small 0.572}$ & ${\small 0.714}$
			& ${\small 0.8}$ \\ \hline
			{\small $n=100$} &  &  &  &  &  &  &  &  &  &  &  \\ \hline
			${\small c=0}$ & ${\small 0.008}$ & ${\small 0.03}$ & ${\small 0.044}$ &  & $%
			{\small 0.006}$ & ${\small 0.012}$ & ${\small 0.026}$ &  & ${\small 0}$ & $%
			{\small 0.004}$ & ${\small 0.006}$ \\
			& ${\small 0.022}$ & ${\small 0.052}$ & ${\small 0.068}$ &  & ${\small 0.006}
			$ & ${\small 0.028}$ & ${\small 0.05}$ &  & ${\small 0.018}$ & ${\small 0.05}
			$ & ${\small 0.062}$ \\
			${\small c=3}$ & ${\small 0.352}$ & ${\small 0.478}$ & ${\small 0.574}$ &  &
			${\small 0.268}$ & ${\small 0.396}$ & ${\small 0.486}$ &  & ${\small 0.158}$
			& ${\small 0.23}$ & ${\small 0.288}$ \\
			& ${\small 0.54}$ & ${\small 0.666}$ & ${\small 0.744}$ &  & ${\small 0.288}$
			& ${\small 0.412}$ & ${\small 0.508}$ &  & ${\small 0.322}$ & ${\small 0.466}
			$ & ${\small 0.55}$ \\
			${\small c=6}$ & ${\small 0.984}$ & ${\small 0.99}$ & ${\small 0.99}$ &  & $%
			{\small 0.958}$ & ${\small 0.986}$ & ${\small 0.992}$ &  & ${\small 0.918}$
			& ${\small 0.97}$ & ${\small 0.98}$ \\
			& ${\small 0.998}$ & ${\small 0.998}$ & ${\small 0.998}$ &  & ${\small 0.952}
			$ & ${\small 0.99}$ & ${\small 0.992}$ &  & ${\small 0.96}$ & ${\small 0.99}$
			& ${\small 0.998}$ \\ \hline
			{\small $n=200$} &  &  &  &  &  &  &  &  &  &  &  \\ \hline
			${\small c=0}$ & ${\small 0.002}$ & ${\small 0.016}$ & ${\small 0.034}$ &  &
			${\small 0.002}$ & ${\small 0.018}$ & ${\small 0.038}$ &  & ${\small 0}$ & $%
			{\small 0}$ & ${\small 0}$ \\
			& ${\small 0.008}$ & ${\small 0.026}$ & ${\small 0.048}$ &  & ${\small 0.012}
			$ & ${\small 0.028}$ & ${\small 0.032}$ &  & ${\small 0.01}$ & ${\small 0.036%
			}$ & ${\small 0.064}$ \\
			${\small c=3}$ & ${\small 0.176}$ & ${\small 0.29}$ & ${\small 0.356}$ &  & $%
			{\small 0.156}$ & ${\small 0.258}$ & ${\small 0.312}$ &  & ${\small 0.022}$
			& ${\small 0.03}$ & ${\small 0.044}$ \\
			& ${\small 0.34}$ & ${\small 0.496}$ & ${\small 0.582}$ &  & ${\small 0.144}$
			& ${\small 0.272}$ & ${\small 0.352}$ &  & ${\small 0.154}$ & ${\small 0.266}
			$ & ${\small 0.352}$ \\
			${\small c=6}$ & ${\small 0.888}$ & ${\small 0.942}$ & ${\small 0.96}$ &  & $%
			{\small 0.816}$ & ${\small 0.908}$ & ${\small 0.936}$ &  & ${\small 0.43}$ &
			${\small 0.522}$ & ${\small 0.554}$ \\
			& ${\small 0.99}$ & ${\small 0.998}$ & ${\small 1}$ &  & ${\small 0.816}$ & $%
			{\small 0.904}$ & ${\small 0.944}$ &  & ${\small 0.856}$ & ${\small 0.924}$
			& ${\small 0.944}$ \\ \hline \hline
		\end{tabular}
		\caption{Rejection probabilities of {\small SARARMA(0,1,0)} using
			asymptotic test ${\fancyt}^a$ at 1, 5, 10\% levels, power series (\textbf{PS}), trigonometric (\textbf{Trig}) and B-spline (\textbf{B-s}) bases. Compactly
			supported regressors.}\label{table:newsimsapp2}}
\end{table}
\newpage
\begin{table}
	\centering{
		\begin{tabular}{cccccccccccc}
			\hline \hline
			&  & \textbf{PS} & ${\fancyt}=\fancyt^{a}$ &  &  &
			\textbf{Trig} & $\fancyt$ &  &  & \textbf{Trig} & $\fancyt^{a}$
			\\
			& {\small 0.01} & {\small 0.05} & {\small 0.10} &  & {\small 0.01} & {\small %
				0.05} & {\small 0.10} &  & {\small 0.01} & {\small 0.05} & {\small 0.10} \\
			\hline
			{\small $n=60$} &  &  &  &  &  &  &  &  &  &  &  \\ \hline
			${\small c=0}$ & ${\small 0.02}$ & ${\small 0.05}$ & ${\small 0.072}$ &  & $%
			{\small 0.016}$ & ${\small 0.038}$ & ${\small 0.052}$ &  & ${\small 0.016}$
			& ${\small 0.038}$ & ${\small 0.052}$ \\
			& ${\small 0.038}$ & ${\small 0.082}$ & ${\small 0.11}$ &  & ${\small 0.038}$
			& ${\small 0.06}$ & ${\small 0.08}$ &  & ${\small 0.038}$ & ${\small 0.06}$
			& ${\small 0.08}$ \\
			${\small c=3}$ & ${\small 0.106}$ & ${\small 0.158}$ & ${\small 0.224}$ &  &
			${\small 0.062}$ & ${\small 0.11}$ & ${\small 0.146}$ &  & ${\small 0.062}$
			& ${\small 0.11}$ & ${\small 0.146}$ \\
			& ${\small 0.152}$ & ${\small 0.25}$ & ${\small 0.31}$ &  & ${\small 0.09}$
			& ${\small 0.158}$ & ${\small 0.204}$ &  & ${\small 0.09}$ & ${\small 0.158}$
			& ${\small 0.204}$ \\
			${\small c=6}$ & ${\small 0.552}$ & ${\small 0.686}$ & ${\small 0.73}$ &  & $%
			{\small 0.234}$ & ${\small 0.352}$ & ${\small 0.482}$ &  & ${\small 0.236}$
			& ${\small 0.354}$ & ${\small 0.43}$ \\
			& ${\small 0.634}$ & ${\small 0.774}$ & ${\small 0.82}$ &  & ${\small 0.404}$
			& ${\small 0.542}$ & ${\small 0.642}$ &  & ${\small 0.404}$ & ${\small 0.542}
			$ & ${\small 0.642}$ \\ \hline
			{\small $n=100$} &  &  &  &  &  &  &  &  &  &  &  \\ \hline
			${\small c=0}$ & ${\small 0.008}$ & ${\small 0.024}$ & ${\small 0.036}$ &  &
			${\small 0.002}$ & ${\small 0.018}$ & ${\small 0.036}$ &  & ${\small 0.002}$
			& ${\small 0.018}$ & ${\small 0.036}$ \\
			& ${\small 0.024}$ & ${\small 0.05}$ & ${\small 0.068}$ &  & ${\small 0.012}$
			& ${\small 0.026}$ & ${\small 0.052}$ &  & ${\small 0.012}$ & ${\small 0.026}
			$ & ${\small 0.052}$ \\
			${\small c=3}$ & ${\small 0.162}$ & ${\small 0.262}$ & ${\small 0.342}$ &  &
			${\small 0.142}$ & ${\small 0.22}$ & ${\small 0.286}$ &  & ${\small 0.142}$
			& ${\small 0.22}$ & ${\small 0.286}$ \\
			& ${\small 0.216}$ & ${\small 0.332}$ & ${\small 0.408}$ &  & ${\small 0.164}
			$ & ${\small 0.274}$ & ${\small 0.35}$ &  & ${\small 0.164}$ & ${\small 0.274%
			}$ & ${\small 0.35}$ \\
			${\small c=6}$ & ${\small 0.824}$ & ${\small 0.894}$ & ${\small 0.926}$ &  &
			${\small 0.79}$ & ${\small 0.868}$ & ${\small 0.892}$ &  & ${\small 0.79}$ &
			${\small 0.866}$ & ${\small 0.894}$ \\
			& ${\small 0.888}$ & ${\small 0.944}$ & ${\small 0.952}$ &  & ${\small 0.862}
			$ & ${\small 0.896}$ & ${\small 0.928}$ &  & ${\small 0.862}$ & ${\small %
				0.896}$ & ${\small 0.928}$ \\ \hline
			{\small $n=200$} &  &  &  &  &  &  &  &  &  &  &  \\ \hline
			${\small c=0}$ & ${\small 0.006}$ & ${\small 0.018}$ & ${\small 0.032}$ &  &
			${\small 0.008}$ & ${\small 0.022}$ & ${\small 0.032}$ &  & ${\small 0.008}$
			& ${\small 0.022}$ & ${\small 0.032}$ \\
			& ${\small 0.012}$ & ${\small 0.032}$ & ${\small 0.068}$ &  & ${\small 0.01}$
			& ${\small 0.026}$ & ${\small 0.046}$ &  & ${\small 0.01}$ & ${\small 0.026}$
			& ${\small 0.046}$ \\
			${\small c=3}$ & ${\small 0.096}$ & ${\small 0.182}$ & ${\small 0.258}$ &  &
			${\small 0.076}$ & ${\small 0.152}$ & ${\small 0.212}$ &  & ${\small 0.078}$
			& ${\small 0.15}$ & ${\small 0.208}$ \\
			& ${\small 0.126}$ & ${\small 0.24}$ & ${\small 0.33}$ &  & ${\small 0.098}$
			& ${\small 0.184}$ & ${\small 0.26}$ &  & ${\small 0.098}$ & ${\small 0.184}$
			& ${\small 0.26}$ \\
			${\small c=6}$ & ${\small 0.754}$ & ${\small 0.858}$ & ${\small 0.892}$ &  &
			${\small 0.596}$ & ${\small 0.728}$ & ${\small 0.794}$ &  & ${\small 0.596}$
			& ${\small 0.724}$ & ${\small 0.79}$ \\
			& ${\small 0.84}$ & ${\small 0.918}$ & ${\small 0.944}$ &  & ${\small 0.684}$
			& ${\small 0.794}$ & ${\small 0.866}$ &  & ${\small 0.684}$ & ${\small 0.792}
			$ & ${\small 0.866}$ \\ \hline \hline
		\end{tabular}
		\caption{Rejection probabilities of {\small SARARMA(0,1,0)} using
			asymptotic tests ${\fancyt},{\fancyt}^a$ at 1, 5, 10\% levels, power series (\textbf{PS}) and trigonometric (\textbf{Trig}) bases.  Unboundedly
			supported regressors.}\label{table:newsimsapp3}}
\end{table}
\newpage
\begin{table}
	\centering{
		\begin{tabular}{cccccccccccc}
			\hline \hline
			&  & \textbf{PS} & ${\fancyt}^\ast=\fancyt^{a \ast}$ &  &  &
			\textbf{Trig} & $\fancyt^{ \ast}$ &  &  & \textbf{Trig} & $\fancyt^{a \ast}$
			\\
			& {\small 0.01} & {\small 0.05} & {\small 0.10} &  & {\small 0.01} & {\small %
				0.05} & {\small 0.10} &  & {\small 0.01} & {\small 0.05} & {\small 0.10} \\
			\hline
			{\small $n=60$} &  &  &  &  &  &  &  &  &  &  &  \\ \hline
			${\small c=0}$ & ${\small 0.008}$ & ${\small 0.058}$ & ${\small 0.108}$ &  &
			${\small 0.01}$ & ${\small 0.046}$ & ${\small 0.124}$ &  & ${\small 0.01}$ &
			${\small 0.046}$ & ${\small 0.124}$ \\
			& ${\small 0.008}$ & ${\small 0.042}$ & ${\small 0.094}$ &  & ${\small 0.006}
			$ & ${\small 0.044}$ & ${\small 0.102}$ &  & ${\small 0.006}$ & ${\small %
				0.044}$ & ${\small 0.102}$ \\
			${\small c=3}$ & ${\small 0.052}$ & ${\small 0.17}$ & ${\small 0.318}$ &  & $%
			{\small 0.036}$ & ${\small 0.14}$ & ${\small 0.21}$ &  & ${\small 0.036}$ & $%
			{\small 0.14}$ & ${\small 0.21}$ \\
			& ${\small 0.034}$ & ${\small 0.16}$ & ${\small 0.184}$ &  & ${\small 0.034}$
			& ${\small 0.132}$ & ${\small 0.234}$ &  & ${\small 0.034}$ & ${\small 0.132}
			$ & ${\small 0.234}$ \\
			${\small c=6}$ & ${\small 0.35}$ & ${\small 0.67}$ & ${\small 0.808}$ &  & $%
			{\small 0.16}$ & ${\small 0.392}$ & ${\small 0.556}$ &  & ${\small 0.16}$ & $%
			{\small 0.392}$ & ${\small 0.558}$ \\
			& ${\small 0.262}$ & ${\small 0.656}$ & ${\small 0.794}$ &  & ${\small 0.204}
			$ & ${\small 0.468}$ & ${\small 0.66}$ &  & ${\small 0.204}$ & ${\small 0.468%
			}$ & ${\small 0.66}$ \\ \hline
			{\small $n=100$} &  &  &  &  &  &  &  &  &  &  &  \\ \hline
			${\small c=0}$ & ${\small 0.006}$ & ${\small 0.05}$ & ${\small 0.102}$ &  & $%
			{\small 0.006}$ & ${\small 0.05}$ & ${\small 0.11}$ &  & ${\small 0.004}$ & $%
			{\small 0.05}$ & ${\small 0.112}$ \\
			& ${\small 0.012}$ & ${\small 0.054}$ & ${\small 0.128}$ &  & ${\small 0.004}
			$ & ${\small 0.044}$ & ${\small 0.112}$ &  & ${\small 0.004}$ & ${\small %
				0.044}$ & ${\small 0.112}$ \\
			${\small c=3}$ & ${\small 0.13}$ & ${\small 0.342}$ & ${\small 0.516}$ &  & $%
			{\small 0.128}$ & ${\small 0.324}$ & ${\small 0.488}$ &  & ${\small 0.126}$
			& ${\small 0.32}$ & ${\small 0.488}$ \\
			& ${\small 0.122}$ & ${\small 0.326}$ & ${\small 0.498}$ &  & ${\small 0.114}
			$ & ${\small 0.298}$ & ${\small 0.474}$ &  & ${\small 0.114}$ & ${\small %
				0.298}$ & ${\small 0.474}$ \\
			${\small c=6}$ & ${\small 0.766}$ & ${\small 0.932}$ & ${\small 0.974}$ &  &
			${\small 0.728}$ & ${\small 0.92}$ & ${\small 0.974}$ &  & ${\small 0.728}$
			& ${\small 0.92}$ & ${\small 0.972}$ \\
			& ${\small 0.774}$ & ${\small 0.934}$ & ${\small 0.968}$ &  & ${\small 0.732}
			$ & ${\small 0.898}$ & ${\small 0.952}$ &  & ${\small 0.732}$ & ${\small %
				0.898}$ & ${\small 0.952}$ \\ \hline
			{\small $n=200$} &  &  &  &  &  &  &  &  &  &  &  \\ \hline
			${\small c=0}$ & ${\small 0.03}$ & ${\small 0.056}$ & ${\small 0.088}$ &  & $%
			{\small 0.028}$ & ${\small 0.06}$ & ${\small 0.098}$ &  & ${\small 0.028}$ &
			${\small 0.06}$ & ${\small 0.098}$ \\
			& ${\small 0.028}$ & ${\small 0.084}$ & ${\small 0.128}$ &  & ${\small 0.022}
			$ & ${\small 0.068}$ & ${\small 0.118}$ &  & ${\small 0.022}$ & ${\small %
				0.068}$ & ${\small 0.118}$ \\
			${\small c=3}$ & ${\small 0.17}$ & ${\small 0.346}$ & ${\small 0.49}$ &  & $%
			{\small 0.132}$ & ${\small 0.286}$ & ${\small 0.384}$ &  & ${\small 0.13}$ &
			${\small 0.288}$ & ${\small 0.38}$ \\
			& ${\small 0.178}$ & ${\small 0.34}$ & ${\small 0.488}$ &  & ${\small 0.128}$
			& ${\small 0.274}$ & ${\small 0.416}$ &  & ${\small 0.128}$ & ${\small 0.274}
			$ & ${\small 0.416}$ \\
			${\small c=6}$ & ${\small 0.794}$ & ${\small 0.92}$ & ${\small 0.966}$ &  & $%
			{\small 0.682}$ & ${\small 0.866}$ & ${\small 0.93}$ &  & ${\small 0.678}$ &
			${\small 0.864}$ & ${\small 0.93}$ \\
			& ${\small 0.84}$ & ${\small 0.936}$ & ${\small 0.976}$ &  & ${\small 0.698}$
			& ${\small 0.888}$ & ${\small 0.93}$ &  & ${\small 0.698}$ & ${\small 0.888}$
			& ${\small 0.93}$ \\ \hline \hline
		\end{tabular}	
		
		\caption{Rejection probabilities of {\small SARARMA(0,1,0)} using
			bootstrap tests ${\fancyt}^{\ast}, {\fancyt}^{a \ast}$ at 1, 5, 10\% levels, power series (\textbf{PS}) and trigonometric (\textbf{Trig}) bases. Unboundedly
			supported regressors.}\label{table:newsimsapp4}}
\end{table}

\bibliographystyle{chicago}
{\bibliography{thesisb}}

\end{document}